\documentclass[journal]{IEEEtran}
\usepackage{amsfonts,amssymb}
\usepackage{hyperref}
\usepackage{titlesec}
\usepackage{titletoc}
\usepackage{amsmath}
\numberwithin{equation}{section}
\usepackage{listings}
\usepackage{verbatim}
\usepackage{graphicx}
\usepackage{mathrsfs}
\usepackage{amsthm}
\usepackage{bbm}
\usepackage{bm}
\usepackage{subcaption}

\newtheorem{definition}{Definition}
\usepackage{xcolor, tikz}
\usetikzlibrary{positioning}
\usetikzlibrary{matrix}
\newtheorem{theorem}{Theorem}

\newtheorem{lem}{Lemma}
\newtheorem{fact}{Fact}
\newtheorem{coro}{Corollary}
\newtheorem{rem}{Remark}
\theoremstyle{definition}

\DeclareMathOperator*{\sign}{\mathrm{sign}}

\DeclareMathOperator{\ii}{\textbf{i}}
\DeclareMathOperator{\rank}{rank}
\DeclareMathOperator{\diag}{diag}

\usepackage[autostyle=false, style=english]{csquotes}
\MakeOuterQuote{"}
\ifCLASSINFOpdf
\else
\fi
\begin{document}

%
\title{Uniform Exact Reconstruction of Sparse Signals and Low-rank Matrices from Phase-Only Measurements}
%
%
%

\author{Junren~Chen, and \and Michael K. Ng,~\IEEEmembership{Senior~Member,~IEEE}
\thanks{J. Chen is with Department of Mathematics, The University of Hong Kong (e-mail: chenjr58@connect.hku.hk). The work of J. Chen was supported by a Hong Kong PhD Fellowship from Hong Kong Research Grant Council. M. K. Ng is with Department of Mathematics, The University of Hong Kong (e-mail: mng@maths.hku.hk). The work of M. K. Ng was supported by Hong Kong Research Grant Council GRF 12300218, 12300519,
17201020, 17300021, C1013-21GF, C7004-21GF and jointly by NSFC-RGC N-HKU76921. {\it (Corresponding author: Junren Chen.)}}
}

%
%

\markboth{Accepted at T-IT}%
{Shell \MakeLowercase{\textit{et al.}}: Bare Demo of IEEEtran.cls for IEEE Journals}
%



\maketitle

\begin{abstract}
 In phase-only compressive sensing (PO-CS), our goal is to recover low-complexity signals (e.g., sparse signals, low-rank matrices) from the phase of complex linear measurements.
	 While perfect recovery of signal direction in PO-CS was observed quite early, the exact reconstruction guarantee for a fixed, real signal was recently 
	done by 
	Jacques and Feuillen [IEEE Trans. Inf. Theory, 67 (2021), pp.~4150\textrm{--}4161]. However, two   questions remain open:
the uniform recovery guarantee and exact recovery of complex signal.
In this paper, we almost completely address these two open questions. 
We prove that, all complex sparse signals or low-rank matrices can be uniformly, exactly recovered from a near optimal number of complex Gaussian measurement phases. 
      By recasting PO-CS as a linear compressive sensing problem, the exact recovery follows from  restricted isometry property (RIP). Our approach to uniform recovery guarantee is based on covering arguments that  involve  a delicate control of the (original linear) measurements with overly small magnitude. To work with complex signal,  a different sign-product embedding property and a careful rescaling of the sensing matrix are employed. 
      In addition, we show an extension that 
		the uniform recovery is stable under moderate bounded noise. We also propose to add Gaussian dither before capturing the phases to   achieve full reconstruction with norm information. Experimental results are   reported to corroborate and demonstrate our theoretical results.  
\end{abstract}


%
\IEEEpeerreviewmaketitle

\section{Introduction}
	Signal reconstruction from the phase of Fourier measurement was intensively studied about 40 years ago, see several pioneering works \cite{hayes1980signal,hayes1982reconstruction,oppenheim1980iterative,oppenheim1981importance}. Theoretically, these works showed almost all signals can be exactly reconstructed from Fourier phase, up to the trivial ambiguity of a positive scaling factor, thus confirming the importance of phase in Fourier transform \cite{oppenheim1981importance}. Algorithms developed in   early works include closed form solution \cite{hayes1980signal},   iterative algorithm \cite{hayes1980signal,quatieri1981iterative} and the Projection Onto Convex Sets (POCS) algorithm \cite{levi1983signal,urieli1998optimal}. Such considerable interest from signal processing community was mainly due to some early motivations such as blind deconvolution \cite{hayes1982reconstruction,stockham1975blind}, signal coding \cite{hayes1980signal,oppenheim1981importance}, Kinoforms \cite{espy1983effects}, while subsequent   applications extended to image restoration \cite{behar1992image,urieli1998optimal} and inpainting \cite{hua2007image}, object shape retrieval \cite{bartolini2005warp} and speech reconstruction \cite{loveimi2010objective}. Besides, generalized phase-only reconstruction, where the linear measurement is not restricted to some specific transform like Fourier transform, was studied in our recent work \cite{chen2022signal}. We established more practical uniqueness conditions that are necessary and sufficient, as well as proved some new results on minimal measurement number for uniform recovery of (almost) all signals. 

        While all   papers reviewed above consider the phase-only reconstruction of unstructured signal,   there is a line of research concerning recovery of  sparse (or more generally, low-complexity) signal from phase \cite{boufounos2013angle,boufounos2013sparse,feuillen2020ell,jacques2021importance}, which is termed as a phase-only compressive sensing (PO-CS) problem. More precisely, given the complex sensing matrix $\bm{\Phi} $, one aims to recover or estimate a sparse signal $\bm{x}$ from the phase of $\bm{\Phi x}$. The motivations   of   PO-CS are at least twofold. From a theoretical side, since the (complex) phase of a real number is just its sign ($1$ or $-1$), PO-CS is a natural generalization of 1-bit compressive sensing (1-bit CS), which is   a   well-studied nonlinear compressive sensing model and can be formulated as the estimation of a sparse real signal $\bm{x}$ from the sign of $\bm{\Phi x}$ under a real sensing matrix $\bm{\Phi}$ \cite{boufounos20081,jacques2013robust,plan2012robust,plan2013one}. Practically, such phase-only sensing scenario can be more stable under large measurement variations or corruption \cite{jacques2021importance}. Moreover, it also allows easier  data quantization due to the compactness of the measurement range $\{a\in \mathbb{C}:|a|=1\}$. \textcolor{black}{For instance, phase-only measurements can be quantized to finite bit with a simple uniform quantizer (e.g., \cite{boufounos2013angle,wang2016multiuser}). By contrast, to uniformly quantize  regular real or complex measurements to a specific number of bits, in general we still need to precisely estimate the measurement magnitude to avoid the overload issue (e.g., see     \cite[Equation (9)]{gray1993dithered}).}  

       To generalize 1-bit CS to complex sensing matrix, Boufounos first proposed and  studied  PO-CS \cite{boufounos2013angle,boufounos2013sparse}. Motivated by  restricted isometry property (RIP) in linear compressive sensing \cite{candes2005decoding,candes2006near} and the binary $\epsilon$-stable embedding (B$\epsilon$SE) in 1-bit CS \cite{jacques2013robust},   he established a kind of angle-preserving property that indicates the possibility of approximately recovering the signal. By recasting PO-CS as a linear compressive sensing problem, however, perfect reconstruction of signal direction beyond the theory was observed \cite{boufounos2013sparse}. \textcolor{black}{Since then, the  experimental exact reconstruction in PO-CS, which is not possible in 1-bit CS, remained theoretically unjustified.} In the past few years,  Jacques and his collaborators revisited the PO-CS problem. In \cite{feuillen2020ell}, Feuillen et al. proposed a non-iterative   approach called Projected Back Projection (PBP) to estimate a   sparse signal that may be complex-valued. Provably, the estimation error of PBP  decays with rate $O\big(\sqrt[\leftroot{-3}\uproot{3}4]{\frac{s}{m}}\big)$ under larger sample size $m$. Nevertheless, such result by no means leads to exact reconstruction. Indeed, even the possibility of exact reconstruction, i.e., the identifiability question, can not be confirmed by \cite{boufounos2013angle,boufounos2013sparse,feuillen2020ell}.

       The theoretical breakthrough was  recently made by Jacques and Feuillen   \cite{jacques2021importance}. Specifically, they considered   $\bm{x}\in\mathcal{K}$ for some low-complexity set $\mathcal{K}$ and  recast PO-CS as a linear compressive sensing problem (this idea goes back to \cite{boufounos2013sparse}). If $\bm{\Phi}$ has  i.i.d. complex Gaussian entries drawn from $\mathcal{N}(0,1)+\mathcal{N}(0,1)\ii$, they showed the sensing matrix of the resulting linear compressive sensing problem 
       respects RIP under a sample size proportional to the intrinsic dimension of $\mathcal{K}$ (characterized by Gaussian width). Therefore, exact signal reconstruction can be achieved by some tractable reconstruction procedure (that works under RIP) developed in the compressive sensing literature.

       At a higher level, the result in \cite{jacques2021importance} demonstrates the importance of phase   in (complex) compressive sensing --- as the  reconstruction procedure can be implemented without aid from any magnitude information. It was also reported as numerical result that PO-CS exhibits similar performance as classical linear compressive sensing at about twice the measurement number, which is, interestingly, aligned with the  minimal measurement number   for recovering almost all unstructured signals  presented in   \cite[Thm. 4.1]{chen2022signal}.

       However, unlike in classical linear compressive sensing where the RIP of sensing matrix   delivers uniform recovery guarantee for all sparse signals (e.g., \cite{Foucart2013,candes2005decoding}), \cite{jacques2021importance} only establishes  non-uniform guarantee for reconstruction of a fixed signal.\footnote{In a uniform recovery guarantee, a single measurement matrix simultaneously ensures the   recovery of all signals of interest. By contrast, in a non-uniform guarantee,   the measurement matrix only works for a fixed signal, and in general a new measurement matrix should be drawn for recovery of another signal.} In short, this is because in the corresponding linear compressive sensing problem, the sensing matrix varies with the signal $\bm{x}$, while \cite{jacques2021importance} only showed the RIP of a specific/fixed sensing matrix. \textcolor{black}{Because uniformity is regarded as an important feature of recovery guarantee in compressive sensing} \cite{genzel2022unified}, it is of particular interest to study whether a uniform exact recovery guarantee is achievable. In addition, 
       many key ingredients in \cite{jacques2021importance} heavily rely on the assumption of real signal $\bm{x}$ and do not extend to complex $\bm{x}$, so it remains unknown whether it is possible to perfectly recover a sparse complex signal from phase-only measurements.\footnote{Among existing works, only the approximate recovery guarantee in \cite{feuillen2020ell} applies to complex signal.} Therefore, Jacques and Feuillen left these two possible improvements  as open questions for future research, see   \cite[Sec. VII]{jacques2021importance}.

       In the major results of this paper, we show that under near optimal sample complexity (up to logarithmic factor), all complex sparse signals or low-rank matrices can be uniformly and perfectly reconstructed  from phase-only measurements, up to the trivial ambiguity of a positive scaling.\footnote{The signals in $\mathbb{R}_+\bm{x} = \{\lambda\bm{x}:\lambda>0\}$ can never be distinguished, thus we can only hope to recover $\bm{x}$ up to a positive scaling $\lambda$. Throughout this work, the exact reconstruction in PO-CS allows such scaling ambiguity, which may not be explicitly mentioned.} For the most classical sparse and low-rank signal structure, we thus simultaneously address the two open questions in affirmative. Our technical contributions for proving these major results are summarized as follows: 
       
       \begin{itemize}
            
           \item (Uniform Exact Recovery)\textbf{.} Since the sensing matrix in the reformulation varies with the underlying signal (\ref{3.6}), to yield a uniform guarantee we need to prove the sensing matrices for all signals simultaneously respect RIP, which is in essence bounding an empirical process from above (\ref{a3.12}). Due to the additional supremum taken for underlying signals, the empirical process essentially relies on $\sign(\cdot)$ (the function to extract complex phase), thus making the techniques in \cite{jacques2021importance} and some other tools not applicable. Instead, we resort to the more elementary covering arguments, and the main difficulty encountered is the pathological behaviour of     $\sign(\cdot)$ in $\{\eta\cdot a:a\in \mathbb{C},|a|\leq 1\}$ ($\eta>0$ is a pre-specified threshold), e.g., discontinuity and large variation. Thus, retaining only the phase would be problematic for the measurements with small absolute value (smaller than $\eta$), which are termed near vanishing measurements. As it turns out, a rather delicate   and involved analysis is needed to control the effect of near vanishing measurements; \textcolor{black}{for instance}, controlling the number of near vanishing measurements (Lemma \ref{lem10}), estimating the norm of sub-matrices (Lemma \ref{lem11}).

           
           \item (Recovery of Complex Signals)\textbf{.} Compared with real signal   \cite{boufounos2013sparse,jacques2021importance}, PO-CS of complex signal should be reformulated as a (real) linear compressive sensing problem with an extended sensing matrix. In our analysis, a signal is decomposed as a parallel part and an orthogonal part regarding the desired underlying signal. Compared with  \cite[Lem. 5.4]{jacques2021importance}, a different sign-product embedding property that only involves the real part of inner product is established for analyzing the parallel part. This is essentially due to the removal of a redundant measurement (Note that, the linear compressive sensing problem in \cite{jacques2021importance} involves $n+2$ measurements, but only $n+1$ measurements in this work). For orthogonal part, our calculations unveil the existence of a bias term (i.e., $|\Im\big<\bm{x},\bm{u}\big>|^2$ in (\ref{3.31})). To deal with this issue, a rescaling of the sensing matrix is applied to achieve a restricted isometry constant (RIC) of $\frac{1}{3}+\delta$, where $\delta$ can be made sufficiently small. This  RIP circumstance is sufficient for sparse or low-rank recovery \cite{cai2013sparse} .   
       \end{itemize}

       Besides the main results presented in Theorems \ref{theorem1}-\ref{theorem2}, we show that the uniform recovery  in PO-CS is stable under moderate bounded noise (Theorem \ref{theorem3}). Considering that a full reconstruction without scaling ambiguity is preferable in some applications, we propose a simple variant of PO-CS that can achieve uniform reconstruction with norm information, due to a random Gaussian dither added before capturing the phases (Theorem \ref{theorem4}). Beyond that, our discussions reveal that our proof for Theorem \ref{theorem1} actually applies to many other structured signal set with low covering dimension. To complement our uniform reconstruction results, we also present a non-uniform guarantee for any fixed complex signal with unit $\ell_2$-norm  (Theorem \ref{theorem5}), which represents   \textcolor{black}{an extension of the main result in \cite{jacques2021importance} to complex signals}.   Due to a finer tool and more careful analysis, Theorem \ref{theorem5} slightly refines the sample complexity needed in \cite{jacques2021importance}.

       The outline of this paper is given as follows. In Section \ref{section2} we provide the notations and preliminaries. We present the technical proof for the uniform recovery guarantee of complex sparse signals in Section \ref{section3}, which is then extended to low-rank matrices without details in Section \ref{section4}. In Section \ref{section6} we present a uniform stable recovery guarantee and the uniform reconstruction with norm in PO-CS with Gaussian dither. Some discussions on extension and limitation of our main result are provided in Section \ref{section7}, where we also present the non-uniform guarantee for complex signal recovery, with technical proof relegated to Appendix. We report the experimental results for PO-CS in Section \ref{section8} and give some remarks to conclude the paper in Section \ref{section9}.
      
    \section{Notations and Preliminaries}
    \label{section2}
    \subsection{Notations}
        We introduce the generic notations used throughout the paper, while   additional ones are defined in subsequent development when appropriate. Boldface letters  are used to represent   vectors and matrices. We write $[m] = \{1,\cdots, m\}$ for positive integer $m$.  The cardinality of a finite set $\mathcal{T}$ will be denoted by $|\mathcal{T}|$. Writing $\ii = \sqrt{-1}$ as the complex unit, for a complex number $a$ we alternatively use $\Re(a)$ or $a^\Re$ to denote its real part, $\Im(a)$ or $a^\Im$ for its imaginary part, $|a|$, $\bar{a}$ for the absolute value and conjugate, respectively. We let $\sign(a)=\frac{a}{|a|}$ be the phase for non-zero $a$. By convention we set $\sign(0) = 0$. All these operations entry-wisely operate on complex vectors and matrices.

        For a vector $\bm{a}= [a_i]\in\mathbb{C}^n$ we introduce the $\ell_2$-norm $\|\bm{a}\| = ({\sum_i |a_i|^2})^{1/2}$, $\ell_1$-norm $\|\bm{a}\|_1 = \sum_i|a_i|$ and max norm $\|\bm{a}\|_\infty = \max_i|a_i|$. Given   vector $\bm{a}$, 
        $\diag(\bm{a})$ represents the diagonal matrix with main diagonal $\bm{a}$. We also denote the number of non-zero entries by $\|\bm{a}\|_0$, then   $\Sigma^{n}_{s,c} = \{\bm{a}\in \mathbb{C}^n:\|\bm{a}\|_0\leq s\}$ (resp. $\Sigma^n_{s,r}= \{\bm{a}\in\mathbb{R}^n:\|\bm{a}\|_0\leq s\}$) is the set of all complex (resp. real) $s$-sparse signals. Let $\mathbb{S}^{n-1}_c$ (resp. $\mathbb{S}^{n-1}_r$) be the unit Euclidean sphere in $\mathbb{C}^d$ (resp. $\mathbb{R}^d$), we also frequently work with $\Sigma^{n,*}_{s,c}= \Sigma^n_{s,c}\cap \mathbb{S}^{n-1}_c$ and $\Sigma_{s,r}^{n,*} = \Sigma^n_{s,r}\cap \mathbb{S}_r^{n-1}$.

        Given a matrix $\bm{A}=[a_{ij}]\in \mathbb{C}^{n_1\times n_2}$ with singular values $\{\sigma_k:k =1,\cdots,\min\{n_1,n_2\}\}$, we let $\|\bm{A}\|_F=({\sum_{ij}|a_{ij}|^2})^{1/2}$, $\|\bm{A}\|_{*}=\sum_k \sigma_k$,  $\|\bm{A}\|=\max_k \sigma_k$, $\|\bm{A}\|_\infty = \max_{ij}|a_{ij}|$ be its Frobenius norm, nuclear norm, operator norm, and max norm, respectively. We  denote the set of $d_1\times d_2$ complex (resp. real) matrices with rank no more than $r$ by $\mathcal{M}^{n_1,n_2}_{r,c}$ (resp. $\mathcal{M}^{n_1,n_2}_{r,r}$). Their restrictions to matrices with unit Frobenius norm are denoted by $(\mathcal{M}^{n_1,n_2}_{r,c})^*$ and $(\mathcal{M}^{n_1,n_2}_{r,r})^*$. For vectors or matrices $\cdot^\top$, $\cdot^*$ stand for the transpose, conjugate transpose. We use the standard inner product $\big<\bm{A},\bm{B}\big> = \text{Tr}(\bm{A^*B})$, \textcolor{black}{which subsumes the inner product  $\big<\bm{a},\bm{b}\big>=\bm{a}^*\bm{b}$ for vectors $\bm{a},\bm{b}\in \mathbb{C}^n$}. In addition, $\bm{I}_d$ is the $d\times d$ identity matrix, and $\bm{e}_k$ is the $k$-th column of the identity matrix with self-evident dimension. We use $\mathcal{N}^{m\times n}(0,1)$ to denote a $m\times n$ matrix with i.i.d. $\mathcal{N}(0,1)$ entries, and simply $\mathcal{N}(\bm{0},\rho^2\bm{I}_d)=\rho\cdot \mathcal{N}^{d\times 1}(0,1)$ $(\rho>0)$.

        	For $\bm{A}\in \mathbb{C}^{m\times n}$, $\mathcal{S}\subset [m]$, $\mathcal{T}\subset [n]$, we use $\bm{A}^\mathcal{S}_{\mathcal{T}}$ to denote the submatrix of $\bm{A}$ constituted by rows in $\mathcal{S}$ and columns in $\mathcal{T}$. We also write $\bm{A}^\mathcal{S} := \bm{A}^{\mathcal{S}}_{[n]}$, $\bm{A}_{\mathcal{T}}:= \bm{A}_{\mathcal{T}}^{[m]}$. This also applies to column vector $\bm{x}$, for which we sometimes alternatively write $\bm{x}^{[l+1:k]} = \bm{x}^{[k]\setminus[l]}$ for $l<k$. We often switch between $\mathbb{R}$ and $\mathbb{C}$ via $[\cdot]_\mathbb{R}$, $[\cdot]_\mathbb{C}$. In particular, $[\cdot]_\mathbb{R}$ turns $\bm{A}\in \mathbb{C}^{m\times n}$ into \begin{equation}\label{c2r}
        	    [\bm{A}]_\mathbb{R} = \begin{bmatrix}
        	        \Re(\bm{A})\\
                 \Im(\bm{A})
        	    \end{bmatrix}\in \mathbb{R}^{2m\times n},
        	\end{equation}
         while $[\cdot]_\mathbb{C}$ stands for the inverse operation that maps $\bm{A}\in \mathbb{R}^{2m\times n}$ to \begin{equation}\label{r2c}
             [\bm{A}]_\mathbb{C} = \bm{A}^{[m]}+ \bm{A}^{[2m]\setminus[m]}\ii \in \mathbb{C}^{m\times n}.
         \end{equation}
         Note that we allow $[\cdot]_\mathbb{R},[\cdot]_\mathbb{C}$ to operate on a set element-wisely, e.g., $[\mathcal{K}]_\mathbb{R}=\{[\bm{x}]_\mathbb{R}:\bm{x}\in \mathcal{K}\}$ for $\mathcal{K}\subset \mathbb{C}^n$.

    In this work, $C$, $C_i$, $c$, or $c_i$ represent  absolute constants  with value varying from line to line, and we make no attempt to refine these constants. If $T_1\leq CT_2$ for some $C>0$, we write $T_1 = O(T_2)$ or $T_1\lesssim T_2$. The opposite $T_1 \geq C T_2$ would be denoted by $T_1\gtrsim T_2$ or $T_1 = \Omega(T_2)$. As standard complexity notation, $\tilde{O}(\cdot)$ and $\tilde{\Omega}(\cdot)$ further hide  logarithmic factors. The probability, expectation would be given by $\mathbbm{P}(\cdot)$, $\mathbbm{E}(\cdot)$ respectively. Given an event $E$,  $\mathbbm{1}(E)$ is an indicator function that equals to $1$ if $E$ happens, and $0$ otherwise.

    \subsection{Preliminaries}
    \subsubsection{High-dimensional Statistics}
     We first provide some necessary knowledge on high-dimensional statistics, including concentration  of sub-Gaussian or sub-exponential random variable and some covering number results.

    For a real random variable $X$, we define its sub-Gaussian norm $\|\cdot\|_{\psi_2}$, sub-exponential norm $\|\cdot\|_{\psi_1}$ as $\|X\|_{\psi_2} = \inf\{t>0:\mathbbm{E}\exp\big(\frac{X^2}{t^2}\big)<2\}$, $\|X\|_{\psi_1} = \inf \{t>0:\mathbbm{E}\exp\big(\frac{|X|}{t}\big)<2\}$. $X$ is said to be sub-Gaussian (resp. sub-exponential) if $\|X\|_{\psi_2}<\infty$ (resp. $\|X\|_{\psi_1}<\infty$). Note that we have the relation   (e.g.,   \cite[Lem. 2.7.7]{vershynin2018high})\begin{equation}
        \|XY\|_{\psi_1}\leq \|X\|_{\psi_2}\|Y\|_{\psi_2} .
        \label{830add2}
    \end{equation} For some absolute constant $C$, sub-Gaussian $X$ enjoys a probability tail (e.g.,  \cite[Prop. 2.5.2]{vershynin2018high})
    \begin{equation}
    \label{2.1}
        \mathbbm{P}\big(|X | \geq t\big)\leq 2\exp\left(-\frac{Ct^2}{\|X\|_{\psi_2}^2}\right),~~\forall~t>0.
    \end{equation} 
    For independent, zero-mean sub-Gaussian random variables $X_1,\cdots,X_n$, it holds for some $C$ that (e.g.,   \cite[Prop. 2.6.1]{vershynin2018high})
    \begin{equation}
    \label{2.2}
        \Big\|\sum_{i=1}^k X_k\Big\|_{\psi_2}^2 \leq  C \sum_{i=1}^n \|X_i\|^2_{\psi_2}
    \end{equation}
    Note that combining (\ref{2.1}) and (\ref{2.2}) immediately yields the concentration   of the mean of independent sub-Gaussian random variables. For random vector $\bm{X}\in \mathbb{R}^d$, its sub-Gaussian norm is defined by $\|\bm{X}\|_{\psi_2} = \sup_{\bm{V}\in \mathbb{S}^{n-1}_r} \|\bm{V}^\top\bm{X}\|_{\psi_2}$. If the components of $\bm{X}$ are independent sub-Gaussian variables with $O(1)$ sub-Gaussian norm, then a simple fact delivered by (\ref{2.2}) is   $\|\bm{X}\|_{\psi_2} = O(1)$.

    For sub-exponential random variables, we have the following Bernstein's inequality for concentration (e.g.,   \cite[Thm. 2.8.1]{vershynin2018high}).  
    \begin{lem}
 \label{bern}{\rm(Bernstein's inequality)\textbf{.}}
  Let $X_1,...,X_N$ be independent, sub-exponential random variables. Then for any $t>0$, for some constant $C$ we have \begin{equation}
  \begin{aligned}
  \label{lemma1}
     \nonumber
     &\mathbbm{P}\Big(\Big| \sum_{i=1}^N (X_i-\mathbbm{E} X_i)\Big|\geq t\Big)\\&\leq 2\exp\left(-C\min\left\{\frac{t^2}{\sum_{i=1}^N \| X_i\|^2_{\psi_1}},\frac{t}{\max_{i\in [N]}\| X_i\|_{\psi_1}} \right\}\right).
     \end{aligned}
 \end{equation}
 \end{lem}
 Given a subset of $\mathbb{R}^N$ or $\mathbb{C}^N$, denoted $\mathscr{X}$, and its  finite subset $\mathcal{T} $, then $
 \mathcal{T}$ is a $\epsilon$-net of $\mathscr{X}$ if for any $\bm{x}\in \mathscr{X}$, there exists $\bm{a_x}\in \mathcal{T}$ such that $\|\bm{x}-\bm{a_x}\|  \leq \epsilon$. In other words, $\mathscr{X}\subset \cup_{\bm{a}\in\mathcal{T}}\mathcal{B}_\epsilon(\bm{a})$ where $\mathcal{B}_\epsilon(\bm{a})$ is the (real or complex) closed Euclidean ball with center $\bm{a}$ and radius $\epsilon$. The minimum cardinality of the  $\epsilon$-net of $\mathscr{X}$ is usually called its covering number, and the covering number for the set of  sparse signals or low-rank matrices will be recurring in our proof.
 \begin{lem}
  \label{lem2} {\rm (Lemma 3.3 in \cite{plan2013one})\textbf{.}}
  Given $\epsilon>0$, there exists $\mathcal{G}$ that is $\epsilon$-net of $\Sigma^{n,*}_{s,r}$ with cardinality bounded by  $|\mathcal{G}|\leq  \big(\frac{9n}{\epsilon s}\big)^s$. 
 \end{lem}
   Evidently,  Lemma \ref{lem2} implies $\epsilon$-net for $\Sigma^{n,*}_{s,c} $ with cardinality at most $\big(\frac{9n}{\epsilon s}\big)^{2s}$.
  \begin{lem}
      \label{lrnumber}
      {\rm (Lemma 3.1 in \cite{candes2011tight})\textbf{.}} Given $\epsilon>0$, there exists $\mathcal{G}$ that is $\epsilon$-net of $(\mathcal{M}^{n_1,n_2}_{r,r})^*$ with cardinality bounded by $|\mathcal{G}|\leq \big(\frac{9}{\epsilon}\big)^{(n_1+n_2+1)r}$.
  \end{lem}

 \subsubsection{Compressive Sensing}
 As our approach is to recast PO-CS as a classical linear compressive sensing problem, we will also use some well-established facts in this field.  
 Arguably, the RIP introduced below is the central of compressive sensing theory. 

 \begin{definition}
 \label{def1}
 {\rm (Vector RIP)\textbf{.}} Given a sensing matrix $\bm{A}\in\mathbb{R}^{m\times n}$ and $1\leq s\leq n$ is an integer. We say $\bm{A}$ respects
  restricted isometry property over the set of $s$-sparse real signals 
	$\Sigma^n_{s,r}$  for some $\delta >0$ if 
	\begin{equation}
     \label{2.3}
     (1-\delta)\|\bm{u}\|^2 \leq \|\bm{Au}\|^2 \leq (1+\delta)\|\bm{u}\|^2,~~\forall \bm{u}\in \Sigma^n_{s,r}.
 \end{equation}
  The smallest $\delta$ such that (\ref{2.3}) holds is   called the restricted isometry constant (RIC) of order $s$, and denoted by $\delta^{\bm{A}}_s$.
 \end{definition}

 \begin{definition}
 {\rm (Matrix RIP)\textbf{.}} Given a linear map $\mathcal{A}(\cdot)$ from $\mathbb{R}^{n_1\times n_2}$ to $\mathbb{R}^m$ and $1\leq r\leq \min\{d_1,d_2\}$ is an integer. We say $\mathcal{A}$ respects restricted isometry property over $\mathcal{M}^{n_1,n_2}_{r,r}$ for some $\delta>0$ if \begin{equation}
 \label{a2.5}
     (1-\delta) \|\bm{X}\|_F^2 \leq \|\mathcal{A}(\bm{X})\|^2 \leq (1+\delta)\|\bm{X}\|_F^2, ~~\forall ~\bm{X}\in \mathcal{M}_{r,r}^{n_1,n_2}.
 \end{equation}
 The smallest $\delta$ such that (\ref{a2.5}) holds is denoted by $\delta_r^{\mathcal{A}}$.
 \end{definition}

  Under sufficiently small $\delta^{\bm{A}}_{2s}$ (resp. $\delta^{\mathcal{A}}_{2r}$), all signals $\bm{x}$ in $\Sigma^n_{s,r}$ (resp. $\mathcal{M}^{n_1,n_2}_{r,r}$) can be reconstructed from the compressive measurements $\bm{Ax}$ (resp. $\mathcal{A}(\bm{x})$), exactly when there is no noise or stably in the noisy setting, \textcolor{black}{via some instance optimal algorithm that is also robust to model error} 
  (see \cite{cohen2009compressed,keriven2018instance,traonmilin2018stable,bourrier2014fundamental} for instance). Specifically, we present the following result regarding basis pursuit.
  
  \begin{lem}
      \label{cslem}
      {\rm (e.g., \cite{cai2013sparse})\textbf{.}} If $\bm{A}\in \mathbb{R}^{m\times n}$ satisfies $\textcolor{black}{\delta^{\bm{A}}_{2s}}< \frac{\sqrt{2}}{2}$, then all $\bm{x}\in \Sigma^n_{s,r}$ can be exactly reconstructed from $\bm{y} = \bm{Ax}$ via basis pursuit 
      \begin{equation}
 \label{2.5}
     \bm{\hat{x}} = \mathop{\arg\min}\limits_{\bm{u}\in \mathbb{R}^n}~\|\bm{u}\|_1,~~\text{s.t. }\bm{Au}=\bm{y}.
 \end{equation}
 In parallel, if a linear map $\mathcal{A}(\cdot)$ from $\mathbb{R}^{n_1\times n_2}$ to $\mathbb{R}^m$ satisfies $\textcolor{black}{\delta^{\mathcal{A}}_{2r}}<\frac{\sqrt{2}}{2}$, then all $\bm{X}\in \mathcal{M}^{n_1,n_2}_{r,r}$ can be exactly reconstructed from $\bm{y} = \mathcal{A}(\bm{X})$ via constrained nuclear norm minimization \begin{equation}
      \label{a2.6}
     \bm{\hat{X}} = \mathop{\arg\min}\limits_{\bm{U}\in \mathbb{R}^{n_1\times n_2}}~\|\bm{U}\|_{*},~~\text{s.t. }\mathcal{A}(\bm{U})=\bm{y}.
 \end{equation} 
  \end{lem}
  With modified constraint, the above two programs are  stable under measurement noise and robust to model error, see \cite[Thm. 2.1]{cai2013sparse} for instance.


	\section{PO-CS of Sparse Signals}
	\label{section3}
    The most classical signal structure in (1-bit) compressive sensing is undoubtedly sparsity. 
    In this section we assume the desired target signal is sparse, i.e., $\bm{x}\in\mathcal{K}$ where $\mathcal{K}=\Sigma^n_{s,r}$ for the real case, or $\mathcal{K}=\Sigma^n_{s,c}$ for the complex case. \textcolor{black}{The sensing vectors $\bm{\Phi}_k$ are drawn from $\mathcal{N}(\bm{0},\bm{I}_n)+\mathcal{N}(\bm{0},\bm{I}_n)\ii$, and we observe the phase-only measurements \begin{equation}\nonumber
        z_k= \sign\big(\big<\bm{\Phi}_k,\bm{x}\big>\big)=\sign(\bm{\Phi}_k^*\bm{x}),~k=1,2,...,m.
    \end{equation}
    Further, we define $\bm{\Phi}=[\bm{\Phi}_1,...,\bm{\Phi}_m]^*$ as the sensing matrix,\footnote{Compared to defining $\bm{\Phi}=[\bm{\Phi}_1,...,\bm{\Phi}_m]^\top$ as sensing matrix, our treatment allows us to write   $z_k=\sign\big(\big<\bm{\Phi}_k,\bm{x}\big>\big)$  and   $\bm{z}=\sign(\bm{\Phi x})$ without involving conjugate.} and note that    $\bm{\Phi}$ has entries i.i.d. drawn from $\mathcal{N}(0,1)+\mathcal{N}(0,1)\ii$, denoted by $\bm{\Phi}\sim \mathcal{N}^{m\times n}(0,1)+\mathcal{N}^{m\times n}(0,1)\ii$.}
    Now, the PO-CS model can be formulated as \begin{equation}
        \bm{z}=\sign(\bm{\Phi x}),
        \label{a3.1}
    \end{equation}
    We would exclusively use $\bm{x}$, $\bm{z}$ to respectively represent the underlying signal and the phase-only observations. 

    While assuming $\bm{\Phi}_k\sim \mathcal{N}(\bm{0},\bm{I}_n)+\mathcal{N}(\bm{0},\bm{I}_n)\ii$, it is worth pointing out that our result applies \textcolor{black}{as long as the independent sensing vectors $\bm{{\Phi}}_k$'s are non-zero almost surely, and $\frac{\bm{{\Phi}}_k}{\|\bm{{\Phi}}_k\|}$ is uniformly distributed on $\mathbb{S}^{n-1}_c$ (i.e., $\big[\frac{\bm{{\Phi}}_k}{\|\bm{{\Phi}}_k\|}\big]_{\mathbb{R}}$ is uniformly distributed on $\mathbb{S}^{2n-1}_r$). In short, this  is because the norm of each sensing vector is completely absorbed into $\sign(\cdot)$ and hence inessential.   More specifically, one can always   i.i.d. draw $r_k\sim \|\mathcal{N}(\bm{0},\bm{I}_{2n})\|$ (that is independent of $\{\bm{{\Phi}_k}\}$) to  construct $\bm{\hat{\Phi}}_k=\frac{r_k\bm{{\Phi}}_k}{\|\bm{{\Phi}}_k\|}$, and then use the data $\big(\bm{\hat{\Phi}_k},\sign(\bm{{\Phi}}_k^*\bm{x})\big)_{k=1}^m$ for reconstruction, which obviously enjoys our guarantee as  $\bm{\hat{\Phi}}_k\sim \mathcal{N}(\bm{0},\bm{I}_n)+\mathcal{N}(\bm{0},\bm{I}_n)\ii$ (e.g., \cite[Exercise 3.3.7]{vershynin2018high}) and $\sign(\bm{{\Phi}}_k^*\bm{x})=\sign(\bm{\hat{\Phi}}_k^*\bm{x})$.}  Note that a similar observation for 1-bit compressive sensing was made  in  \cite[Remark 1.5]{plan2013one}.

    We also pause to give a big picture of our technical derivation. Compared with sensing a real signal in \cite{jacques2021importance}, dealing with complex signal requires a different reformulation and more complicated calculations. More tricky technical changes include taking the real part in the sign-product embedding property (Remark \ref{rem1}) and rescaling the sensing matrix to render a sufficiently small RIC (Remark \ref{remark2}).  On the other hand, our uniform guarantee follows from   a series of covering arguments
    (e.g., \cite{baraniuk2008simple}) that strengthen each piece in  \cite{jacques2021importance} to be  uniform. Indeed,   covering argument has now become an elementary technique in the field, and the standard procedure is to first show the desired property over a discrete $\epsilon$-net (via union bound), then extends it to the whole signal set of interest. Note that such extension usually relies on certain continuity of the desired property. Thus, in the phase-only scenario, the main difficulty would be the discontinuity $0$ of $\sign(\cdot)$. In fact, $\sign(\cdot)$ behaves extremely badly around $0$. For instance,  complex linear measurements close to $0$ are also close to each other (by triangle inequality), but they can have significantly different phases. Thus, for  $\bm{\Phi}^*_k\bm{x}$ with extremely small magnitude (i.e., $|\bm{\Phi}_k^*\bm{x}|$), taking only the phase fails to preserve information and can even be misleading. Therefore, special attention will be paid to these problematic measurements.
    \textcolor{black}{Note that in PO-CS it is hopeless to recover the signal norm $\|\bm{x}\|$, hence we will soon concentrate on $\bm{x}\in \mathbb{S}^{n-1}$ in subsequent developments. This constraint provides us a  simple way to identify measurements with overly small magnitude, i.e., those with magnitude smaller than some pre-specified threshold.} Specifically, we specify a threshold $\eta$ ($\eta>0$) and collectively call $\bm{\Phi}^*_k\bm{x}$ with  $|\bm{\Phi}^*_k\bm{x}|<\eta$ near vanishing measurement. The indices of near vanishing measurements are collected in the set 
    \begin{equation}
	\label{a3.2}
	    \mathcal{J}_{\bm{x}} = \big\{k\in [m]: |\bm{\Phi}_k^*\bm{x}|<\eta\big\}.
	\end{equation}
    While near vanishing measurements are problematic in the phase-only sensing scenario, fortunately, with high probability   $|\mathcal{J}_{\bm{x}}|/m=O(\eta)$ holds uniformly over $\bm{x}\in \mathcal{K}$ (Lemma \ref{lem10}). Hence, under sufficiently small $\eta$, the influence of near vanishing measurements can be well controlled.


  \subsection{Reformulation}\label{sub3,1}
      We begin with the reformulation of PO-CS to a classical linear compressive sensing problem, which stems from the experiments in \cite{boufounos2013sparse} and also serves as the starting point of the proof in \cite{jacques2021importance}. Note that our task is to find a sparse $\bm{u}$ that admits phase consistency $\bm{z} = \sign(\bm{\Phi u})$. This leads to $\diag(\bm{z})^*\bm{\Phi u}\geq \bm{0}$ with element-wise "$\geq$" and all-zeros vector $\bm{0}$, which equals to \begin{equation}
        \nonumber
        \Im\big(\diag(\bm{z})^*\bm{\Phi u}\big) = \bm{0},~\Re\big(\diag(\bm{z})^*\bm{\Phi u}\big) \geq\bm{0}.
    \end{equation}
We simply retain the dominant information $\Im\big(\diag(\bm{z})^*\bm{\Phi u}\big) = \bm{0}$.  \textcolor{black}{For the real case $\bm{u} \in\mathcal{K} = \Sigma^n_{s,r}$,} this gives $m$ real linear measurements\begin{equation}
        \label{3.1}
        \Im\Big(\frac{1}{\sqrt{m}}\diag(\bm{z})^*\bm{\Phi }\Big) \bm{u} = \bm{0}.
    \end{equation}
    \textcolor{black}{For the complex case $\bm{u}\in\mathcal{K} = \Sigma^n_{s,c}$,} it reads 
    \begin{equation}
        \label{3.2}
       \begin{bmatrix}
            \frac{1}{\sqrt{m}}\Im\big(\diag(\bm{z})^*\bm{\Phi}\big) & \frac{1}{\sqrt{m}}\Re\big(\diag(\bm{z})^*\bm{\Phi}\big) 
        \end{bmatrix}
        \begin{bmatrix}
            \bm{u}^\Re\\ \bm{u}^\Im
        \end{bmatrix} = \bm{0}.
    \end{equation}
    Derived from the phase-only measurements $\bm{z}=\sign(\bm{\Phi x})$, (\ref{3.1}) and (\ref{3.2}) obviously do not contain any information about the signal norm $\|\bm{x}\|$. Hence, to be more aligned with the classical linear compressive sensing, we add  virtual measurement to uniquely specify the norm of the desired signal. Specifically,   due to  the observation $\bm{z^*\Phi x} = \|\bm{\Phi x}\|_1 >0$, we define $\kappa:=\mathbbm{E}|\mathcal{N}(0,1)+\mathcal{N}(0,1)\ii| =\sqrt{\frac{\pi}{2}}$ \footnote{\textcolor{black}{Here, $|\mathcal{N}(0,1)+\mathcal{N}(0,1)\ii|$ follows the Rayleigh distribution with unit variance \cite{papoulis2002probability}.}}  and impose the constraint $\frac{1}{\kappa m} \bm{z^*\Phi u}=1$,\footnote{One can take any $t_1>0$ and specify $\frac{1}{m}\bm{z^*\Phi u}= t_1$, but our   choice here is intended for subsequent RIP analysis (see Lemma \ref{lem4}, Corollary \ref{coro1}).}    which further gives two real linear measurements \cite{boufounos2013sparse,jacques2021importance} \begin{equation}
        \label{add830}
        \Re\left(\frac{ 1
        }{\kappa m}\bm{z^*\Phi u}\right)  = 1,~\Im\left(\frac{ 1
        }{\kappa m}\bm{z^*\Phi u}\right) = 0.
    \end{equation}
    However, a closer look finds that $\Im\big(\frac{1
        }{\kappa m}\bm{z^*\Phi u}\big) = 0$ can be implied by $\Im\big(\diag(\bm{z})^*\bm{\Phi u}\big)=\bm{0}$ and   is hence redundant. Thus,  we only add one virtual measurement $ \Re\big(\frac{1
        }{\kappa m}\bm{z^*\Phi u}\big)  = 1$, which gives $ \frac{ 1}{\kappa m} \Re\big(\bm{z^*\Phi }\big)\bm{u} = 1$
    for the real case and $\frac{1}{\kappa m} \Re\big(\bm{z^*\Phi}\big)\bm{u}^\Re- \frac{ 1}{\kappa m}\Im\big(\bm{z^*\Phi}\big)\bm{u}^\Im = 1$ for the complex case.   We  further introduce a positive scaling $\hat{t}$ to (\ref{3.1}), (\ref{3.2}), then for $\mathcal{K} = \Sigma^n_{s,r}$ we obtain the reformulation \begin{equation}
    \label{3.5}
        \begin{aligned}
         &\text{find }\bm{u}\in \Sigma^n_{s,r},\text{ s.t. }\bm{A}_{\bm{z},r}\bm{u} = \bm{e}_1,\text{ where }\\&\bm{A}_{\bm{z},r} = \begin{bmatrix}
            \frac{1}{\kappa m}\cdot \Re\big(\bm{z^*\Phi}\big)\\
            \frac{\hat{t}}{\sqrt{m}}\cdot\Im\big( \diag(\bm{z^*})\bm{\Phi}\big)
        \end{bmatrix}\in\mathbb{R}^{(m+1)\times n};
        \end{aligned}
    \end{equation}
   For the complex case $\bm{x}\in \Sigma^n_{s,c}$ we similarly arrive at a linear compressive sensing problem 
   \begin{equation}
   \label{3.6}
       \begin{aligned}
       &\text{find }\bm{u}\in\Sigma^{2n}_{2s,r},\text{ s.t. }\bm{A}_{\bm{z},c}\bm{u} = \bm{e}_1,\text{ where}\\ \bm{A}_{\bm{z},c}&=\begin{bmatrix}
            \frac{ 1}{\kappa m}\cdot \Re\big(\bm{z^*\Phi}\big) & - \frac{ 1}{\kappa m}\cdot\Im\big(\bm{z^*\Phi}\big) \\
           \frac{\hat{t}}{\sqrt{m}}\Im\big(\diag(\bm{z^*})\bm{\Phi}\big) &  \frac{\hat{t}}{\sqrt{m}}\Re\big(\diag(\bm{z^*})\bm{\Phi}\big)
       \end{bmatrix}\\&\in\mathbb{R}^{(m+1)\times 2n}.
       \end{aligned}
   \end{equation}
     We shall see shortly that a proper choice of $\hat{t}$ is crucial to render a small RIC of the new sensing matrix\footnote{To distinguish with the original sensing matrix $\bm{\Phi}$, we collectively call $\bm{A}_{\bm{z},r}$, $\bm{A}_{\bm{z},c}$ the new sensing matrix.} $\bm{A}_{\bm{z},c}$. Besides, we mention that $\bm{u}\in \Sigma^{2n}_{2s,r}$ in (\ref{3.6}) does not precisely rephrase the prior knowledge $\bm{x}\in \Sigma^n_{s,c}$. For instance, $\bm{x}\in \Sigma^n_{s,c}$ restricts that at most $s$ of the first $n$ entries of $\bm{u}$ are non-zero, while such information is lost in $\bm{u}\in \Sigma^{2n}_{2s,r}$. But as we do not attempt to refine multiplicative constant in this work, $\bm{u}\in \Sigma^{2n}_{2s,r}$ would be more than sufficient. 

       Based on Lemma \ref{cslem}, our strategy   is to  show the $2s$ order RIC of $\bm{A}_{\bm{z},r}$, and the $4s$ order RIC of $\bm{A}_{\bm{z},c}$, are less than $\frac{\sqrt{2}}{2}$. Compared with the classical compressive sensing where the RIP of the sensing matrix delivers uniform recovery guarantee (Lemma \ref{cslem}),  here   different $\bm{x}\in \mathcal{K}$ corresponds to different $\bm{A}_{\bm{z},r}$ or $\bm{A}_{\bm{z},c}$. Hence, the RIP of a specific new sensing matrix  can only deliver the exact reconstruction of a single underlying $\bm{x}$ (such that $\bm{z} = \sign(\bm{\Phi x})$). As
       the main aim of this paper is to prove a uniform recovery guarantee, we need to show the new sensing matrices for all $\bm{x}\in\mathcal{K}$ simultaneously respect RIP.  By Definition \ref{def1}, this can be precisely formulated for the real case $\mathcal{K} = \Sigma^n_{s,r}$ as   
        \begin{equation}
            \label{3.7}
            \begin{aligned}
            &\sup_{\bm{x}\in \mathcal{K}}\sup_{\bm{u}\in \Sigma^{n}_{2s,r}}\Big|\|\bm{A}_{\bm{z},r}\bm{u} \|^2 - \|\bm{u}\|^2\Big|< \frac{\sqrt{2}}{2} \|\bm{u}\|^2 ,
            \end{aligned}
        \end{equation}
        or for the complex  case $\mathcal{K}= \Sigma^n_{s,c}$ as
        \begin{equation}
            \label{a3.9}
            \sup_{\bm{x}\in\mathcal{K} }\sup_{\bm{u}\in \Sigma^{2n}_{4s,r}}\Big|\|\bm{A}_{\bm{z},c}\bm{u}\|^2 - \|\bm{u}\|^2\Big|< \frac{\sqrt{2}}{2}\|\bm{u}\|^2.
        \end{equation}
        We add the superscript $*$ to restrict a set to elements with unit $\ell_2$-norm, e.g., $\mathcal{K}^* = \Sigma^{n,*}_{s,r}$ (the real case), $\mathcal{K}^*=\Sigma^{n,*}_{s,c}$ (the complex case). Then by the homogeneity of $\bm{u}$ and the property of $\sign(\cdot)$,   (\ref{3.7}) is equivalent to 
        \begin{equation}
            \label{a3.10}
            \begin{aligned}
            &\sup_{\bm{x}\in \Sigma_{s,r}^{n,*}} \sup_{\bm{u}\in \Sigma^{n,*}_{2s,r}} \Big|\frac{1}{\kappa^2m^2}[\Re(\bm{z^*\Phi u})]^2\\&~~~~+\frac{\hat{t}^2}{m}\big\|\Im \big(\diag(\bm{z}^*)\bm{\Phi u}\big)\big\|^2 -1\Big|< \frac{\sqrt{2}}{2}.
        \end{aligned}
        \end{equation}
        Similarly, for (\ref{a3.9}) we identify $\bm{u}\in \Sigma^{2n}_{4s,r}$ with $[\bm{u}]_{\mathbb{C}}=\bm{u}^{[1:n]}+ \bm{u}^{[n+1:2n]}\ii\in \Sigma^{n}_{4s,c}$, the desired (\ref{a3.9}) can thus be implied by  
        \begin{equation}
        \begin{aligned}
            \label{a3.11}
            &\sup_{\bm{x}\in \Sigma_{s,c}^{n,*}} \sup_{\bm{u}\in \Sigma^{n,*}_{4s,c}} \Big|\frac{1}{\kappa^2m^2}[\Re(\bm{z^*\Phi u})]^2\\&~~~~~+\frac{\hat{t}^2}{m}\big\|\Im \big(\diag(\bm{z}^*)\bm{\Phi u}\big)\big\|^2 -1\Big|< \frac{\sqrt{2}}{2}.
            \end{aligned}
        \end{equation}
        Now it shall be clear that our goals in both  cases are to prove
        \begin{equation}
        \label{a3.12}
            \sup_{\bm{x}}\sup_{\bm{u}}  f(\bm{x},\bm{u})  < \frac{\sqrt{2}}{2}
        \end{equation}
         where we introduce a shorthand 
        \begin{equation}\label{short1}
        \begin{aligned}
            &f(\bm{x},\bm{u}) : = \Big|\frac{1}{\kappa^2m^2}[\Re(\bm{z^*\Phi u})]^2\\&~~~~~~~+\frac{\hat{t}^2}{m}\big\|\Im \big(\diag(\bm{z}^*)\bm{\Phi u}\big)\big\|^2 -1\Big|,
            \end{aligned}
        \end{equation}
        and the only difference is that the supremum is taken over $(\bm{x},\bm{u})\in \Sigma^{n,*}_{s,r}\times \Sigma^{n,*}_{2s,r}$ for the real case, while $(\bm{x},\bm{u}) \in \Sigma^{n,*}_{s,c}\times \Sigma^{n,*}_{4s,c}$ for the complex case. Thus, in the sequel we only deal with the complex case where all arguments directly apply to the real case. We will point out the   difference between real case and complex case when appropriate.

    
   To proceed our analysis, the following orthogonal decomposition is an important ingredient. Given a specific underlying signal $\bm{x}\in \Sigma^{n,*}_{s,c}$, we decompose $\bm{u}\in \Sigma^{n,*}_{4s,c}$ into $\bm{u}=\bm{u}_{\bm{x}}^{\|}+\bm{u}_{\bm{x}}^{\bot}$ where 
    \begin{equation}
    \label{a3.14}
            \bm{u}_{\bm{x}}^{\|} = \Re\big<\bm{u},\bm{x}\big>\cdot \bm{x}  ,~
            \bm{u}_{\bm{x}}^\bot = \bm{u} -  \Re\big<\bm{u},\bm{x}\big>\cdot \bm{x}.
    \end{equation}Note that they evidently satisfy $\|\bm{u}_{\bm{x}}^\|\|^2 + \|\bm{u}_{\bm{x}}^\bot\|^2 = 1$. In the sequel, we respectively call $\bm{u_x}^\|,\bm{u_x}^\bot$ the parallel part, the orthogonal part (regarding $\bm{x}$), \textcolor{black}{and now we pause to illustrate these two notions.  To explain our naming of $(\bm{u_x}^\|,\bm{u_x}^\bot)$, we identify $\bm{a}\in \mathbb{S}^{n-1}_c$ with $[\bm{a}]_\mathbb{R} \in \mathbb{S}^{2n-1}_r$, then simple algebra can verify \begin{equation}
        \begin{aligned}
    &~~~[\bm{u_x}^\|]_\mathbb{R}=\big<[\bm{u}]_\mathbb{R},[\bm{x}]_\mathbb{R}\big>\cdot [\bm{x}]_\mathbb{R},\\&[\bm{u_x}^\bot]_\mathbb{R}=[\bm{u}]_\mathbb{R}-\big<[\bm{u}]_\mathbb{R},[\bm{x}]_\mathbb{R}\big>\cdot [\bm{x}]_\mathbb{R}.
        \end{aligned}
    \end{equation}
    Thus,  {\it when $(\bm{x},\bm{u})$ are viewed as vectors in $\mathbb{R}^{2n}$,} $\bm{u}_{\bm{x}}^\|$ can be understood as the projection of $\bm{u}$ onto the linear space spanned by $\bm{x}$ (via real scalars), while $\bm{u_x}^\bot$ is orthogonal to $\bm{x}$. Compared to   the seemingly more natural decomposition $\bm{u}=\big(\big<\bm{u},\bm{x}\big>\cdot\bm{x}\big)+\big(\bm{u}-\big<\bm{u},\bm{x}\big>\cdot\bm{x}\big)$, the main intuition of viewing $(\bm{x},\bm{u})$ as $([\bm{x}]_\mathbb{R},[\bm{u}]_\mathbb{R})$ and   using (\ref{a3.14}) is that, $\bm{x}$ and $\bm{u}$ are already regarded as real signals when we derive the reformulation, see (\ref{3.2}), (\ref{3.6}). We moreover comment that,   PO-CS itself seems more like a real problem, in the sense that the trivial ambiguity is a  positive scaling factor (rather than a complex scalar). Moreover, (\ref{a3.14}) allows some useful observations such as (\ref{3.11})   that facilitate our RIP analysis.}

    Note that   RIP indeed requires that the sensing matrix   operates on $\Sigma^{n,*}_{4s,c}$ (the range of $\bm{u}$) in a near isometry way. We shall see that, $\frac{1}{\kappa^2m^2}|\Re(\bm{z^*\Phi u})|^2$, which stems from the first row of $\bm{A}_{\bm{z},c}$   (\ref{3.6})  that is added to specify signal norm, contributes to the parallel part $\|\bm{u}_{\bm{x}}^\|\|^2$. On the other hand, $\frac{\hat{t}^2}{m}\|\Im\big(\diag(\bm{z^*})\bm{\Phi u}\big)\|^2$, which corresponds to the the second row of $\bm{A}_{\bm{z},c}$  (\ref{3.6}) that comes from the phase-only measurements, will  provide the orthogonal part $\| \bm{u}^\bot_{\bm{x}}\|^2$. Indeed, a simple observation to preview such division is  that $\Im\big(\diag(\bm{z^*})\bm{\Phi u}\big)$ fails to sense the parallel part: \begin{equation}
    \label{3.11}
    \begin{aligned}
        &\Im\big(\diag(\bm{z^*})\bm{\Phi u^\|_x}\big) \\=&\Re\big<\bm{u},\bm{x}\big> \cdot \Im\big(\diag\big(\sign(\bm{\Phi x})\big)^*\bm{\Phi x}\big) = \bm{0}.
        \end{aligned}
    \end{equation} 
    Having observed this, we further decompose $f(\bm{x},\bm{u})$ by triangle inequality as 
    \begin{equation}
        \label{aa3.14}
        \begin{aligned}
        f(\bm{x},\bm{u}) &\leq \underbrace{\left|\frac{[\Re(\bm{z^*\Phi u})]^2}{\kappa^2m^2} - \|\bm{u}_{\bm{x}}^\|\|^2\right|}_{f^\|(\bm{x},\bm{u})} \\&+ \underbrace{\left|\frac{\hat{t}^2}{m}\|\Im\big(\diag(\bm{z}^*)\bm{\Phi u}\big)\|^2- \|\bm{u}_{\bm{x}}^\bot\|^2\right|}_{f^\bot(\bm{x},\bm{u})}.
        \end{aligned}
    \end{equation}
    Focusing on the complex case of (\ref{a3.12}), we only need to show
    \begin{equation}
    \label{aa3.17}\begin{aligned}
        &\sup_{\bm{x}\in \Sigma^{n,*}_{s,c}}\sup_{\bm{u}\in \Sigma^{n,*}_{4s,c}} f^\|(\bm{x},\bm{u}) \\&~~~~~+   \sup_{\bm{x}\in \Sigma^{n,*}_{s,c}}\sup_{\bm{u}\in \Sigma^{n,*}_{4s,c}} f^\bot(\bm{x},\bm{u})< \frac{\sqrt{2}}{2}.\end{aligned}
    \end{equation}
    Now we pause to present our main result in this section first. 
    \begin{theorem}
       \label{theorem1}
       {\rm (Uniform Exact Recovery of Complex Sparse Signals)\textbf{.}} Assume  the sensing matrix $\bm{\Phi} \sim \mathcal{N}^{m\times n}(0,1)+\mathcal{N}^{m\times n}(0,1)\ii$,   $\bm{x}\in \mathcal{K}$ with $\mathcal{K}= \Sigma^n_{s,r}$ or $\Sigma^n_{s,c}$, and one aims to reconstruct $\bm{x}$ from $\bm{\Phi}$ and the phase-only measurements $\bm{z} = \sign(\bm{\Phi x})$. Such PO-CS problem can be reformulated as 
       (\ref{3.5}) with $\hat{t} = 1$ if $\mathcal{K} = \Sigma^{n}_{s,r}$, or (\ref{3.6}) with $\hat{t}=\sqrt{\frac{2}{3}}$ if $\mathcal{K} = \Sigma^{n}_{s,c}$.
       Given any $\delta>0$, if   \begin{equation}
       \label{3.46}
           m\geq \frac{Cs}{\delta^4}\log\Big(\frac{n^2\log(mn)}{\delta^6 s}\Big)
       \end{equation} 
       for some sufficiently large $C$, then   with probability at least $1-c_1\exp(-c_2\delta^4m)-2(mn)^{-9}$,     
       $\bm{A}_{\bm{z},r}$ for all $\bm{x}\in \Sigma^{n}_{s,r}$ has $2s$ order RIC smaller than $\delta$, or $\bm{A}_{\bm{z},c}$ for all $\bm{x}\in \Sigma^{n}_{s,c}$ has $4s$ order RIC smaller than $\frac{1}{3}+\delta$. In particular,  under a sample complexity $m \geq C_1s\log\big(\frac{n^2\log(mn)}{s}\big)$ with sufficiently large $C_1$, with probability at least $1-2(mn)^{-9}-c_1\exp(-c_3m)$, all
       $\bm{x}\in \Sigma^n_{s,c}$ can be exactly reconstructed  (up to positive scaling factor) from $\bm{z}$ via basis pursuit of (\ref{3.5}) for the real  case, or (\ref{3.6}) for the complex case.
    \end{theorem}
    
    \begin{rem}  
    {\rm (Near optimality of uniform sparse recovery)\textbf{.}}
    Compared with the optimal sample complexity $O(s\log \big(\frac{n}{s}\big))$ for uniform $s$-sparse recovery in linear compressive sensing \cite{Foucart2013}, the presented $O(s\log\big(\frac{n^2\log(mn)}{s}\big))$ for achieving the same goal via only the measurement phase is near optimal up to logarithmic degradation. This illustrates the importance the phase in complex compressive sensing.  
    \end{rem}
    \begin{rem}\label{pricetoget}\textcolor{black}{
        {\rm (The price to get uniformity)\textbf{.}} Specialized to sparse signal, the non-uniform guarantee presented in \cite[Thm. 3.3]{jacques2021importance} and our Theorem \ref{theorem5} (non-uniform guarantee for complex signal) require $m\gtrsim \frac{s}{\delta^2}\log \frac{n}{s}$ to make the RIC   of the new sensing matrix (the one corresponding to some fixed signal $\bm{x}$) decrease linearly with $\delta$. By contrast, besides additional logarithmic factors, Theorem \ref{theorem1} requires  $m=\tilde{\Omega}\big(\frac{s}{\delta^4}\big)$   to guarantee the RIC of the new sensing matrices corresponding to all $\bm{x}\in \Sigma^n_{s,r}/\Sigma^n_{s,c}$ simultaneously linearly decrease with $\delta$. Such worse dependence on $\delta$ can be understood as the price we pay to get uniformity, and currently we do not know whether the  $\delta^{-4}$ dependence is optimal. Fortunately, since we do not attempt to refine multiplicative constant, this is almost harmless to the recovery guarantee of primary interest --- as we only need the RSC  to be smaller than some strictly positive threshold (rather than concentration around $0$).}
    \end{rem}

    \subsection{Auxiliary Estimates}
    The proof of our main result hinges on a bunch of auxiliary estimates, which are collectively presented in this subsection to improve the readability of this paper. We first give the RIP of the original sensing matrix  $\bm{\Phi}$. 
    \begin{lem}
    \label{lem3}
     Assume $\bm{\Phi}\sim \mathcal{N}^{m\times n}(0,1)+\mathcal{N}^{m\times n}(0,1)\ii$, $l$ is a some positive integer. Given any $\delta>0$, there exists some constants $c_1$, $c_2$ \textcolor{black}{depending  on $\delta,l$}, such that when 
     $m\geq  {c_1s  \log\big(\frac{n}{s}\big)}$, with probability at least $1-2\exp(-c_2m)$, we have
    $(1-\delta)\|\bm{u}\|^2\leq\frac{1}{2m}\|\bm{\Phi u}\|^2\leq (1+\delta)\|\bm{u}\|^2 $ holds uniformly for all $\bm{u}\in \Sigma^{n}_{ls,c}$. In particular we take $\delta = \frac{1}{4}$, then for  some constants $c_3,c_4$ \textcolor{black}{only depending on $l$ (note that they become absolute constants if $l$ is   specified and fixed)}, when $m\geq c_3s\log\big(\frac{n}{s}\big)$, with probability exceeding $1-2\exp(-c_4m)$,   we have\begin{equation}
    \label{a3.16}
        \frac{3}{2} \leq \frac{1}{m}\|\bm{\Phi u }\|^2 \leq \frac{5}{2},~~\forall \bm{u}\in  \Sigma^{n,*}_{ls,c}~.
    \end{equation} 
    \end{lem} 
    \noindent{\it Proof.}  We only need to prove the  statement for general $\delta>0$. For the real case $\bm{\Phi} \sim \mathcal{N}^{m\times n}(0,1)\ii$ and $\bm{u}\in \Sigma^{n}_{ls,r}$, such result (with $\frac{1}{2m}\|\bm{\Phi u}\|^2$ changed to $\frac{1}{m}\|\bm{\Phi u}\|^2$) is indeed the backbone of compressive sensing theory, see    \cite[Thm. 5.2]{baraniuk2008simple} for instance. To adjust it to the complex case here, we only need to calculate $\|\bm{\Phi u}\|^2 = \|[\bm{\Phi}^\Re,-\bm{\Phi}^{\Im}][\bm{u}]_{\mathbb{R}}\|^2 + \|[\bm{\Phi}^\Im,\bm{\Phi}^\Re][\bm{u}]_\mathbb{R}\|^2$ and note  that $[\bm{\Phi}^\Re,-\bm{\Phi}^\Im],[\bm{\Phi}^\Im,\bm{\Phi}^\Re]\sim \mathcal{N}^{m\times 2n}(0,1)$,   $\bm{u}\in \Sigma^n_{ls,c}$ implies $[\bm{u}]_{\mathbb{R}}\in \Sigma^{2n}_{2ls,r}$. 
    Then applying the result for the real case concludes the proof. \hfill $\square$
    
    \vspace{2mm}
    

\textcolor{black}{The rotational invariance of $\bm{\Phi}_k\sim \mathcal{N} (\bm{0},\bm{I}_n)+\mathcal{N}(\bm{0},\bm{I}_n)\ii$  will be recurring in our proofs. Specifically, given any complex unitary matrix $\bm{P}$ (i.e., $\bm{PP}^*=\bm{I}_n$), the rotational invariance states that $\bm{P\Phi}_k$ and $\bm{\Phi}_k$   have the same distribution. To see this, one can verify the equality $[\bm{P\Phi}_k]_\mathbb{R}=\bm{P}'[\bm{\Phi}_k]_\mathbb{R}$ where \begin{equation}\nonumber
   \bm{P}' = \begin{bmatrix}\Re(\bm{P}) & -\Im (\bm{P})\\
   \Im (\bm{P}) & \Re(\bm{P})\end{bmatrix}
\end{equation} is an orthogonal matrix (i.e., $\bm{P}'(\bm{P}')^\top= \bm{I}_{2n}$), and then use the rotational invariance of $[\bm{\Phi}_k]_\mathbb{R}\sim \mathcal{N}^{2n\times 1}(0,1)$, see \cite[Prop. 3.3.2]{vershynin2018high} for instance. For $\bm{\Phi}\sim \mathcal{N}^{m\times n}(0,1)+\mathcal{N}^{m\times n}(0,1)\ii$, similarly, $\bm{\Phi}\bm{P}\sim \mathcal{N}^{m\times n}(0,1)+\mathcal{N}^{m\times n}(0,1)\ii$ holds for any unitary matrix $\bm{P}$.}

    Due to the rotational invariance of $\bm{\Phi}$, $\mathbbm{E}\frac{1}{m}\|\bm{\Phi w}\|_1 = \mathbbm{E}|\bm{\Phi}^*_k\bm{w}| = \mathbbm{E}|\mathcal{N}(0,1)+\mathcal{N}(0,1) \ii| = \kappa$  holds for any $\bm{w}\in \mathbb{S}_{c}^{n-1}$. As   \cite[Lem. 5.2]{jacques2021importance}, the concentration of $\frac{1}{\kappa m}\|\bm{\Phi w}\|_1$ around $1$ will be needed later. However, unlike their Lemma 5.2   only established for a fixed $\bm{w}$, we need a uniform concentration property over all sparse $\bm{w}$. This can be achieved by straightforwardly applying a covering argument over $\bm{w}\in \Sigma^{n,*}_{s,c}$. \textcolor{black}{We note that, a comparable result termed    $(\ell_1,\ell_2)$-RIP was established in   \cite[Thm. 6]{feuillen2020ell}, while we still include a proof here for completeness.} 
    
    \begin{lem}
    \label{lem4}
      Assume $\bm{\Phi}\sim \mathcal{N}^{m\times n}(0,1)+\mathcal{N}^{m\times n}(0,1)\ii$ and fix any $\delta\in (0,1)$. When $m\geq C\frac{s}{\delta^2}\log\big(\frac{n}{\delta s}\big)$ for some absolute constant $C$, with probability at least $1-2\exp(-c\delta^2 m)$ we have \begin{equation}
          \label{3.14}
          \sup_{\bm{w}\in \Sigma^{n,*}_{s,c}} \Big|\frac{\|\bm{\Phi w}\|_1}{\kappa m} -1\Big|\leq \delta.      
          \end{equation}
    \end{lem}

    \noindent{\it Proof.} \textbf{(Step 1.)  Establish the concentration regarding a fixed $\bm{w}\in \Sigma^{n,*}_{s,c}$}

    Note that \begin{equation}
        \label{a3.18}
        \Big|\frac{\|\bm{\Phi w}\|_1}{\kappa m}  - 1\Big| = \Big|\frac{1}{m}\sum_{k=1}^m \Big(\frac{| \bm{\Phi}_k^*\bm{ w}|}{\kappa}  - \mathbbm{E} \frac{| \bm{\Phi}_k^* \bm{w}|}{\kappa}\Big)\Big|.
    \end{equation}
    By rotational invariance of $\bm{\Phi}_k$, $\{|\bm{\Phi}_k^*\bm{w}|:k\in [m]\}$ are independent copies of $|\mathcal{N}(0,1)+\mathcal{N}(0,1)\ii|$. Thus, we can estimate its sub-Gaussian norm by definition
    \begin{equation}
    \label{3.15}
       \begin{aligned}
        & \big\| |\bm{\Phi}_k^* \bm{w}|\big\|_{\psi_2} = \big\| |\mathcal{N}(0,1)+\mathcal{N}(0,1)\ii|\big\|_{\psi_2}\\&\leq \|\mathcal{N}(0,1)\|_{\psi_2}+\|\mathcal{N}(0,1)\|_{\psi_2} = O(1).
       \end{aligned}
    \end{equation}
    Then, by centering  (see \cite[Lem. 2.6.8]{vershynin2018high}), $\big\||\bm{\Phi_k^* w}| - \mathbbm{E}|\bm{\Phi_k^* w}|\big\|_{\psi_2} = O(1)$. Moreover, recalling (\ref{a3.18}) and using (\ref{2.2})  leads to 
   $
        \big\|\frac{\|\bm{\Phi w}\|_1}{\kappa m}-1\big\|_{\psi_2}^2 = O\big(\frac{1}{m}\big).
   $
  Thus, we apply (\ref{2.1})  and obtain  
    \begin{equation}
    \label{a3.21}
        \mathbbm{P}\left(\Big|\frac{\|\bm{\Phi w}\|_1}{\kappa m} -1\Big| \geq t\right) \leq 2\exp(-cmt^2) 
    \end{equation}
    holds for any $t\geq 0$, for some absolute constant $c$. 
    
    \vspace{1mm}
    \noindent
\textbf{(Step 2.) Strengthen (\ref{a3.21}) to a finite net and then further to $\bm{w}\in \Sigma^{n,*}_{s,c}$}

     We now invoke a covering argument to strengthen the concentration to $\bm{w}\in \Sigma^{n,*}_{s,c}$. We construct $\mathcal{G}_{\delta/4}$ as a $\frac{\delta}{4}$-net of $\Sigma^{n,*}_{s,c}$, and we can assume $|\mathcal{G}_{\delta/4}|\leq \big(\frac{36n}{\delta s}\big)^{2s}$ by Lemma \ref{lem2}. Based on (\ref{a3.21}), a union bound gives 
    \begin{equation}\nonumber
    \begin{aligned}
        &\mathbbm{P}\left(\sup_{\bm{w}\in\mathcal{G}_{\delta/4}}\Big|\frac{\|\bm{\Phi w}\|_1}{\kappa m} -1\Big|\geq t\right)\\&~~~~~~~~~\leq 2\exp\left(-cmt^2 +2s\log\Big(\frac{36n}{\delta s}\Big)\right) \end{aligned}
    \end{equation}for any $t>0$.
    We set $t = \frac{\delta }{2}$, then as long as $m=C\frac{s}{\delta^2}\log\big(\frac{n}{\delta s}\big)$ for sufficiently large $C$, with probability at least $1-2\exp(-\frac{c}{8}\delta^2m)$ we have $\sup_{\bm{w}\in\mathcal{G}_{\delta/4}}\big|\frac{1}{\kappa m}\|\bm{\Phi w}\|_1 - 1\big|\leq \frac{\delta }{2}$.

    We proceed by   assuming (\ref{a3.16}) in Lemma \ref{lem3} holds (with a parameter $l\geq 2$).  Since $\Sigma^{n,*}_{s,c}$ is compact, we can assume $\sup_{\bm{w}\in \Sigma^{n,*}_{s,c}}\big|\frac{1}{\kappa m}\|\bm{\Phi w}\|_1 - 1\big| = \big|\frac{1}{\kappa m}\|\bm{\Phi\hat{w}}\|_1 - 1\big|$ for some $\bm{\hat{w}}\in \Sigma^{n,*}_{s,c}$. Note that $\bm{\hat{w}}$ can be approximated by some $\bm{\tilde{w}}\in\mathcal{G}_{\delta/4}$ such that $\|\bm{\tilde{w}}-\bm{\hat{w}} \| \leq \frac{\delta}{4}$, then we  can estimate that
    \begin{equation}
        \begin{aligned}\nonumber
        & \Big|\frac{1}{\kappa m}\|\bm{\Phi \hat{w}}\|_1 - 1\Big|\\\leq &\Big|\frac{1}{\kappa m}\|\bm{\Phi \tilde{w}}\|_1 -1\Big| + \frac{1}{\kappa m}\big|\|\bm{\Phi\tilde{w}}\|_1 -\|\bm{\Phi\hat{w}}\|_1 \big|\\
         \leq &\sup_{\bm{w}\in\mathcal{G}_{\delta/4}} \Big|\frac{1}{\kappa m}\|\bm{\Phi w}\|_1 -1\Big| + \frac{1}{\kappa m} \|\bm{\Phi(\tilde{w}-\hat{w})}\|_1 \\
         \leq&\frac{\delta}{2} + \frac{1}{\kappa\sqrt{m}} \big\|\bm{\Phi}\cdot \frac{\bm{\tilde{w}}-\bm{\hat{w}}}{\|\bm{\tilde{w}}-\bm{\hat{w}}  \|}\big\|_2\cdot \| \bm{\tilde{w}}-\bm{\hat{w}} \|   <\delta.
        \end{aligned}
    \end{equation}
    The proof is complete. \hfill $\square$

    \vspace{2mm}
    The next lemma gives an upper bound on $\|\bm{\Phi}\|_\infty$, which is indeed a standard estimate for a finite sequence of sub-Gaussian random variables. We provide the proof for completeness.  
    \begin{lem}
      \label{lemma5}
     Assume $\bm{\Phi}\sim \mathcal{N}^{m\times n}(0,1)+\mathcal{N}^{m\times n}(0,1)\ii$. For some absolute constant $C$,   $\|\bm{\Phi}\|_\infty\leq C \sqrt{\log(mn)}$ holds with probability exceeding $1-2(mn)^{-9}$.
    \end{lem}
    
    \noindent{\it Proof.}  As shown in (\ref{3.15}), the $(i,j)$-entry of $\bm{\Phi}$, denoted $\phi_{ij}$, has $O(1)$ sub-Gaussian norm, hence by (\ref{2.1}) for some $c$ we have $
        \mathbbm{P}(|\phi_{ij}|\geq t)\leq 2\exp(-ct^2) 
    $ for all $t>0$.
    A union bound gives $\mathbb{P}(\|\bm{\Phi}\|_\infty \geq t)\leq 2mn\cdot \exp(-ct^2)$. The proof can be concluded by setting $t = \sqrt{\frac{10}{c}\log(mn)}$. \hfill $\square$


    Then we give a lemma for       estimating $|\sign(a) - \sign(b)|$ by $|a-b|$. Note that the bound itself reflects the issue when both $a$, $b$ are close to $0$.
    
    \begin{lem}
        \label{signconti}
        Given $a,b\in \mathbb{C}$, we conventionally let $\frac{x}{0} = \infty$ for any $x\geq 0$, then we have \begin{equation}
            \label{a3.23}
            |\sign(a)-\sign(b)| \leq \min\Big\{\frac{2|a-b|}{\max\{|a|,|b|\}} ,2\Big\}.
        \end{equation}
    \end{lem}

    \noindent{\it Proof.} When $a=0$ or $b =0$, the result is obvious, so we assume $a,b$ are non-zero. By boundedness, $|\sign(a)-\sign(b)| \leq |\sign(a)|+|\sign(b)| = 2$. Since $a$, $b$ are symmetric, it remains to show the upper bound $\frac{2|a-b|}{|a|}$, which follows from some algebra as follows 
    \begin{equation}
        \begin{aligned}\nonumber
         &|\sign(a)-\sign(b)| = \left|\frac{a}{|a|} - \frac{b}{|b|}\right| \leq \left|\frac{a-b}{|a|}\right| + \left|\frac{b}{|a|}-\frac{b}{|b|}\right|\\
         &\leq \frac{|a-b|}{|a|} + \left|\frac{|b|}{|a|}-1\right| = \frac{|a-b|}{|a|} + \frac{\big||b|-|a|\big|}{|a|}\leq \frac{2|a-b|}{|a|},
        \end{aligned}
    \end{equation}
thus completing the proof. \hfill $\square$

    \vspace{1mm}

    Under the pre-specified threshold $\eta>0$, recall that the $k$-th measurement $\bm{\Phi}^*_k\bm{x}$ is called near vanishing measurement if $|\bm{\Phi}^*_k\bm{x}|<\eta$. The corresponding indices are collected in $\mathcal{J}_{\bm{x}}$ (\ref{a3.2}). As analyzed at the beginning of this section, the phase of near vanishing measurement is not informative for recovery. To address such issue, we establish a   bound for $\sup_{\bm{x}\in \Sigma^{n,*}_{s,c}}|\mathcal{J}_{\bm{x}}|$ in the next lemma.

     \begin{lem}
	 \label{lem10}
	 Given any $\beta\in (0,1)$ such that $\beta m$ is an integer, we assume the pre-specified threshold $\eta$ satisfies $\eta < \frac{\beta}{4}$.
	 If for some sufficiently large $C_1 $\begin{equation}
	 \label{4.3}
	     m \geq   \frac{C_1  s}{\beta^2}\log\Big(\frac{n^2\log(mn)}{\eta^2 s}\Big) ,
	 \end{equation}
	    then with probability at least $1- \exp(-c\beta^2m)-2(mn)^{-9}$ we have \begin{equation}
	 \label{4.4}
	     \sup_{\bm{x}\in \Sigma^{n,*}_{s,c}}|\mathcal{J}_{\bm{x}}|< \beta m.
	 \end{equation}
	 \end{lem}
	 
	 \noindent{\it Proof.}  {\textbf{(Step 1.) A useful observation.}}

  Recall that $\mathcal{J}_{\bm{x}}=\{k\in [m]:|\bm{\Phi}_k^*\bm{x}| <\eta \}$, $\mathcal{K}^*=\Sigma_{s,c}^{n,*}$.
  We define $\mathcal{E}=\{E\subset [m]:|E| = \beta m\}$ and first observe that (\ref{4.4}) is equivalent to  
	 \begin{equation}
	 \label{4.5}	      \inf_{\bm{x}\in\mathcal{K}^*}\inf_{E\in \mathcal{E}} \|\bm{\Phi}^E\bm{x}\|_\infty \geq \eta.
	 \end{equation}
On one hand, if (\ref{4.4}) holds, then for any $\bm{x}\in \mathcal{K}^*$, $|\mathcal{J}_{\bm{x}}|<\beta m$. Thus, given any $E\in \mathcal{E}$, $E\setminus \mathcal{J}_{\bm{x}}$ is non-empty, and note that for any $k\in E\setminus \mathcal{J}_{\bm{x}}$, $|\bm{\Phi}_k^*\bm{x}|\geq \eta$.  We hence arrive at $\|\bm{\Phi}^E\bm{x}\|_\infty \geq \eta$ uniformly for all $\bm{x}\in \mathcal{K}^*,E\in \mathcal{E}$, i.e., (\ref{4.5}) holds true.    On the other hand, if (\ref{4.4}) does not hold, then there exists $\bm{x}_0\in \mathcal{K}^*$, such that $|\mathcal{J}_{\bm{x}_0}|\geq \beta m$. Hence, one can find $E_0\subset \mathcal{J}_{\bm{x}_0}$ such that $|E_0| = \beta m$. By definition of $\mathcal{J}_{\bm{x}_0}$ (\ref{a3.2}), we know $\|\bm{\Phi}^{E_0}\bm{x}_0\|_\infty <\eta$. This contradicts (\ref{4.5}). Therefore, we only need to show (\ref{4.5}) holds with high probability. By Lemma \ref{lemma5} with probability at least $1-2(mn)^{-9}$ we can assume $\|\bm{\Phi}\|_\infty \leq C\sqrt{\log(mn)}$. In the remainder of this proof, we first consider a fixed $\bm{x}\in \mathcal{K}^*$   {(Step 2)} and then invoke a covering argument   {(Step 3)}.

	\noindent \textbf{(Step 2.) Deal with a fixed $\bm{x}\in \mathcal{K}^*$.}

 We fix $\bm{x}\in\mathcal{K}^*$. Due to $\|\bm{x}\|=1$ and rotational invariance,   $\bm{\Phi}^*_k\bm{x} \sim \mathcal{N}(0,1)+\mathcal{N}(0,1)\ii$,   so a simple estimate follows
	 \begin{equation}
	     \begin{aligned}
	     \label{a4.5}
	      &\mathbbm{P}(|\bm{\Phi_k^*x}|\leq 2\eta) \leq \mathbbm{P} (|\Re (\bm{\Phi_k^*x})| \leq 2\eta)\\ &=\int_{-2\eta}^{2\eta} \frac{1}{\sqrt{2\pi}}\exp\big(-\frac{w^2
	     }{2}\big)\mathrm{d}w \leq 2\eta.
	     \end{aligned}
	 \end{equation}
	Because $\sum_{k=1}^m \mathbbm{1}\big(|\bm{\Phi}_k^*\bm{x}|<2\eta\big)=\big|\{k\in [m]:|\bm{\Phi}_k^*\bm{x}|<2\eta\}\big|$, by similar reasoning in part \textbf{(i)} of this proof we obtain
	 \begin{equation}\nonumber\begin{aligned}
	    & \inf_{E\in \mathcal{E}} \| \bm{\Phi}^E\bm{x}\|_\infty <2\eta \\\iff& \sum_{k=1}^m \mathbbm{1}\big(|\bm{\Phi^*}_k\bm{x}|<2\eta\big) \geq \beta m.\end{aligned}
	 \end{equation}
    	 Note that $\mathbbm{1}\big(|\bm{\Phi^*}_k\bm{x}|<2\eta\big)$ with $k\in [m]$ are i.i.d. bounded random variable, and    (\ref{a4.5}) gives $\mathbbm{E}\big(\mathbbm{1}\big(|\bm{\Phi^*}_k\bm{x}|<2\eta\big)\big) \leq 2\eta$. Thus, for given $\beta>0$ and some $\eta$ satisfying $\eta <\frac{\beta}{4}$, we can apply Hoeffding's inequality (e.g.,   \cite[Thm. 1.9]{rigollet2015high}) to obtain \begin{equation}
	     \begin{aligned}
	     \label{a4.7}
	      & \mathbbm{P}\Big(\inf_{E\in \mathcal{E}} \| \bm{\Phi}^E\bm{x}\|_\infty <2\eta\Big)\\ = &\mathbbm{P}\Big(\frac{1}{m}\sum_{k=1}^m \mathbbm{1}\big(|\bm{\Phi}_k^*\bm{x}|<2\eta\geq \beta\big)\Big)\\ \leq &\mathbbm{P}\Big(\frac{1}{m}\sum_{k=1}^m \mathbbm{1}\big(|\bm{\Phi^*}_k\bm{x}|<2\eta\big)\\&~~~~~~~~~ - \mathbbm{E}\big( \mathbbm{1}\big(|\bm{\Phi^*}_k\bm{x}|<2\eta\big)\big)\geq \frac{\beta}{2}   \Big) \\
	       \leq &\exp \big(-\frac{1}{2}\beta^2 m\big).
	     \end{aligned}
	 \end{equation}

	 \noindent \textbf{(Step 3.)     Strengthen (\ref{a4.7}) to all $\bm{x}\in \mathcal{K}^*$ via covering argument}
  
	 We construct $\mathcal{G}_{\hat{\delta}}$ as a $\hat{\delta}$ net of $\mathcal{K}^*$ (where $\hat{\delta}$ will be chosen later), by Lemma \ref{lem2} we assume $|\mathcal{G}_{\hat{\delta}}| \leq \big(\frac{18n}{{\hat{\delta}} s}\big)^s$.  Taking a union bound over $\mathcal{G}_{\hat{\delta}}$, (\ref{a4.7}) gives \begin{equation}
	 \label{a4.8}
  \begin{aligned}
	    &\mathbbm{P}\Big(\inf_{\bm{x}\in\mathcal{G}_{\hat{\delta}}}\inf_{E\in \mathcal{E}} \| \bm{\Phi}^E\bm{x}\|_\infty<2\eta\Big) \\&~~~~~\leq \exp\Big(-\frac{1}{2}\beta^2 m + s\log\big(\frac{18n}{{\hat{\delta}} s}\big)\Big).
      \end{aligned}
	 \end{equation}
	 By compactness, there exist $\bm{\hat{x}}\in \mathcal{K}^*$, $\hat{E}\in \mathcal{E}$ such that  $\inf_{\bm{x}\in \mathcal{K}^*}\inf_{E\in\mathcal{E}}\|\bm{\Phi}^E\bm{x}\|_\infty=\|\bm{\Phi}^{\hat{E}}\bm{\hat{x}}\|_\infty.$  Thus, we can pick $\bm{\tilde{x}}\in \mathcal{G}_{\hat{\delta}}$ so that $\|\bm{\tilde{x}}- \bm{\hat{x}}\| \leq {\hat{\delta}}$, which leads to 
	 \begin{equation}
	     \begin{aligned}
	     \label{a4.9}
	      &~~~\inf_{\bm{x}\in \mathcal{K}^*}\inf_{E\in\mathcal{E}}\|\bm{\Phi}^E\bm{x}\|_\infty 
	     \\& \geq \inf_{\bm{x}\in \mathcal{G}_{\hat{\delta}}}\inf_{E\in\mathcal{E}}\|\bm{\Phi}^E\bm{x}\|_\infty +  \|\bm{\Phi}^{\hat{E}}\bm{\hat{x}}\|_\infty - \|\bm{\Phi}^{\hat{E}}\bm{\tilde{x}}\|_\infty \\
	      &\geq \inf_{\bm{x}\in \mathcal{G}_{\hat{\delta}}}\inf_{E\in\mathcal{E}}\|\bm{\Phi}^E\bm{x}\|_\infty - \|\bm{\Phi} (\bm{\hat{x}}-\bm{\tilde{x}})\|_\infty  \\
	      &\geq  \inf_{\bm{x}\in \mathcal{G}_{\hat{\delta}}}\inf_{E\in\mathcal{E}}\|\bm{\Phi}^E\bm{x}\|_\infty - \|\bm{\Phi}\|_\infty\cdot \sqrt{2s}\cdot \|\bm{\hat{x}}-\bm{\tilde{x}}\| \\
	      &\geq \inf_{\bm{x}\in \mathcal{G}_{\hat{\delta}}}\inf_{E\in\mathcal{E}}\|\bm{\Phi}^E\bm{x}\|_\infty - C\sqrt{2s\log(mn)}\cdot {\hat{\delta}} \\&= \inf_{\bm{x}\in \mathcal{G}_{\hat{\delta}}}\inf_{E\in\mathcal{E}}\|\bm{\Phi}^E\bm{x}\|_\infty-\eta ,
	     \end{aligned}
	 \end{equation}
     where in the fourth line we use $\bm{\hat{x}}-\bm{\tilde{x}}\in \Sigma^n_{2s}$, and 
      we  set ${\hat{\delta}} = \frac{\eta}{C\sqrt{2s\log(mn)}}$ so that the last line holds. Also put our choice of ${\hat{\delta}}$ into (\ref{a4.8}), provided the sample complexity (\ref{4.3}) for   sufficiently large $C_1$, with probability at least $1-\exp\big(-\frac{1}{4}\beta^2m\big)$ we have $$\inf_{\bm{x}\in \mathcal{G}_{\hat{\delta}}}\inf_{E\in\mathcal{E}}\|\bm{\Phi}^E\bm{x}\|_\infty \geq 2\eta.$$Combined with (\ref{a4.9}) the desired (\ref{4.5}) holds. Thus, The proof is concluded. \hfill $\square$

     \vspace{1mm}

	   Lemma \ref{lem10} states that the number of near vanishing measurements does not exceed $\beta m$ with high probability. With a sufficiently small $\beta$, due to the boundedness $|\sign(\cdot)|\leq 1$,   the influence of the near vanishing measurements is expected to be    controllable. But   an actual attempt finds that,  we still need   the following Lemma \ref{lem11} to bound the operator norm of sub-matrices of $\bm{\Phi}$.

	  We point out that Lemma \ref{lem11} is an implication of a more in-depth result called Chevet's inequality \cite{chevet1977series,gordon1985some}, see also   \cite[Sec. 8.7]{vershynin2018high}. To be   self-contained,    we include an elementary    proof based on covering argument. 
     
    \begin{lem}
	  \label{lem11}
   \textcolor{black}{We suppose that $\bm{\Phi}\sim \mathcal{N}(0,1)+\mathcal{N}(0,1)\ii$,}  $\beta\in (0,1)$ is some given sufficiently small constant such that $\beta m$ is an integer. Let $\mathcal{E}_1 = \{\mathcal{S}\subset [m]:|\mathcal{S}| = \beta m\}$, $\mathcal{E}_2 = \{\mathcal{T}\subset [n]:|\mathcal{T}| = s\}$. Then with probability at least $1-\exp(-c\beta \log\big(\frac{72}{\beta}\big)m)$, for some absolute constant $C$ it holds that \begin{equation}
	      \label{4.9}
	      \begin{aligned}&\sup_{\mathcal{S}\in\mathcal{E}_1}\sup_{\mathcal{T}\in\mathcal{E}_2}\|\bm{\Phi}^\mathcal{S}_\mathcal{T}\| \\&~~~~~~\leq  C\sqrt{m\beta\log\Big(\frac{72}{\beta}\Big)+ s\log\Big(\frac{72n}{s}\Big)}.
	\end{aligned}  
   \end{equation}
	  \end{lem}

	  \noindent{\it Proof.}  First note that there exist $\bm{\hat{a}}\in \Sigma_{ \beta m,c}^{m,*}$, $\bm{\hat{b}}\in \Sigma^{n,*}_{s,c}$ such that  \begin{equation}
	    \label{3.39}  \begin{aligned}\sup_{\mathcal{S}\in\mathcal{E}_1}\sup_{\mathcal{T}\in \mathcal{E}_2} \|\bm{\Phi}_\mathcal{T}^\mathcal{S}\| &= \sup_{\bm{a}\in \Sigma^{m,*}_{ \beta m,c}}\sup_{\bm{b}\in \Sigma^{n,*}_{s,c}} \Re\big(\bm{a^*\Phi b}\big) \\&= \Re\big(\bm{\hat{a}^*\Phi \hat{b}}\big) .
     \end{aligned}
	  \end{equation}
  	  We construct $\mathcal{G}_{1}$ as a $\frac{1}{8}$-net of $\Sigma^{m,*}_{\beta m,c}$, $\mathcal{G}_{2}$ as a $\frac{1}{8}$-net of $\Sigma^{n,*}_{s,c}$, then by Lemma \ref{lem2} we can assume $|\mathcal{G}_1|\leq \big(\frac{72}{\beta}\big)^{2\beta m}$, $|\mathcal{G}_2|\leq \big(\frac{72n}{s}\big)^{2s}$.
      The remainder of this proof is a standard covering argument. For clarity we present it in two steps:   we first control $\Re(\bm{a}^*\bm{\Phi b})$ over $(\bm{a},\bm{b})\in \mathcal{G}_1\times \mathcal{G}_2$ {(Step 1)}; then,    we control the approximation error of the nets to prove the desired claim  {(Step 2)}.

   \noindent{\textbf{(Step 1.) Bound $|\Re(\bm{a^*\Phi b})|$ over $(\bm{a},\bm{b})\in \mathcal{G}_1\times \mathcal{G}_2$}}

   By rotational invariance of $\bm{\Phi}$, for fixed $\bm{a}\in \Sigma^{m,*}_{\beta m,c}$, $\bm{b}\in \Sigma^{n,*}_{s,c}$, $\Re(\bm{a^*\Phi b})\sim \mathcal{N}(0,1)$. Thus, for any $t>0$ (\ref{2.1}) gives $\mathbbm{P}\big(|\Re(\bm{a^*\Phi b})|\geq t\big)\leq 2\exp(-c_1t^2)$ for some absolute constant $c_1$, then a union bound gives \begin{equation}\nonumber\begin{aligned}
	      &\mathbbm{P}\Big(\sup_{\bm{a}\in \mathcal{G}_1}\sup_{\bm{b}\in\mathcal{G}_2}|\Re(\bm{a^*\Phi b})|\geq t\Big)\\&~~~~\leq 2\exp\Big(-c_1t^2+2\beta\log\Big(\frac{72}{\beta}\Big)\cdot m + 2s\log\Big(\frac{72n}{s}\Big)\Big).
       \end{aligned}
	  \end{equation}
	  Thus, we take $t = \frac{C}{2}\sqrt{\beta\log\big(\frac{72}{\beta}\big)m+s\log\big(\frac{72n}{s}\big)}$ 
for some sufficient large $C$, with probability at least $1- \exp\big(-c_2\beta \log\big(\frac{72}{\beta}\big)m\big)$, we have \begin{equation}
    \label{overnet1}
    \begin{aligned}
    &\sup_{\bm{a}\in \mathcal{G}_1}\sup_{\bm{b}\in\mathcal{G}_2}|\Re(\bm{a^*\Phi b})|\leq \frac{C}{2}\sqrt{\beta \log\Big(\frac{72}{\beta}\Big)m+s\log\Big(\frac{72n}{s}\Big)}.
    \end{aligned}
\end{equation}

\noindent{\textbf{(Step 2.) Strengthen (\ref{overnet1}) from $(\bm{a},\bm{b})\in \mathcal{G}_1\times \mathcal{G}_2$ to $(\bm{a},\bm{b})\in \Sigma_{\beta m,c}^{m,*}\times \Sigma^{n,*}_{s,c}$}}

Recalling (\ref{3.39}),
we can pick $\bm{\tilde{a}}\in \mathcal{G}_1$, $\bm{\tilde{b}}\in\mathcal{G}_2$ such that $\|\bm{\tilde{a}}-\bm{\hat{a}}\|\leq \frac{1}{8}$, $\|\bm{\tilde{b}}-\bm{\hat{b}} \|\leq \frac{1}{8}$, then we have  \begin{equation}
    \begin{aligned}\label{sparsede}
     &\sup_{\mathcal{S}\in\mathcal{E}_1}\sup_{\mathcal{T}\in \mathcal{E}_2} \|\bm{\Phi}_\mathcal{T}^\mathcal{S}\| =\Re\big(\bm{\hat{a}^*\Phi \hat{b}}\big)= \Re\big(\bm{\tilde{a}^*\Phi \tilde{b}}\big) \\&+ \Re \big(\bm{[\hat{a}}-\bm{\tilde{a}}]^*\bm{\Phi} \bm{\hat{b}}\big) + \Re\big(\bm{\tilde{a}}^*\bm{\Phi(\hat{b}-\tilde{b})}\big)\\
     &\stackrel{(i)}{\leq} \Re\big(\bm{\tilde{a}^*\Phi \tilde{b}}\big) + \|\bm{\hat{a}}-\bm{\tilde{a}}\|\cdot\Re\big([\bm{a}_1+\bm{a}_2]^*\bm{\Phi\hat{b}}\big) \\&~~~~~~~~~~~~~~~~+\|\bm{\hat{b}}-\bm{\tilde{b}}\|\cdot \Re\big(\bm{\tilde{a}}^*\bm{\Phi}(\bm{b}_1+\bm{b}_2)\big)
     \\&\stackrel{(ii)}{\leq}\Re\big(\bm{\tilde{a}^*\Phi \tilde{b}}\big)+ \frac{1}{8}\cdot 4\cdot \Re\big(\bm{\hat{a}^*\Phi \hat{b}}\big) \\&\leq  \sup_{\bm{a}\in \mathcal{G}_1}\sup_{\bm{b}\in\mathcal{G}_2}|\Re(\bm{a^*\Phi b})|+\frac{1}{2}\sup_{\mathcal{S}\in\mathcal{E}_1}\sup_{\mathcal{T}\in \mathcal{E}_2} \|\bm{\Phi}_\mathcal{T}^\mathcal{S}\|,
    \end{aligned}
\end{equation}
where $(i)$ follows from the decomposition $\frac{\bm{\tilde{a}}-\bm{\hat{a}}}{\|\bm{\tilde{a}}-\bm{\hat{a}}\|} = \bm{a}_1+\bm{a}_2$ for some $(\beta m)$-sparse $\bm{a}_1$, $\bm{a}_2$ satisfying $\|\bm{a}_1\|,\|\bm{a}_2\|\leq 1$,  and similarly $\frac{\bm{\hat{b}}-\bm{\tilde{b}}}{\|\bm{\hat{b}}-\bm{\tilde{b}}\|}=\bm{b}_1+\bm{b}_2$ for some $s$-sparse $\bm{b}_1,\bm{b}_2$ satisfying $\|\bm{b}_1\|,\|\bm{b}_2\|\leq 1$; $(ii)$ follows from (\ref{3.39}). Therefore, we obtain  $$\sup_{\mathcal{S}\in\mathcal{E}_1}\sup_{\mathcal{T}\in \mathcal{E}_2} \|\bm{\Phi}_\mathcal{T}^\mathcal{S}\|\leq2\sup_{\bm{a}\in \mathcal{G}_1}\sup_{\bm{b}\in\mathcal{G}_2}|\Re(\bm{a^*\Phi b})|,$$ by using (\ref{overnet1}) the desired claim is immediate. \hfill $\square$

    \subsection{The Parallel Part}
    With all above estimates in place, we are now in a position to present the proof for our main result. 
   In this subsection, our main goal is to bound the parallel part $\bm{u}_{\bm{x}}^\|$ and show $\sup_{\bm{x},\bm{u}}f^\|(\bm{x},\bm{u})$ in (\ref{aa3.17}) can be sufficiently small. To this end, we    plug in $\bm{z}= \sign(\bm{\Phi x})$, and $\|\bm{u}^\|_{\bm{x}}\|^2 = \big(\Re\big<\bm{u},\bm{x}\big>\big)^2$, then an initial attempt gives  (\ref{3.21}).\begin{figure*}[!t]
\normalsize
    \begin{equation}
    \label{3.21}
       \begin{aligned}
       &\sup_{\bm{x}\in \Sigma^{n,*}_{s,c}}\sup_{\bm{u}\in \Sigma^{n,*}_{4s,c}}f^\|(\bm{x},\bm{u}) = \sup_{\bm{x}\in \Sigma^{n,*}_{s,c}}\sup_{\bm{u}\in \Sigma^{n,*}_{4s,c}} \Big|\frac{\big[\Re\big<\sign(\bm{\Phi x}),\bm{\Phi u}\big>\big]^2}{\kappa^2m^2}-\big[\Re\big<\bm{u},\bm{x}\big>\big]^2\Big|\\&
       =  \sup_{\bm{x}\in \Sigma^{n,*}_{s,c}}\sup_{\bm{u}\in \Sigma^{n,*}_{4s,c}} \Big|\frac{1}{\kappa m}\Re\big<\sign(\bm{\Phi x}),\bm{\Phi u}\big>- \Re\big<\bm{u},\bm{x}\big>\Big|\cdot \Big| \frac{1}{\kappa m}\Re\big<\sign(\bm{\Phi x}),\bm{\Phi u}\big>+\Re\big<\bm{u},\bm{x}\big> \Big| \\
       &\leq  \left\{ \sup_{\bm{x}\in \Sigma^{n,*}_{s,c}}\sup_{\bm{u}\in \Sigma^{n,*}_{4s,c}} \Big|\frac{1}{\kappa m}\Re\big<\sign(\bm{\Phi x}),\bm{\Phi u}\big>- \Re\big<\bm{u},\bm{x}\big>\Big|\right\}  \cdot \left\{\frac{1}{m}\sup_{\bm{x}\in \Sigma^{n,*}_{s,c}}\sup_{\bm{u}\in \Sigma^{n,*}_{4s,c}}\|\sign(\bm{\Phi x})\|\cdot \|\bm{\Phi u}\|+1\right\}\\&\stackrel{(i)}{\lesssim} \sup_{\bm{x}\in \Sigma^{n,*}_{s,c}}\sup_{\bm{u}\in \Sigma^{n,*}_{4s,c}}\Big|\frac{1}{\kappa m}\Re\big<\sign(\bm{\Phi x}),\bm{\Phi u}\big>- \Re\big<\bm{u},\bm{x}\big>\Big|. 
       \end{aligned}
    \end{equation}
\hrulefill
\vspace*{4pt}
\end{figure*}
    Note that   $(i)$ is due to $\frac{1}{m}\sup_{\bm{x},\bm{u}} \|\sign(\bm{\Phi x})\|\cdot \|\bm{\Phi u}\| \leq\sup_{\bm{x},\bm{u}}\|\frac{\bm{\Phi u}}{\sqrt{m}}\| = O(1)$, which holds with high probability as long as $m \gtrsim s\log\big(\frac{n}{s}\big)$ (Lemma \ref{lem3}).

    We  need to show the last line of (\ref{3.21}) is sufficiently small. Compared with the local sign-product embedding property established in   \cite[Lem. 5.4]{jacques2021importance}, what we are going to show is indeed a global sign-product embedding property that holds uniformly for all sparse $\bm{x}$. \textcolor{black}{While \cite{jacques2021importance} and this work consider complex $\bm{\Phi}$, the real counterpart of sign-product embedding property (under $\bm{\Phi}\in \mathbb{R}^{m\times n}$) has been established in \cite{plan2012robust}, see Remark \ref{globalrem} for more discussions.} Besides, a crucial technical change is to only take the real part of the inner product, which stems from the removal the second constraint in (\ref{add830}). Otherwise, additional bias would arise and $\sup_{\bm{x},\bm{u}}f^\|(\bm{x},\bm{u})$ cannot be bounded close to $0$ (see Remark \ref{rem1} for details).

  {By using $\|\bm{\Phi x}\|_1=\big<\sign(\bm{\Phi x}),\bm{\Phi x}\big>$ we first note the following equality
   \begin{equation}
   \begin{aligned}\label{addclear1}
       &\frac{1}{\kappa m}\Re \big<\sign(\bm{\Phi x}),\bm{\Phi u}\big>-\Re\big<\bm{u},\bm{x}\big>\\&=\Re\big<\bm{u},\bm{x}\big>\Big(\frac{\|\bm{\Phi x}\|_1}{\kappa m}-1\Big)\\&+\frac{1}{\kappa m}\Re\big<\sign(\bm{\Phi x}),\bm{\Phi u_x}^\bot\big>,
       \end{aligned}
   \end{equation}}
   which allows us to decompose the last term of (\ref{3.21}) into  
    \begin{equation}
        \begin{aligned}
        \label{aa3.40}
        & \sup_{\bm{x}\in \Sigma^{n,*}_{s,c}}\sup_{\bm{u}\in \Sigma^{n,*}_{4s,c}} \Big|\frac{1}{\kappa m}\Re\big<\sign(\bm{\Phi x}),\bm{\Phi u}\big>- \Re\big<\bm{u},\bm{x}\big>\Big|\\&\leq  \sup_{\bm{x}\in \Sigma^{n,*}_{s,c}}\sup_{\bm{u}\in \Sigma^{n,*}_{4s,c}} \underbrace{|\Re\big<\bm{u},\bm{x}\big>|\Big|\frac{\|\bm{\Phi x}\|_1}{\kappa m} - 1\Big|}_{f^\|_1(\bm{x},\bm{u})}\\&+\sup_{\bm{x}\in \Sigma^{n,*}_{s,c}}\sup_{\bm{u}\in \Sigma^{n,*}_{4s,c}}\underbrace{ \frac{1}{\kappa m}\Big|\Re\big<\sign(\bm{\Phi x}),\bm{\Phi u_x^\bot}\big>\Big| }_{f^\|_2(\bm{x},\bm{u})}
        \end{aligned}
    \end{equation}
  By previous development, the bound for $\sup_{\bm{x},\bm{u}}f^\|_1(\bm{x},\bm{u})$ is immediate. 
    
    \begin{coro}
    \label{coro1}
      Given any $\delta>0$, when $m\geq \frac{C s}{\delta^2}\log\big(\frac{n}{\delta s}\big)$ for some absolute constant $C$, with probability at least $1-2\exp(-c\delta^2m)$ we have \begin{equation}\nonumber
          \sup_{\bm{x}\in \Sigma^{n,*}_{s,c}}\sup_{\bm{u}\in \Sigma^{n,*}_{4s,c}}f_1^\|(\bm{x},\bm{u})\leq \delta.
      \end{equation}
    \end{coro}
    
    \noindent{\it Proof.} Because $(\bm{x},\bm{u})\in \Sigma^{n,*}_{s,c}\times \Sigma^{n,*}_{4s,c}$, we have $|\Re\big<\bm{u},\bm{x}\big>|\leq 1$. To prove the claim, it remains to invoke Lemma \ref{lem4}. \hfill $\square$
    
    \vspace{1mm}


   
   To deal with $\sup_{\bm{x},\bm{u}}f_2^\|(\bm{x},\bm{u})$,
   an approach similar to Lemma \ref{lem4}, \ref{lem10}, \ref{lem11} is in order. That is, we first study a fixed $(\bm{x},\bm{u})$ in Lemma \ref{lemma6}, and then apply a covering argument in Lemma \ref{lem7}.

    \begin{lem}
    \label{lemma6}
       Fix $\bm{x}\in\Sigma^{n,*}_{s,c}$, $\bm{u}\in\Sigma^{n,*}_{4s,c}$, for some absolute constant $c>0$ we have \begin{equation}
           \label{3.22}\begin{aligned}
           &\mathbbm{P}\Big(\Big|\frac{1}{\kappa m}\Re\big<\sign(\bm{\Phi x}),\bm{\Phi u_x^\bot}\big>\Big|\geq t\Big)\\&~~~~~~~~~~~~
           \leq 2\exp(-cmt^2),~~\forall~t>0.
           \end{aligned}
       \end{equation}
    \end{lem}
    
    \noindent{\it Proof.}  For a fixed $\bm{x}$, there exists a unitary matrix $\bm{P}$ (i.e., $\bm{PP^*} = \bm{I}_n$) such that $\bm{Px} = \bm{e}_1$. Furthermore, we let $\bm{\widetilde{\Phi}} = \bm{\Phi P^*}$, then $\bm{\widetilde{\Phi}}$ and $\bm{\Phi P^*}$ have the same distribution. We  divide it into two blocks $\bm{\widetilde{\Phi}} = [\bm{g},\bm{G}]$ with $\bm{g}\in\mathbb{C}^{m\times 1}$ and $\bm{G}\in \mathbb{C}^{m\times (n-1)}$. We set $\bm{\tilde{g}} = \frac{\sign(\bm{g^*})}{\sqrt{m}}\bm{G}\in\mathbb{C}^{1\times (n-1)}$. Since with full probability $\|\frac{\sign(\bm{g^*})}{\sqrt{m}}\| = 1$, and \textcolor{black}{$\bm{{g}}$ is independent of $\bm{G}$}, by conditionally on $\bm{g}$,  $\bm{\tilde{g}}$ has entries i.i.d. distributed as $\mathcal{N}(0,1)+\mathcal{N}(0,1)\ii$ almost surely. In addition, we define $\bm{v}: = \bm{Pu_x^\bot}$ and denote its $i$-th entry by $v_i$. Since $\bm{Pu_x^\bot} = \bm{Pu}-\Re\big<\bm{Pu},\bm{Px}\big>\bm{Px}= \bm{Pu}-\Re\big<\bm{Pu},\bm{e}_1\big>\bm{e}_1$, we have $\Re( v_1) = 0$.  Now it follows that\begin{equation}
        \begin{aligned}
        \label{3.23}
        &\Re \big<\sign(\bm{\Phi x}),\bm{\Phi u_x^\bot}\big> = \Re\big<\sign(\bm{\Phi P^*Px}),\bm{\Phi P^* P u_x^\bot}\big>\\& = \Re\big<\sign(\bm{\widetilde{\Phi}e}_1),\bm{\widetilde{\Phi}v}\big> =\Re\big<\sign(\bm{g}),\bm{g}v_1+\bm{Gv}^{[2:n]}\big> \\
        &= \Re\big(\|\bm{g}\|_1v_1+\sign(\bm{g}^*)\bm{G}\bm{v}^{[2:n]}\big) =\sqrt{m}\cdot \Re(\bm{\tilde{g}v}^{[2:n]}).
        \end{aligned}
    \end{equation}
    Note that $\Re(\bm{\tilde{g}w})\sim \mathcal{N}(0,1)$ if $\|\bm{w}\|=1$, and evidently we have $\|\bm{v}^{[2:n]}\|\leq 1$, and hence $\frac{1}{\kappa m}\Re\big<\sign(\bm{\Phi x}),\bm{\Phi u_x^\bot}\big> = \frac{1}{\kappa\sqrt{m}}\Re(\bm{\tilde{g}}\bm{v}^{[2:n]})$ has sub-Gaussian norm bounded by $\frac{C}{\sqrt{m}}$ for some $C$. Then the result follows immediately from (\ref{2.1}).\hfill $\square$
    
    \vspace{1mm}
     
    \begin{rem}
    \label{rem1}\textcolor{black}{Taking Corollary \ref{coro1}, Lemma \ref{lemma6} and (\ref{aa3.40}) collectively, at this moment, we can already conclude that $\big|\frac{1}{\kappa m}\Re\big<\sign(\bm{\Phi x}),\bm{\Phi u}\big>-\Re\big<\bm{u},\bm{x}\big>\big|$ concentrates close to $0$ for a fixed $(\bm{x},\bm{u})\in \Sigma^{n,*}_{s,c}\times \Sigma^{n,*}_{4s,c}$. This is the cornerstone for  proving the subsequent global sign-product embedding property (Corollary \ref{signpro}). Further, we comment that removing the second redundant constraint of (\ref{add830}) seems quite necessary for our development. Otherwise, the parallel part we want to bound would become $\sup_{\bm{x},\bm{u}}f_{*}^{\|}(\bm{x},\bm{u})=\sup_{\bm{x},\bm{u}}\big|\frac{1}{\kappa^2m^2}|\big<\sign(\bm{\Phi x}),\bm{\Phi u}\big>|^2-\|\bm{u}_{\bm{x}}^{\|}\|^2\big|$, which by treatments similar to (\ref{3.21}), (\ref{aa3.40}) leads us to bound $I^*_{\bm{x},\bm{u}}:=\big|\frac{1}{\kappa m}\big<\sign(\bm{\Phi x}),\bm{\Phi u_x}^\bot\big>\big|$ for fixed $(\bm{x},\bm{u})$ first. The issue is that, unlike  in Lemma \ref{lemma6}, $I^*_{\bm{x},\bm{u}}$ is not close to $0$. This can be seen from (\ref{3.23}) which, without taking the real part, contains a   bias term of $\|\bm{g}\|_1 v_1$ that could largely deviate from $0$, as $v_1$ is non-zero in general.}
    \end{rem}
    
    \begin{lem}  
       \label{lem7}
      Given any sufficiently small constant $\delta>0$ and threshold $\eta$ (for defining near vanishing measurement), if \begin{equation}
       \label{3.24}
           m\geq C_1s\max\Big\{\frac{1}{\eta^2}\log\Big(\frac{n^2\log({mn})}{  \eta^2s}\Big),\frac{1}{\delta^2}\log\Big(\frac{n}{\delta \eta s}\Big)\Big\}
       \end{equation} 
       for some sufficiently large $C_1$, then with probability at least $1-2(mn)^{-9}-4\exp(-c_1\delta^2m)- \exp(-c_2\eta^2 m )$, there exists some absolute constant $C_2$ such that\begin{equation}\nonumber
           \sup_{\bm{x}\in \Sigma^{n,*}_{s,c}}\sup_{\bm{u}\in \Sigma^{n,*}_{4s,c}} f_2^\|(\bm{x},\bm{u}) \leq C_2(\delta + \sqrt{\eta}).
       \end{equation}
    \end{lem}
    
    \noindent{\it Proof.} By the boundedness of $\Sigma^{n,*}_{ s,c} \times \Sigma^{n,*}_{4s,c}$ we can find $(\bm{\hat{x}},\bm{\hat{u}})\in\Sigma^{n,*}_{ s,c}\times \Sigma^{n,*}_{4s,c}$ such that 
    \begin{equation}
    \label{suptopoint}
        \sup_{\bm{x}\in \Sigma^{n,*}_{s,c}}\sup_{\bm{u}\in \Sigma^{n,*}_{4s,c}}f_2^\|(\bm{x},\bm{u}) < f_2^\|(\bm{\hat{x}},\bm{\hat{u}})+\delta.
    \end{equation}  
    We divide the proof into two steps.

    \noindent{\textbf{(Step 1.) Control $f_2^\|(\bm{x},\bm{u})$ over discrete nets}}

    For a given $\delta$ we construct $\mathcal{G}_\delta\subset \Sigma^{n,*}_{4s,c}$ as a $\delta$-net, $\mathcal{G}_{\delta\eta} \subset \Sigma^{n,*}_{s,c}$ as a $(\eta\delta)$-net, then Lemma \ref{lem2} allows us to suppose $|\mathcal{G}_\delta|\leq \big(\frac{9n}{4\delta s}\big)^{4s}$ and $|\mathcal{G}_{\eta\delta}| \leq \big(\frac{9n}{\eta \delta s}\big)^{s}$.  We first extend (\ref{3.22}) to $(\bm{x},\bm{u})\in \mathcal{G}_{\eta\delta}\times \mathcal{G}_\delta$, which by a union bound gives\begin{equation}
        \label{3.25}\begin{aligned}
       & \mathbbm{P}\left(\sup_{\bm{x}\in \mathcal{G}_{\eta\delta}}\sup_{\bm{u}\in \mathcal{G}_\delta}\Big| \frac{1}{\kappa m}\Re\big<\sign(\bm{\Phi x}),\bm{\Phi u_x^\bot} \big>\Big|\geq t\right)\\&\leq 2\exp\Big(-cmt^2+4s\log\Big(\frac{9n}{4\delta s}\Big)+s\log\Big(\frac{9n}{\eta \delta s}\Big)\Big),
        \end{aligned}
    \end{equation}
    where $t>0$.
    Setting $t = \delta$, then when (\ref{3.24}) holds for sufficiently large $C_1$, it holds with probability at least $1-2\exp(-c_1\delta^2m)$ that \begin{equation}
        \label{overnet}\begin{aligned}
        &\sup_{\bm{x}\in\mathcal{G}_{\eta\delta}}\sup_{\bm{u}\in \mathcal{G}_\delta}f_2^\|(\bm{x},\bm{u})\\=&
        \sup_{\bm{x}\in\mathcal{G}_{\eta\delta}}\sup_{\bm{u}\in \mathcal{G}_\delta}\Big|\frac{1}{\kappa m}\Re\big<\sign(\bm{\Phi x}),\bm{\Phi u_x^\bot}\big>\Big|\leq \delta.\end{aligned}
    \end{equation}

    \noindent{\textbf{(Step 2.) Strengthen (\ref{overnet}) to a uniform bound for $(\bm{x},\bm{u})\in \Sigma^{n,*}_{s,c}\times \Sigma^{n,*}_{4s,c}$}}

    Due to (\ref{suptopoint}) we only need to show $f_2^\|(\bm{\hat{x},\hat{u}}) = O(\delta + \sqrt{\eta})$. For approximation, we can pick $(\bm{\tilde{x}},\bm{\tilde{u}})\in
    \mathcal{G}_{\eta\delta}\times \mathcal{G}_\delta$ such that $\|\bm{\tilde{x}}- \bm{\hat{x}}\|\leq \eta \delta$, $\|\bm{\tilde{u}}-\bm{\hat{u}}\|\leq \delta$. For the given sufficiently small $\eta$, we can find $4\eta<\beta = O(\eta)$ such that $\beta m $ is an integer. Then   Lemma \ref{lem10} states that under the sample size (\ref{3.24}), with probability at least $1-\exp(-c_2\eta^2 m ) - 2(mn)^{-9}$, we have \begin{equation}
        \label{boundnearvani}
        |\mathcal{J}_{\bm{\hat{x}}}|\leq \sup_{\bm{x}\in \Sigma^{n,*}_{s,c}}|\mathcal{J}_{\bm{x}}|<\beta m.
    \end{equation}  Now we can proceed as in (\ref{3.28add}).
     \begin{figure*}[!t]
\normalsize
\begin{equation}
        \begin{aligned}
        \label{3.28add}
       \begin{aligned}
                \frac{1}{\kappa m}\big|\Re\big<\sign(\bm{\Phi \hat{x}}),\bm{\Phi \hat{u}^\bot_{\hat{x}}}\big>\big|  \leq  & \frac{1}{\kappa m}\big|\Re\big<\sign(\bm{\Phi \tilde{x}}),\bm{\Phi \tilde{u}^\bot_{\tilde{x}}}\big>\big|  +\frac{1}{\kappa m}\Big|\Re\big<\sign(\bm{\Phi \hat{x}}),\bm{\Phi \hat{u}^\bot_{\hat{x}}}\big>-\Re\big<\sign(\bm{\Phi \tilde{x}}),\bm{\Phi \tilde{u}^\bot_{\tilde{x}}}\big>\Big|\\
               \stackrel{(i)}{\leq} &\delta + \frac{1}{\kappa m}\big|\Re\big<\sign(\bm{\Phi\hat{x}})-\sign(\bm{\Phi\tilde{x}}),\bm{\Phi \hat{u}^\bot_{\hat{x}}}\big>\big| +\frac{1}{\kappa m}\big|\Re\big<\sign(\bm{\Phi\tilde{x}}),\bm{\Phi (\bm{\hat{u}^\bot_{\hat{x}}}-\bm{\tilde{u}^\bot_{\tilde{x}}})}\big>\big| \\
               \stackrel{(ii)}{\lesssim} &\delta + \frac{1}{\sqrt{m}}\|\sign(\bm{\Phi \hat{x}})-\sign(\bm{\Phi \tilde{x}})\| + \|\bm{\hat{u}^\bot_{\hat{x}}}-\bm{\tilde{u}^\bot_{\tilde{x}}}\| .
       \end{aligned}  
        \end{aligned}
    \end{equation}
\hrulefill
\vspace*{4pt}
\end{figure*}
    Note that   $(i)$ is due to (\ref{overnet}), and in $(ii)$ we use $\Re\big<\bm{a},\bm{b}\big>\leq \|\bm{a}\|\|\bm{b}\|$ and \begin{equation}
        \label{lem5impli}
        \frac{1}{m}\sup_{\bm{u}\in\Sigma^{n,*}_{10s,c}}\|\bm{\Phi u}\|^2 = O(1).
    \end{equation} Note that by Lemma \ref{lem3} (\ref{lem5impli})   holds with probability exceeding $1-2\exp(-c_3m)$, and it is applicable because $\bm{\hat{u}^\bot_{\hat{x}}},\bm{\tilde{u}^\bot_{\tilde{x}}}\in \Sigma^{n}_{5s,c}$, $\|\bm{\hat{u}^\bot_{\hat{x}}}\|\leq 1,\|\bm{\tilde{u}^\bot_{\tilde{x}}}\|\leq 1$. According to (\ref{3.28add}) we divide the remaining proof into two steps.

    \noindent{\textbf{(Step 2.1.) Bound $m^{-1/2}\|\sign(\bm{\Phi\hat{x}})-\sign(\bm{\Phi\tilde{x}})\|$}}

We let $ \mathcal{J}_{\bm{\hat{x}}}^c : = [m]\setminus \mathcal{J}_{\bm{\hat{x}}}$, then for $k\in  \mathcal{J}_{\bm{\hat{x}}}^c$ we have $|\bm{\Phi}^*_k\bm{\hat{x}}|\geq \eta$.  We decompose the term according to 
    whether a measurement belongs to $\mathcal{J}_{\bm{\hat{x}}}$ and then apply the estimate in Lemma \ref{signconti}, it yields
    \begin{equation}
        \begin{aligned}
        \label{aa347}
         &\frac{1}{\sqrt{m}}\|\sign(\bm{\Phi \hat{x}}) - \sign(\bm{\Phi\tilde{x}})\|   \\\leq &\frac{1}{\sqrt{m}}\|\sign(\bm{\Phi}^{\mathcal{J}_{\bm{\hat{x}}}} \bm{\hat{x}}) - \sign(\bm{\Phi}^{\mathcal{J}_{\bm{\hat{x}}}}\bm{\tilde{x}})\|  \\&~~~~~~ +\frac{1}{\sqrt{m}}\|\sign(\bm{\Phi}^{\mathcal{J}^c_{\bm{\hat{x}}}}\bm{ \hat{x}}) - \sign(\bm{\Phi}^{\mathcal{J}^c_{\bm{\hat{x}}}}\bm{\tilde{x}})\| \\
         \stackrel{(i)}{\leq} & \frac{2\sqrt{|\mathcal{J}_{\bm{\hat{x}}}|}}{\sqrt{m}} + \frac{2\|\bm{\Phi}^{\mathcal{J}^c_{\bm{\hat{x}}}}(\bm{\hat{x}}-\bm{\tilde{x}})\|}{\eta\sqrt{m}}\\\stackrel{(ii)}{\lesssim} & \sqrt{\beta} + \frac{\|\bm{\hat{x}}-\bm{\tilde{x}}\|}{\eta}\leq \sqrt{\beta} + \delta.
        \end{aligned}
    \end{equation}
    Note that $(i)$ is due to Lemma \ref{signconti}, specifically we use the bound "$2$" for the measurements in $\mathcal{J}_{\bm{\hat{x}}}$, while the bound "$\frac{2|a-b|}{\max\{|a|,|b|\}}$" for the measurements in  $\mathcal{J}^c_{\bm{\hat{x}}}$. Then, $(ii)$ is because $|\mathcal{J}_{\bm{\hat{x}}}|<\beta m$ (\ref{boundnearvani}) and $\frac{1}{\sqrt{m}}\|\bm{\Phi}^{\mathcal{J}^c_{\bm{\hat{x}}}}(\bm{\hat{x}}-\bm{\tilde{x}})\|\leq \frac{1}{\sqrt{m}}\|\bm{\Phi}(\bm{\hat{x}}-\bm{\tilde{x}})\|=O(1)$ (Lemma \ref{lem3}). 

  \noindent{\textbf{(Step 2.2.) Bound $\|\bm{\hat{u}}_{\bm{\hat{x}}}^\bot -\bm{\tilde{u}}_{\bm{\tilde{x}}}^\bot \|$}}

This is a  more standard estimation: \begin{equation}
        \begin{aligned}
        \label{3.30}
               & \|\bm{\hat{u}^\bot_{\hat{x}}}-\bm{\tilde{u}^\bot_{\tilde{x}}}\| = \|\bm{\hat{u}}-\Re\big<\bm{\hat{u}},\bm{\hat{x}}\big>\bm{\hat{x}} - \bm{\tilde{u}}+\Re\big<\bm{\tilde{u}},\bm{\tilde{x}}\big>\bm{\tilde{x}}\|\\
               &\leq \|\bm{\hat{u}}-\bm{\tilde{u}}\|+|\Re\big<\bm{\tilde{u}},\bm{\tilde{x}}\big>|\cdot \|\bm{\hat{x}}-\bm{\tilde{x}}\|\\&~~~~~~~~~ + \|\bm{\hat{x}}\|\cdot \big(|\Re\big<\bm{\tilde{u}}-\bm{\hat{u}},\bm{\tilde{x}}\big>|+|\Re\big<\bm{\hat{u}},\bm{\tilde{x}}-\bm{\hat{x}}\big>|\big)\\
               &\leq 2\cdot\big(\|\bm{\hat{u}}-\bm{\tilde{u}}\| + \|\bm{\hat{x}}-\bm{\tilde{x}}\|\big) \leq 4\delta.
        \end{aligned}
    \end{equation}
    Since $\beta = O(\eta)$, putting (\ref{aa347}), (\ref{3.30}) into (\ref{3.28add}) concludes the proof. \hfill $\square$
    
    \vspace{1mm}
    
    Recall that the left-hand side of (\ref{aa3.40}) is divided into $f^\|_1(\bm{x},\bm{u})$ and $f^\|_2(\bm{x},\bm{u})$ for clarity, and we have derived the bound for these two terms. Now, we further put them together and present a   global sign-product embedding property as the following corollary, which may be interesting on its own right (Remark \ref{globalrem}). 
    \begin{coro}
    \label{signpro}
    {\rm (Global Sign-Product Embedding Property)\textbf{.}} Assume $\bm{\Phi}\sim \mathcal{N}^{m\times n}(0,1)+\mathcal{N}^{m\times n}(0,1)\ii$, given any $\delta>0$, if $m \geq \frac{Cs}{\delta^4} \log\big(\frac{n^2\log(mn)}{\delta^4 s}\big)$ holds for sufficiently large $C$, with probability at least $1-2(mn)^{-9}-7\exp(-c\delta^4m)$ we have \begin{equation}\nonumber\begin{aligned}
       &\Big|\frac{1}{\kappa m}\Re\big<\sign(\bm{\Phi u}),\bm{\Phi v}\big>-\Re\big<\bm{u},\bm{v}\big>\Big|\leq \delta \|\bm{v}\|,\\& ~~~~~~~~~~~~~~~~~\forall~ \bm{u}\in\Sigma^{n,*}_{s,c},\bm{v}\in \Sigma^n_{s,c}.\end{aligned}
    \end{equation}
    \end{coro}
    \begin{rem}
       \label{globalrem}
       From a geometry perspective, the sign-product embedding property states that the projection length of a sparse $\bm{v}$ onto a unit sparse $\bm{u}$, can be uniformly (over both $\bm{u},\bm{v}$) encoded into the projection length of $\frac{1}{\sqrt{m}}\bm{\Phi v}$ onto the normalized phase-only measurements $\frac{1}{\sqrt{m}}\sign(\bm{\Phi u})$, up to a rescaling of $\frac{1}{\kappa}$. Corollary \ref{signpro} presents a twofold extension of  \cite[Lem. 5.4]{jacques2021importance}, that is, to complex $\bm{u},\bm{v}$ and to uniformity of $\bm{u}$. However, the uniformity of $\bm{u}$ comes at the cost of worse dependence $\delta^{-4}$ on $\delta$, as contrasted to $\delta^{-2}$ in \cite[Lem. 5.4]{jacques2021importance}. Note that similar remark was already offered for     Theorem \ref{theorem1} (Remark \ref{pricetoget}). 
    \end{rem}

\begin{rem}
    \label{embeenco}
     Various embedding/encoding results were developed in the literature of 1-bit compressive sensing, e.g., the binary $\epsilon$-stable embedding (B$\epsilon$SE)     \cite{jacques2013robust}, Hamming cube encoding   \cite{plan2013one}, \textcolor{black}{while the (real) sign-product embedding property (SPE) in \cite{plan2012robust} appeared to be closest to \cite[Lem. 5.4]{jacques2021importance} and our Corollary \ref{signpro}. Specifically, let $\bm{\Phi_r}\sim\mathcal{N}^{m\times n}(0,1)$, $\bm{u},\bm{v}$ be some low-complexity signal, and $\lambda$ be a properly chosen scaling, it was proved in \cite{plan2012robust} that $\big|\frac{1}{\lambda m}\big<\sign(\bm{\Phi_r u}),\bm{\Phi v}\big>-\big<\bm{u},\bm{v}\big>\big|$ can be bounded close to $0$ for a fixed $\bm{u}$ and all $\bm{v}$  \cite[Prop. 4.2]{plan2012robust} (i.e., local SPE), or even uniformly for all $(\bm{u},\bm{v})$ \cite[Prop. 4.3, Lem. 6.4(1)]{plan2012robust} (i.e., global SPE). Interestingly, in \cite{plan2012robust} the global SPE also displays worse dependence on the embedding distortion than the local SPE.} 
\end{rem}
    \subsection{The Orthogonal Part}
    In this subsection we switch to the orthogonal part $\sup_{\bm{x}\in\Sigma^{n,*}_{s,c}}\sup_{\bm{u}\in \Sigma^{n,*}_{4s,c}} f^\bot(\bm{x},\bm{u})$.  We first calculate its expectation and show the concentration property for  fixed $(\bm{x},\bm{u})\in \mathcal{K} ^* \times \Sigma^{n,*}_{4s,c}$. The main technique is similar to the proof of Lemma \ref{lemma6}, which is to take advantage of the rotational invariance of $\bm{\Phi}$.
    \begin{lem}
    \label{lem8}
       Fix $\bm{x}\in \Sigma^{n,*}_{s,c}$, $\bm{u}\in \Sigma^{n,*}_{4s,c}$, it holds that \begin{equation}
       \label{3.31}
           \mathbbm{E}\Big(\frac{1}{m}\|\Im\big(\diag(\bm{z^*})\bm{\Phi u}\big)\|^2\Big) = \|\bm{u}_{\bm{x}}^\bot\|^2 + |\Im\big<\bm{x},\bm{u}\big>|^2.
       \end{equation}
         Moreover,  there exists $c>0$ such that for all $t>0$ \begin{equation}
             \label{3.32}
             \begin{aligned}
             &\mathbbm{P}\Big(\Big|\frac{1}{m}\|\Im\big(\diag(\bm{z}^*)\bm{\Phi u}\big) \|^2 \\&~~~~~~~~~~~~~- (\|\bm{u}_{\bm{x}}^\bot\|^2 + |\Im\big<\bm{x},\bm{u}\big>|^2)\Big|\geq t\Big)\\&\leq 2\exp(-cm\min\{t,t^2\}).
             \end{aligned}
         \end{equation}
    \end{lem}
    \noindent
    {\it Proof.} Similar to the proof of Lemma \ref{lemma6}, we can find a unitary matrix $\bm{P}$ such that $\bm{Px} = \bm{e}_1$, then $\bm{v}=[v_i] = \bm{Pu_x^\bot}$ satisfies $\Re v_1 = 0$, \textcolor{black}{which follows from the calculation $\bm{Pu_x}^\bot =\bm{P}\bm{u}-\Re\big<\bm{u},\bm{x}\big>\bm{Px}=\bm{Pu}-\Re\big<\bm{Pu},\bm{Px}\big>\bm{Px}=\bm{Pu}-\Re\big<\bm{Pu},\bm{e}_1\big>\bm{e}_1$}. We further define $\bm{\gamma} = [\gamma_i] := \bm{P\Phi}_1\in \mathbb{R}^{n\times 1}$, whose entries are i.i.d. distributed as $\mathcal{ N}(0,1)+\mathcal{N}(0,1)\ii$. For clarity we divide the remainder of this proof into two steps.

\noindent{\textbf{(Step 1.) Prove the expectation (\ref{3.31}) via calculation}}

    Recall that $\bm{\Phi}^*_1$ is the first row of $\bm{\Phi}$,  some algebra gives (\ref{expcal}).    
 \begin{figure*}[!t]
\normalsize
\begin{equation}
         \begin{aligned}       \label{expcal}       &\mathbbm{E}\Big(\frac{1}{m}\|\Im\big(\diag(\bm{z^*})\bm{\Phi u}\big)\|^2\Big) \stackrel{(i)}{=}\mathbbm{E}\Big(\frac{1}{m}\|\Im\big(\diag(\bm{z^*})\bm{\Phi u^\bot_{x}}\big)\|^2\Big)\\ =& \mathbbm{E} \big[\Im\big( \overline{\sign(\bm{\Phi_1^*x})}\cdot \bm{\Phi}_1^*\bm{u^\bot_{x}}\big)\big]^2\stackrel{(ii)}{=} \mathbbm{E}\big[\Im\big(\sign(\gamma_1)\cdot (\bm{\gamma^*v})\big)\big]^2\\
               =&\mathbbm{E}\Big[|\gamma_1|v_1^\Im + \sum_{j=2}^n \Re\big(\sign(\gamma_1)\overline{\gamma_j}\big)\cdot v_j^\Im+\sum_{j=2}^n \Im\big(\sign(\gamma_1)\overline{\gamma_j}\big)\cdot v_j^\Re\Big]^2\\
               \stackrel{(iii)}{=}& \big(\mathbbm{E}|\gamma_1|^2\big)\cdot (v_1^\Im)^2 + \sum_{j=2}^n \big(\mathbbm{E}\big[\Re (\sign(\gamma_1)\overline{\gamma_j} )\big]^2\big)\cdot  (v_j^\Im)^2 + \sum_{j=2}^n\big(\mathbbm{E}\big[\Im(\sign(\gamma_1)\overline{\gamma_j})\big]^2\big)\cdot(v_j^\Re)^2\\
                = &2|v_1|^2 + \sum_{j=2}^n |v_j|^2 = \|\bm{v}\|^2 + |v_1|^2 = \|\bm{u}_{\bm{x}}^\bot\|^2 + |\Im\big<\bm{x},\bm{u}\big>|^2.
        \end{aligned} 
    \end{equation}
\hrulefill
\vspace*{4pt}
\end{figure*}
    This displays (\ref{3.31}). Note that we use (\ref{3.11}) in $(i)$, we plug in $\bm{\gamma}=\bm{P\Phi}_1, \bm{Px}=\bm{e}_1,\bm{Pu_x}^\bot = \bm{v}$ in $(ii)$, and $(iii)$ is because the expectation of  the cross terms  is zero.

\noindent{\textbf{(Step 2.)   Show the concentration inequality (\ref{3.32})}}
    
      Note that $$\frac{1}{m}\|\Im\big(\diag(\bm{z^*})\bm{\Phi u}\big)\|^2 = \frac{1}{m}\sum_{k=1}^m\big[\Im\big(\overline{\sign(\bm{\Phi}_k^*\bm{x})}\cdot \bm{\Phi}_k^*\bm{u^\bot_{x}}\big)\big]^2$$
    is the mean of $m$ independent copies of $\big[\Im\big(\overline{\sign(\bm{\Phi}_1^*\bm{x})}\cdot \bm{\Phi}_1^*\bm{u^\bot_{x}}\big)\big]^2$. Then we use (\ref{830add2}) and the definition of sub-Gaussian norm to estimate\begin{equation}
        \begin{aligned}\nonumber
            &\big\|\big[\Im\big(\overline{\sign(\bm{\Phi}_1^*\bm{x})}\cdot \bm{\Phi}_1^*\bm{u^\bot_{x}}\big)\big]^2\big\|_{\psi_1}\\\leq& \big\| \Im\big(\overline{\sign(\bm{\Phi}_1^*\bm{x})}\cdot \bm{\Phi}_1^*\bm{u^\bot_{x}}\big)\big\|^2_{\psi_2}\\ \leq& \big\|\bm{\Phi}_1^*\bm{u_x^\bot}\big\|^2_{\psi_2} = O(1).
        \end{aligned}
    \end{equation} 
    Therefore, we can invoke Bernstein's inequality (Lemma \ref{bern}) to obtain (\ref{3.32}).\hfill $\square$
    \begin{rem}
    \label{remark2}
        Lemma \ref{lem8} exhibits an essential difference between the real case and the complex case. That is, while for the real case in \cite{jacques2021importance} $\Im\big<\bm{x},\bm{u}\big>=0$ and hence in expectation $\frac{1}{m}\|\Im\big(\diag(\bm{z}^*)\bm{\Phi u}\big)\|^2$ exactly provides the orthogonal part $\|\bm{u}^\bot_{\bm{x}}\|^2$, in the complex case there appears a bias term $|\Im\big<\bm{x},\bm{u}\big>|^2$ that may rise up to $1$, e.g., when $\bm{x} = \ii\cdot \bm{u}$. Thus, if letting the scaling factor $\hat{t} = 1$ in (\ref{aa3.14}), $\sup_{\bm{x}\in \Sigma^{n,*}_{s,c}}\sup_{\bm{u}\in \Sigma^{n,*}_{4s,c}}f^\bot(\bm{x},\bm{u})$ can never be bounded below $1$, which even does not guarantee the identifiability of $\bm{x}$, not to mention (\ref{aa3.17}). We shall see shortly in the proof of Theorem \ref{theorem1} that, a careful choice of $\hat{t}$ can yield $\sup_{\bm{x}\in \Sigma^{n,*}_{s,c}}\sup_{\bm{u}\in \Sigma^{n,*}_{4s,c}} f^\bot(\bm{x},\bm{u})\leq \frac{1}{3}+\delta$ for any pre-specified $\delta\in (0,1)$, thus fulfilling (\ref{aa3.17}).
    \end{rem}
     We introduce the shorthand $\hat{f}^\bot(\bm{x},\bm{u})$ that contains the bias term $|\Im\big<\bm{x},\bm{u}\big>|^2$ 
     \begin{equation}
     \label{aa3.54}\begin{aligned}
         \hat{f}^\bot(\bm{x},\bm{u}) :&= \Big|\frac{1}{m}\|\Im\big(\diag(\bm{z}^*)\bm{\Phi u}\big) \|^2 \\&~~~~~~~~~~~- (\|\bm{u}_{\bm{x}}^\bot\|^2 + |\Im\big<\bm{x},\bm{u}\big>|^2)\Big|.\end{aligned}
     \end{equation}
    Applying a covering argument to strengthen Lemma \ref{lem8}   to $({\bm{x},\bm{u}})\in\Sigma^{n,*}_{s,c}\times \Sigma^{n,*}_{4s,c}$, the next Lemma plays a similar role as Lemma \ref{lem7} for  the parallel part.  
    The technical difference is that, a finer net for approximation of $\bm{x}\in \Sigma^{n,*}_{s,c}$ is needed   to overcome the difficulty of lack of good estimation on $\|\bm{\Phi x}\|_\infty$  (\ref{3.41}), which is analogous to Lemma \ref{lem10}.
    
    \begin{lem}
    \label{lem9}
    Given any sufficiently small constant $\delta>0$ and the threshold $\eta$ (for defining near vanishing measurement), if \begin{equation}
    \label{3.37}
        m\geq C_0s\cdot\max\left\{\frac{1}{\eta^2}\log \Big(\frac{n^2\log(mn)}{\eta^2s}\Big),\frac{1}{\delta^2}\log\Big(\frac{n^2\log(mn)}{\eta^2\delta^2s}\Big)\right\}
    \end{equation}
    for some sufficiently large absolute constant $C_0$, then with probability at least $1-2(mn)^{-9}-c_1\exp(-c_2\min\{\delta^2,\eta^2\}m)$, for some $C$ we have $$\sup_{\bm{x}\in \Sigma^{n,*}_{s,c}}\sup_{\bm{u}\in \Sigma^{n,*}_{4s,c}} \hat{f}^\bot(\bm{x},\bm{u})\leq C(\delta + \sqrt{\eta}).$$ 
    \end{lem}
    
    \noindent{\it Proof.} The proof is again based on covering argument. We present it in two steps. 

    \noindent{\textbf{(Step 1.)   Strengthen (\ref{3.32}) to discrete nets}}

    We  construct $\mathcal{G}_\delta$ as  $\delta$-net  of $\Sigma^{n,*}_{4s,c}$, $\mathcal{G}_{\tilde{\delta}}$ as $\tilde{\delta}$-net of $\Sigma^{n,*}_{ s,c}$ (where $\tilde{\delta}$ will be chosen later), then Lemma \ref{lem2} allows us to  suppose $|\mathcal{G}_\delta|\leq \big(\frac{9n}{4\delta s}\big)^{8s}$,   $|\mathcal{G}_{\tilde{\delta}}|\leq \big(\frac{9n}{\tilde{\delta}s}\big)^{2s}$. A union bound extends (\ref{3.32}) to $(\bm{x},\bm{u}) \in \mathcal{G}_{\tilde{\delta}}\times \mathcal{G}_{\delta}$, i.e., for any $t\geq 0$,\begin{equation}
        \begin{aligned}
        \label{aa3.55}
         &\mathbbm{P}\Big(\sup_{\bm{x}\in\mathcal{G}_{\tilde{\delta}}}\sup_{\bm{u}\in\mathcal{G}_\delta}\hat{f}^\bot(\bm{x},\bm{u}) \geq t\Big) \\\leq&2\exp\Big(-cm\min\{t,t^2\}+8s\log\Big(\frac{9n}{4\delta s}\Big)\\&~~~~~~~~~~~~~~~~~~+2s\log\Big(\frac{9n}{\tilde{\delta}s}\Big)\Big) .
        \end{aligned}
    \end{equation}
   We set $t =\delta$, as $\delta$ is sufficiently small constant,  (\ref{aa3.55}) implies that the following bound \begin{equation}
       \label{overnet2}
       \sup_{\bm{x}\in\mathcal{G}_{\tilde{\delta}}}\sup_{\bm{u}\in\mathcal{G}_\delta}\hat{f}^\bot(\bm{x},\bm{u}) \leq \delta
   \end{equation}
   holds with probability exceeding $1-2\exp(-c\delta^2 m)$ as long as 
   \begin{equation}\label{ala}
       m\gtrsim \frac{s}{\delta^2}\Big[\log\Big(\frac{n}{\delta s}\Big)+\log\Big(\frac{n}{\tilde{\delta} s}\Big)\Big].
   \end{equation}
We will choose $\tilde{\delta}$ such that  (\ref{3.37}) implies (\ref{ala}), hence we can assume (\ref{overnet2}) holds.

\noindent{\textbf{(Step 2.) Strengthen (\ref{overnet2}) to from $(\bm{x},\bm{u})\in \mathcal{G}_{\tilde{\delta}}\times \mathcal{G}_\delta$ to $(\bm{x},\bm{u})\in \Sigma^{n,*}_{s,c}\times \Sigma^{n,*}_{4s,c}$}}

Given $\delta>0$,   for some $(\bm{\hat{x}},\bm{\hat{u}})\in \Sigma^{n,*}_{s,c}\times\Sigma^{n,*}_{4s,c} $ it holds that  \begin{equation}
   \sup_{\bm{x}\in \Sigma^{n,*}_{s,c}}\sup_{\bm{u}\in \Sigma^{n,*}_{4s,c}}\hat{f}^\bot(\bm{x},\bm{u})< \hat{f}^\bot(\bm{\hat{x}},\bm{\hat{u}}) +\delta.\label{suptopoint2}
\end{equation} Note that there exist $\bm{\tilde{u}}\in\mathcal{G}_\delta$ and $\bm{\tilde{x}}\in \mathcal{G}_{\tilde{\delta}}$ such that $\|\bm{\tilde{u}}-\bm{\hat{u}}\|\leq\delta$, $\|\bm{\tilde{x}}-\bm{\hat{x}}\|\leq   \tilde{\delta}$.   Further introduce the shorthand $\bm{\hat{z}} = \sign(\bm{\Phi\hat{x}})$, $\bm{\tilde{z}} = \sign(\bm{\Phi \tilde{x}})$, we perform some algebra to obtain
    \begin{equation}
        \begin{aligned}
        \label{3.40}
         &\hat{f}^\bot(\bm{\hat{x}},\bm{\hat{u}}) \leq \hat{f}^\bot(\bm{\tilde{x}},\bm{\tilde{u}})  + \hat{f}^\bot(\bm{\hat{x}},\bm{\hat{u}})  - \hat{f}^\bot(\bm{\tilde{x}},\bm{\tilde{u}})     \\&\leq \sup_{\bm{x}\in \mathcal{G}_{\tilde{\delta}}}\sup_{\bm{u}\in\mathcal{G}_{\delta}} \hat{f}^\bot(\bm{x},\bm{u})\\&+ \frac{1}{m}\big|\|\Im\big(\diag(\bm{\hat{z}}^*)\bm{\Phi \hat{u}}\big)\|^2 - \|\Im\big(\diag(\bm{\tilde{z}}^*)\bm{\Phi \tilde{u}}\big)\|^2\big|\\& + \big|\|\bm{\hat{u}}_{\bm{\hat{x}}}^\bot\|^2-\|\bm{\tilde{u}}_{\bm{\tilde{x}}}^\bot\|^2\big|+\big||\Im\big<\bm{\hat{x}},\bm{\hat{u}}\big>|^2-|\Im\big<\bm{\tilde{x}},\bm{\tilde{u}}\big>|^2\big|\\:&=I_1+I_{2}+I_{3}+I_{4}.
        \end{aligned}
    \end{equation}
    Recall that we have derived $I_1\leq \delta$ in (\ref{overnet2}). It remains to bound $I_2,I_3,I_4$ separately.

    \noindent
    \textbf{(Step 2.1.) Bound the term $I_2$}
    
    With the promised probability we can first assume $\sup_{\bm{u}\in \Sigma^{n,*}_{8s,c}}\frac{1}{\sqrt{m}}\|\bm{\Phi u}\|= O(1)$ (Lemma \ref{lem3}) and $\|\bm{\Phi}\|_\infty\lesssim \sqrt{\log(mn)}$
     (Lemma \ref{lemma5}). For any $\bm{w}\in\Sigma^{n,*}_{2s,c}$, this implies \begin{equation}
   \label{3.41}
       \begin{aligned}
        &\|\bm{\Phi w}\|_\infty= \max_{k\in [m]}|\bm{\Phi_k^*w}|\leq \max_{k\in[m]} \|\bm{\Phi_k}\|_\infty\|\bm{w}\|_1\\&\leq \|\bm{\Phi}\|_\infty\cdot\sqrt{2s}\|\bm{w}\|\lesssim \sqrt{s\log(mn)}.
       \end{aligned}
   \end{equation}
   Moreover, we can pick some $\beta$ such that $4\eta <\beta \lesssim \eta$, and $\beta m$ is an integer. Then with the promised probability, we  can invoke Lemma \ref{lem10}  to obtain     $\max\{|\mathcal{J}_{\bm{\hat{x}}}|,|\mathcal{J}_{\bm{\tilde{x}}}|\}<\beta m$, which implies that $\hat{E}: = \mathcal{J}_{\bm{\hat{x}}}\cup \mathcal{J}_{\bm{\tilde{x}}} $ satisfies $|\hat{E} |<2\beta m$. We also write $\hat{E}^c = [m]\setminus \hat{E}$. Moreover,  Lemma \ref{lem11} delivers that with high probability, 
   \begin{equation}
   \label{aa3.58}\begin{aligned}
       &\sup_{\bm{u} \in \Sigma^{n,*}_{4s,c}} \|\bm{\Phi}^{\hat{E}}\bm{u}\| \leq \sup_{|\mathcal{S}|\leq 2\beta m}\sup_{|\mathcal{T}|\leq 4s}\|\bm{\Phi}^\mathcal{S}_\mathcal{T}\|\\&= O\Big(\sqrt{\eta m \log\Big(\frac{72}{\eta}\Big)+s\log\Big(\frac{72n}{s}\Big)}\Big).\end{aligned}
   \end{equation}
    With all these preparations, we can first divide $I_2$ according to $\hat{E}$, it provides
    \begin{equation}
	     \begin{aligned}\nonumber
	       {I}_{2 } &\leq \frac{1}{m}\Big|\big\|\Im\big([\diag(\bm{\hat{z}^*})\bm{\Phi}]^{\hat{E}^c}\bm{\hat{u}}\big)\big\|^2-\big\|\Im\big([\diag(\bm{\tilde{z}^*})\bm{\Phi}]^{\hat{E}^c}\bm{\tilde{u}}\big)\big\|^2\Big|\\ &+ \frac{1}{m}\Big|\big\|\Im\big([\diag(\bm{\hat{z}^*})\bm{\Phi}]^{\hat{E} }\bm{\hat{u}}\big)\big\|^2-\big\|\Im\big([\diag(\bm{\tilde{z}^*})\bm{\Phi}]^{\hat{E}}\bm{\tilde{u}}\big)\big\|^2\Big|\\:&={I}_{21}+{I}_{22}.
	     \end{aligned}
	 \end{equation}
	      For $I_{21}$ we can proceed as in (\ref{3.42}).
        \begin{figure*}[!t]
\normalsize
\begin{equation}
        \begin{aligned}
        \label{3.42}
         &I_{21}=\frac{1}{m}\Big|\|\Im\big([\diag(\bm{\hat{z}}^*)]^{\hat{E}^c}\bm{\Phi } \bm{\hat{u}}\big)\|-\|\Im\big([\diag(\bm{\tilde{z}}^*)]^{\hat{E}^c}\bm{\Phi} \bm{ \tilde{u}}\big)\|\Big|\\&~~~~~~~~~~~~~~~~~~~ ~~~~~\cdot \Big|\|\Im\big([\diag(\bm{\hat{z}}^*)]^{\hat{E}^c}\bm{\Phi } \bm{\hat{u}}\big)\|+\|\Im\big([\diag(\bm{\tilde{z}}^*)]^{\hat{E}^c}\bm{\Phi } \bm{\tilde{u}}\big)\|\Big|\\
         &\stackrel{(i)}{\leq} \frac{1}{m}\big\|[\diag(\bm{\hat{z}}^*)]^{\hat{E}^c}\bm{\Phi \hat{u}}-[\diag(\bm{\tilde{z}}^*)]^{\hat{E}^c}\bm{\Phi \tilde{u}}\big\| \cdot \big|\|\bm{\Phi\hat{u}}\|+\|\bm{\Phi \tilde{u}}\|\big|\\
         &\stackrel{(ii)}{\lesssim}\frac{1}{\sqrt{m}}\big( \big\|[\diag(\bm{\hat{z}^*})- \diag(\bm{\tilde{z}^*})]^{\hat{E}^c}\bm{\Phi\hat{u}}\big\|+\big\|[\diag(\bm{\tilde{z}^*})]^{\hat{E}^c}\bm{\Phi}(\bm{\hat{u}}-\bm{\tilde{u}})\big\|\big)\\
         &\leq \|\bm{\hat{z}}^{\hat{E}^c}-\bm{\tilde{z}}^{\hat{E}^c}\|_\infty \cdot\frac{1}{\sqrt{m}}\|\bm{\Phi \hat{u}}\| + \frac{1}{\sqrt{m}}\|\bm{\Phi}(\bm{\hat{u}}-\bm{\tilde{u}})\|\stackrel{(iii)}{\lesssim} \|\bm{\hat{z}}^{\hat{E}^c}-\bm{\tilde{z}}^{\hat{E}^c}\|_\infty+\delta 
        \end{aligned}
    \end{equation}
\hrulefill
\vspace*{4pt}
\end{figure*}
    Note that we deal with the first factor via triangle inequality in $(i),(ii)$, and 
    use $\sup_{\bm{u}\in \Sigma^{n,*}_{8s,c}}\frac{1}{\sqrt{m}}\|\bm{\Phi u}\|= O(1)$ (Lemma \ref{lem3}) in $(ii),(iii)$. Furthermore,  by Lemma \ref{signconti} we can bound $\|\bm{\hat{z}}^{\hat{E}^c}-\bm{\tilde{z}}^{\hat{E}^c}\|_\infty$ from above as follows:  \begin{equation}
       \begin{aligned}
       \label{3.43}
        \|\bm{\hat{z}}^{\hat{E}^c}-\bm{\tilde{z}}^{\hat{E}^c}\|_\infty &= \max_{k\in{\hat{E}^c}} |\sign(\bm{\Phi}_k^*\bm{\hat{x}})-\sign(\bm{\Phi}_k^*\bm{\tilde{x}})|\\&\leq  2\max_{k\in{\hat{E}^c}}\frac{|\bm{\Phi}_k^*(\bm{\hat{x}}-\bm{\tilde{x}})|}{|\bm{\Phi}_k^*\bm{\hat{x}}|}\\&\leq 2\frac{\max_{k\in{\hat{E}^c}} |\bm{\Phi}_k^*(\bm{\hat{x}}-\bm{\tilde{x}})|}{\min_{k\in{\hat{E}^c}} |\bm{\Phi}_k^*\bm{\hat{x}}|}\\&\leq \frac{2 \|\bm{\hat{x}}-\bm{\tilde{x}}\|}{\eta}\cdot \Big( \sup_{\bm{w}\in \Sigma_{2s,c}^{n,*}}\|\bm{\Phi w}\|_\infty\Big) \\&\stackrel{(i)}{\lesssim} \frac{\tilde{\delta}\sqrt{s\log (mn)}}{\eta},
       \end{aligned}
    \end{equation}
where we use (\ref{3.41}) in the last inequality. Thus, we can take $\tilde{\delta} = \frac{\eta \delta}{\sqrt{s\log (mn)}}$ to guarantee $\|\bm{\hat{z}}^{\hat{E}^c}-\bm{\tilde{z}}^{\hat{E}^c}\|_\infty\lesssim \delta$, hence (\ref{3.42}) gives $I_{21}\lesssim \delta$. 
We still need to deal with the near vanishing part of $I_2$, i.e., $I_{22}$. To this end, previous developments provide
    \begin{equation}
	     \begin{aligned}
	     \label{a3.63}
	      {I}_{22}&=\frac{1}{m}\Big( \|\Im\big([\diag(\bm{\hat{z}^*})\bm{\Phi}]^{\hat{E} }\bm{\hat{u}}\big)\|^2\\&~~~~~~~~~~~~~~~~+\|\Im\big([\diag(\bm{\tilde{z}^*})\bm{\Phi}]^{\hat{E}}\bm{\tilde{u}}\big)\|^2 \Big)\\&\leq\frac{1}{m}\big(\|\bm{\Phi}^{\hat{E}}\bm{\hat{u}}\|^2+\|\bm{\Phi}^{\hat{E}}\bm{\tilde{u}}\|^2\big)\\&\leq \frac{2}{m}\sup_{|\mathcal{S}|\leq 2\beta m}\sup_{|\mathcal{T}|\leq 4s} \|\bm{\Phi}^\mathcal{S}_\mathcal{T}\|^2\\
	     & \stackrel{(i)}{\lesssim}\eta \log\Big(\frac{72}{\eta} \Big)+\frac{s}{m}\log\Big(\frac{72n}{s}\Big) \stackrel{(ii)}{\lesssim} \sqrt{\eta} + \delta,
	     \end{aligned}
	 \end{equation}
     where we use (\ref{aa3.58}) in $(i)$, and $(ii)$ follows from    sufficiently small $\eta$ and \textcolor{black}{the sample complexity (\ref{3.37})}. Overall, we arrive at $I_2 \lesssim \delta + \sqrt{\eta}$.

    \noindent
    {\textbf{(Step 2.2.) Derive the bounds for $I_3,I_4$}}

    These are more standard estimates. 
Specifically, 
    we have \begin{equation}
        \begin{aligned}\nonumber
            &I_{3}=  \big|\|\bm{\hat{u}^\bot_{\hat{x}}}\|^2-\|\bm{\tilde{u}^\bot_{\tilde{x}}}\|^2\big|\\&\leq \big|\|\bm{\hat{u}^\bot_{\hat{x}}}\| + \|\bm{\tilde{u}^\bot_{\tilde{x}}}\|\big|\cdot \big|\|\bm{\hat{u}^\bot_{\hat{x}}}-\bm{\tilde{u}^\bot_{\tilde{x}}}\|\big|\stackrel{(i)}{\leq} 8\delta,
        \end{aligned}
    \end{equation}
    where we use $\|\bm{\hat{u}^\bot_{\hat{x}}}\|\leq 1,\|\bm{\tilde{u}^\bot_{\tilde{x}}}\|\leq 1$ and earlier bound (\ref{3.30}).
      Likewise, for $I_{4}$  we have 
    \begin{equation}
        \begin{aligned}\nonumber
        & I_{4} = \big||\Im\big<\bm{\hat{x},\bm{\hat{u}}}\big>|^2 - |\Im\big<\bm{\tilde{x},\bm{\tilde{u}}}\big>|^2\big|\\&\leq \big||\Im\big<\bm{\hat{x},\bm{\hat{u}}}\big>|+|\Im\big<\bm{\tilde{x},\bm{\tilde{u}}}\big>|\big|\cdot \big|\Im\big<\bm{\hat{x},\bm{\hat{u}}}\big>-\Im\big<\bm{\tilde{x},\bm{\tilde{u}}}\big>\big|\\&\leq 2\cdot \big(|\Im\big<\bm{\hat{x}}-\bm{\tilde{x}},\bm{\hat{u}}\big>|+|\Im\big<\bm{\tilde{x}},\bm{\hat{u}}-\bm{\tilde{u}}\big>|\big) \leq 4\delta.
        \end{aligned}
    \end{equation}
    Putting all the pieces into (\ref{3.40}), and also counting all the involved probability terms, the proof is complete. \hfill $\square$
    
    \subsection{The Proof of Theorem \ref{theorem1}}
     \noindent{\it Proof of Theorem \ref{theorem1}.} \textbf{(Step 1.) Recall the reformulation}

     A simple inspection of   (\ref{3.7})--(\ref{aa3.17}) finds that \begin{equation}\label{upperRSC}
         \sup_{\bm{x},\bm{u}}f^\|(\bm{x},\bm{u}) +\sup_{\bm{x},\bm{u}}f^\bot(\bm{x},\bm{u})
     \end{equation}
     is an upper bound for the $2s$ order RIC of $\bm{A}_{\bm{z},r}$, or $4s$ order RIC of $\bm{A}_{\bm{z},c}$. Such upper bound is uniform over all $\bm{x}\in \mathcal{K}$. Here, for the real case where $\mathcal{K} = \Sigma^n_{s,r}$, $\sup_{\bm{x},\bm{u}}$ stands for $\sup_{\bm{x}\in \Sigma^{n,*}_{s,r}}\sup_{\bm{u}\in \Sigma^{n,*}_{2s,r}}$, for the complex case where $\mathcal{K} = \Sigma^n_{s,c}$, $\sup_{\bm{x},\bm{u}}$ refers to $\sup_{\bm{x}\in \Sigma^{n,*}_{s,c}}\sup_{\bm{u}\in \Sigma^{n,*}_{4s,c}}$.

\noindent{\textbf{(Step 2.) Bound the parallel part $\sup_{\bm{x},\bm{u}}f^\|(\bm{x},\bm{u})$}}

     Recall that $\sup_{\bm{x},\bm{u}}f^\|(\bm{x},\bm{u})$ has   been further decomposed as $\sup_{\bm{x},\bm{u}}f_1^\|(\bm{x},\bm{u})+\sup_{\bm{x},\bm{u}}f_2^\|(\bm{x},\bm{u})$ in (\ref{aa3.40}). Applying Corollary \ref{coro1} to $\sup_{\bm{x},\bm{u}}f_1^\|(\bm{x},\bm{u})$, and applying Lemma \ref{lem7} with $\eta = \delta^2$ to $\sup_{\bm{x},\bm{u}}f_2^\|(\bm{x},\bm{u})$, (under simple rescaling of $\delta$) we obtain $\sup_{\bm{x},\bm{u}}f^\|(\bm{x},\bm{u}) \leq \frac{\delta}{2}$ for both the real case and the  complex case.

\noindent{\textbf{(Step 3.) Bound the orthogonal part $\sup_{\bm{x},\bm{u}}f^\bot(\bm{x},\bm{u})$}}

     We split the discussion of $\sup_{\bm{x},\bm{u}}f^\bot(\bm{x},\bm{u})$ into the real case and the  complex case.

     \noindent({\it Real case}.)  As stated in the Theorem, we take $\hat{t}=1$ in $\bm{A}_{\bm{z},r}$. Moreover, in the real case $\bm{x},\bm{u}\in \mathbb{R}^n$, hence it always holds that $\Im\big<\bm{x},\bm{u}\big>=0$.     
     Thus, one can easily see $f^\bot(\bm{x},\bm{u}) = \hat{f}^\bot({\bm{x},\bm{u}})$ by comparing (\ref{aa3.14}) and (\ref{aa3.54}).
     Then   Lemma \ref{lem9} with $\eta = \delta$ (up to rescaling of $\eta$) yields $\sup_{\bm{x},\bm{u}}f^\bot(\bm{x},\bm{u})\leq \frac{\delta}{2}$ under the stated sample size and promised probability. Thus. in the real case (\ref{upperRSC}) can be bounded by $\delta$.

     \noindent({\it Complex case}.) As stated in the Theorem, we take $\hat{t} = \sqrt{\frac{2}{3}}$ in $\bm{A}_{\bm{z},c}$. Similar to the real case, by Lemma \ref{lem9} we can still obtain $\sup_{\bm{x},\bm{u}}\hat{f}^\bot(\bm{x},\bm{u})\leq \frac{\delta}{2}$. Moreover,  we can compare $f^{^\bot}(\bm{x},\bm{u}),\hat{f}^\bot(\bm{x},\bm{u})$ and  apply triangle inequality $(i)$, then plug in $\|\bm{u}_{\bm{x}}^\bot\|^2 = 1-[\Re\big<\bm{u},\bm{x}\big>]^2$ in $(ii)$ to proceed as follows: \begin{equation}
         \begin{aligned}
         \label{aa3.66}
         &\sup_{\bm{x},\bm{u}} f^\bot(\bm{x},\bm{u}) \\\stackrel{(i)}{\leq}& \frac{2}{3}\sup_{\bm{x},\bm{u}} \hat{f}^\bot(\bm{x},\bm{u}) +\sup_{\bm{x},\bm{u}} \frac{1}{3}\cdot\Big|2\big|\Im\big<\bm{x},\bm{u}\big>\big|^2 - \|\bm{u}_{\bm{x}}^\bot\|^2\Big|\\
         \stackrel{(ii)}{\leq} &\frac{\delta}{2}+ \frac{1}{3}\sup_{\bm{x},\bm{u}}\Big|\big|\big<\bm{x},\bm{u}\big>\big|^2 + \big|\Im\big<\bm{x},\bm{u}\big>\big|^2 -1\Big|\leq \frac{1}{3}+\frac{\delta}{2}.
         \end{aligned}
     \end{equation} 
     Therefore, in the complex case (\ref{upperRSC}) is bounded by $\frac{1}{3}+\delta$.

     \noindent{\textbf{(Step 4.) The uniform exact reconstruction for PO-CS}}
     
     Regarding the uniform exact reconstruction guarantee (i.e., the statement after "in particular" in   Theorem \ref{theorem1}),   we only need to take sufficiently small $\delta$ (e.g., $\delta = \frac{1}{3}$) to guarantee (\ref{upperRSC}) is smaller than $\frac{\sqrt{2}}{2}$. Then the claim follows from Lemma \ref{cslem}. \hfill $\square$

     \vspace{1mm}
     
     We point out that by introducing one rescaling factor $\hat{t}$ to the new sensing matrix, our choice $\hat{t}=\sqrt{\frac{2}{3}}$ for the complex case has already been optimized. 

	\section{PO-CS of Low-Rank Matrices}

	\label{section4}
	We shall present here the parallel result of uniform low-rank recovery from phase-only measurements. 
	The main differences between real matrix and complex matrix are in the global sign product embedding property (Remark \ref{rem1}) and a crucial rescaling of new sensing matrix (Remark \ref{remark2}), hence we only focus on the complex  case.

	In this section    we assume   $\bm{X}\in \mathcal{M}^{n_1,n_2}_{r,c}$     is the underlying low-rank matrix. We still use $\bm{\Phi}$ to denote the original   sensing matrix, $\bm{z}$ for phase-only measurements of $\bm{X}$, but keep in mind that here $\bm{\Phi}(\cdot)$ is a random linear map from $\mathbb{C}^{n_1\times n_2}$ to $\mathbb{C}^m$ defined for any   $\bm{U}\in \mathbb{C}^{n_1\times n_2}$ as \begin{equation}
	    \label{randomlinear}\begin{aligned}
     &\bm{\Phi}(\bm{U})  =  \big(\big<\bm{\Phi}_1,\bm{U}\big>,\big<\bm{\Phi}_2,\bm{U}\big>,\cdots,\big<\bm{\Phi}_m,\bm{U}\big>\big)^\top ,\end{aligned}
	\end{equation} where $\{\bm{\Phi}_k:k\in [m]\}$ are i.i.d. copies of $\mathcal{N}^{n_1\times n_2}(0,1)+\mathcal{N}^{n_1\times n_2}(0,1)\ii$, $\big<\bm{\Phi}_k,\bm{U}\big>=\mathrm{Tr}(\bm{\Phi}^*_k\bm{U})$. In PO-CS, our goal is to reconstruct $\bm{X}$ (up to a positive scaling factor) from the phase-only observations $\bm{z}:= \sign\big(\bm{\Phi}(\bm{X})\big)$. Analogous to (\ref{3.6}), by identifying $\bm{B}\in \mathcal{M}_{r,c}^{n_1,n_2}$ with $[\bm{B}]_\mathbb{R}\in \mathcal{M}^{2n_1,n_2}_{2r,r}$,
 we can reformulate PO-CS with $\bm{X}\in \mathcal{M}_{r,c}^{n_1,n_2}$ as (\ref{lrrefo}).
   \begin{figure*}[!t]
\normalsize
\begin{equation}
	    \begin{aligned}\label{lrrefo}
	    &~~~~~~~~\text{find }\bm{U}\in \mathbb{M}^{2n_1, n_2}_{2r,r},\text{ s.t. }\mathcal{A}_{\bm{z}}(\bm{U}) = \bm{e}_1,~\text{where}\\
	    \mathcal{A}_{\bm{z}}(\bm{U})=    &\begin{bmatrix}
	        \frac{1}{\kappa m}\big<\Re\big(\sum_{k=1}^mz_k\bm{\Phi}_k\big),\bm{U}^{[n_1]}\big>+\frac{1}{\kappa m}\big<\Im\big(\sum_{k=1}^mz_k\bm{\Phi}_k\big),\bm{U}^{[2n_1]\setminus[n_1]}\big> \\
	        -\sqrt{\frac{2}{3m}} \big<\Im(z_1\bm{\Phi}_1),\bm{U}^{[n_1]}\big> + \sqrt{\frac{2}{3m}}\big<\Re(z_1\bm{\Phi}_1),\bm{U}^{[2n_1]\setminus[n_1]}\big>\\
	        \vdots \\
	          -\sqrt{\frac{2}{3m}} \big<\Im(z_m\bm{\Phi}_m),\bm{U}^{[n_1]}\big> + \sqrt{\frac{2}{3m}}\big<\Re(z_m\bm{\Phi}_m),\bm{U}^{[2n_1]\setminus[n_1]}\big>
	    \end{bmatrix}.
	    \end{aligned}
	\end{equation}
\hrulefill
\vspace*{4pt}
\end{figure*}
    Based on Lemma \ref{cslem}, our strategy is to show $\mathcal{A}_{\bm{z}}(\cdot)$ for all $\bm{X}\in \mathcal{M}^{n_1,n_2}_{r,c}$ simultaneously respect RIP with  $\delta^{\mathcal{A}_{\bm{z}}}_{4r}<\frac{\sqrt{2}}{2}$. Some algebra verifies that proving \begin{equation}\nonumber
        \sup_{\bm{X}\in (\mathcal{M}^{n_1,n_2}_{r,c})^*}\sup_{\bm{U}\in (\mathcal{M}^{n_1,n_2}_{4r,c})^*}f(\bm{X},\bm{U})<\frac{\sqrt{2}}{2}
    \end{equation} 
	is sufficient, where we define \begin{equation}
	\label{a4.3}\begin{aligned}
	    &f(\bm{X},\bm{U}) :=\Big|\frac{1}{\kappa^2m^2}\big[\Re(\bm{z}^* \bm{\Phi}(\bm{U}))\big]^2\\&~~~~~~~~~~~~~~ + \frac{2}{3m}\big\|\Im\big(\diag(\bm{z}^*)\bm{\Phi}(\bm{U})\big)\big\|^2-1\Big|.\end{aligned}
	\end{equation}
	Our result states that the uniform low-rank recovery guarantee  can be achieved from $O(r(n_1+n_2)\log(r(n_1+n_2)))$ phase-only measurements. Notably, the measurement number for achieving uniform exact reconstruction is near optimal compared to $O(r(n_1+n_2))$ in the classical compressive sensing regime \cite{candes2011tight}.  
	\begin{theorem}{\rm (Uniform Exact Recovery of Complex Low-Rank Matrices)\textbf{.}}
	    Consider the setting of PO-CS of $\bm{X}\in \mathcal{M}^{n_1,n_2}_{r,c}$ described above. Under the sample complexity of $m  \gtrsim r(n_1+n_2)\log(r(n_1+n_2))$, with probability at least $1-  c_1\exp(-c_2m) -2m\exp(-c(n_1+n_2))$, all $\bm{X} \in\mathcal{M}^{n_1,n_2}_{r,c}$ can be exactly reconstructed (up to positive scaling factor) from $\bm{z}=\sign\big(\bm{\Phi}(\bm{X})\big)$ by finding $\bm{\hat{X}} =[\bm{\hat{U}}]_\mathbb{C}$. Here, $\bm{\hat{U}}$ is solved from \begin{equation}\nonumber
	       \bm{\hat{U}} =   \mathop{\arg\min}\limits_{\bm{U}\in \mathbb{R}^{2n_1\times n_2}} \|\bm{U}\|_{*},~~\text{s.t. }\mathcal{A}_{\bm{z} }(\bm{U})= \bm{e}_1,
	    \end{equation}
	     $\mathcal{A}_{\bm{z}}$ is defined in (\ref{lrrefo}).
	   \label{theorem2}
	\end{theorem}
	The proof of Theorem \ref{theorem2} is parallel to Theorem \ref{theorem1}, hence we omit the details but only point out several technical changes. Specifically, Lemma \ref{lemma5} should be substituted with the following Lemma, which gives rise to the probability term $2m\exp(-c(n_1+n_2))$  in Theorem \ref{theorem2}.  
	\begin{lem}
	\label{lrlemma}
	For the linear random map $\bm{\Phi}$ in (\ref{randomlinear}), there exists some absolute constant $c$, $C$, such that $\max_{k\in [m]}\|\bm{\Phi}_k\| \leq C\sqrt{n_1+n_2}$ holds with probability at least $1-2m\exp(-c(n_1+n_2))$.
	\end{lem}
	
	\noindent{\it Proof.} By   \cite[Thm. 4.4.5]{vershynin2018high} $\mathbbm{P}(\|\bm{\Phi}_k\|\geq C_0(\sqrt{n_1}+\sqrt{n_2}+t)\leq 2\exp(-t^2)$ for some constant $C_0$, thus a union bound gives $$\mathbbm{P}\big(\max_{k\in [m]}\|\bm{\Phi}_k\| \geq C(\sqrt{n_1}+\sqrt{n_2}+t)\big)\leq 2m\exp(-t^2).$$
	Setting $ t=\sqrt{c(n_1+n_2)}$ concludes the proof. \hfill $\square$
	
	Based on Lemma \ref{lrlemma}, the estimate $\sup_{\bm{w}\in \Sigma^{n,*}_{2s,c}}\|\bm{\Phi w}\|_\infty \lesssim \sqrt{s\log(mn)}$ used in  (\ref{a4.9}) and (\ref{3.43})  can be replaced with $$\sup_{\bm{W}\in( \mathcal{M}^{n_1,n_2}_{2r,c})^*}\|\bm{\Phi }(\bm{W})\|_\infty \lesssim \sqrt{r(n_1+n_2)},$$ which is due to $|\big<\bm{\Phi}_k,\bm{W}\big>|\leq \|\bm{\Phi}_k\|\cdot\sqrt{\rank(\bm{W})}\cdot \|\bm{W}\|_F$. In addition, the RIP of $\bm{\Phi}$ over low-rank matrices needed in Lemma \ref{lem3} can be positioned in  \cite[Thm. 2.3]{candes2011tight}. For   more details on other modifications, we refer readers to Section \ref{lowcovering} where we will show that the techniques developed in Section \ref{section3} is sufficient for proving uniform recovery guarantee over $\bm{x}\in \mathcal{K}$, as long as $\mathcal{K}$ has low covering dimension.

	\section{Stability and Reconstruction with Norm}
	\label{section6}
	Having presented the uniform exact recovery guarantee over $\Sigma^{n}_{s,c}$ and $\mathcal{M}^{n_1,n_2}_{r,c}$ in PO-CS, we further 
	investigate the uniform stable recovery in a noisy setting. After that, 
 we show  that uniform full reconstruction with norm can be achieved by adding Gaussian dither before capturing the phases. For succinctness we would only present the result for the recovery of complex sparse   signal $\bm{x}\in \Sigma^n_{s,c}$. 
	
	\subsection{Uniform Stable Reconstruction}
    Our approach is to recast PO-CS as a linear compressive sensing problem, and note that the stability for linear compressive sensing has been well developed under the framework of RIP, see for instance,   \cite[Thm. 2.1]{cai2013sparse} for sparse recovery. Our analysis would be based on this well-known result. More precisely, we consider bounded complex additive noise $\bm{\tau}\in \mathbb{C}^m$ satisfying $\|\bm{\tau}\|_\infty\leq \tau_0$,  and the resulting  noisy PO-CS model \begin{equation}
    \label{5.1}
        \bm{\breve{z}} = \bm{z}+\bm{\tau} = \sign(\bm{\Phi x}) +\bm{\tau} .
    \end{equation}  
	This problem set-up accommodates many noise pattern of interest, specifically a moderate phase disturbance $\bm{z}_{\bm{\Lambda}} = \bm{\Lambda z}$ where $\bm{\Lambda}$ is a diagonal matrix with unit entries close to $1$. This also embraces the noise brought by uniform quantization over $\{|z|=1\}$. Moreover, the bounded assumption on noise has been noted to be quite necessary for success of PO-CS, see   \cite[Remark 4.4]{jacques2021importance}. Indeed, to establish the   stable recovery result, we follow a   strategy similarly to \cite{jacques2021importance}. The difference is that, due to the uniformity in Theorem \ref{theorem1}, we are now able to prove a stable recovery guarantee uniformly for all complex sparse signals, (while their result only handles a fixed real signal).
	
	\begin{theorem}
	\label{theorem3}
	{\rm (Uniform Stable Recovery under Bounded Noise)\textbf{.}}
	    Assume $\bm{x}\in \Sigma^{n}_{s,c}$. Recall that in noiseless PO-CS model $\bm{z}=\sign(\bm{\Phi z})$,   one can exactly reconstruct $\bm{x}^\star = \frac{\kappa m}{\|\bm{\Phi x}\|_1}\bm{x}$ using the reformulation (\ref{3.6}).\footnote{This is because we add the virtual measurement (\ref{add830}) to specify the signal norm as $\|\bm{\Phi x}\|_1=\kappa m$.} Consider the noisy PO-CS model (\ref{5.1}) with $\|\bm{\tau}\|_\infty\leq \tau_0$,   we use $\bm{\breve{z}}$ to construct the new sensing matrix $\bm{A}_{\bm{\breve{z}},c}$ in (\ref{3.6}) and solve $\bm{\hat{u}}\in \mathbb{R}^{2n}$ from \begin{equation}\label{stableprog}
	        \bm{\hat{u}} =  \mathop{\arg\min}\limits_{\bm{u}\in \mathbb{R}^{2n}} \|\bm{u}\|_1,~~\text{s.t. }\|\bm{A}_{\bm{\breve{z}},c}\bm{u}- \bm{e}_1\|\leq \tilde{\tau}.
	    \end{equation}
	     Let $\bm{\hat{x}} = [\bm{\hat{u}}]_\mathbb{C}$. If $\tau_0$ is sufficiently small, $m\geq Cs\log \big(\frac{n^2\log(mn)}{s}\big)$ for some sufficiently large $c$,   one can pick $\tilde{\tau} = \sqrt{2}\tau_0$, then     with probability at least $1-2(mn)^{-9} -c_1\exp(-c_2m)$  the stable recovery guarantee $$\|\bm{\hat{x}} - \bm{x}^\star\|\leq c\tau_0$$holds for all $\bm{x}\in  \Sigma^n_{s,c}$ with some absolute constant $c$.  
	\end{theorem}
    	
	\noindent{\it Proof.} We will use $[\cdot]_\mathbb{R}$, $[\cdot]_\mathbb{C}$ introduced in (\ref{c2r}), (\ref{r2c}). Our proof relies on \cite[Thm. 2.1]{cai2013sparse}, a stable recovery guarantee for linear compressive sensing. It states that it is sufficient for us to prove $\|\bm{A}_{\bm{\breve{z}},c}[\bm{x^\star}]_\mathbb{R}-\bm{e}_1\|\leq \tilde{\tau}$ and that $\bm{A}_{\bm{\breve{z}},c}$ has $4s$ order RIC lower than $\frac{\sqrt{2}}{2}-c_0$ for some absolute constant $c_0>0$. For clarity, we present the proof in two steps.

\noindent{\textbf{(Step 1.) Prove $\|\bm{A}_{\bm{\breve{z}},c}[\bm{x^\star}]_\mathbb{R}-\bm{e}_1\|\leq \tilde{\tau}=\sqrt{2}\tau_0$}}

 By (\ref{3.6}) $\bm{A}_{\bm{z},c}$ linearly depends on $\bm{z}$, thus we can write $\bm{A}_{\bm{\breve{z}},c} = \bm{A}_{\bm{z},c}+ \bm{A}_{\bm{\tau},c}$, which together with $\bm{A}_{\bm{z},c}[\bm{x^\star}]_\mathbb{R}=\bm{e}_1$ gives \begin{equation}
     \begin{aligned}
     \|\bm{A}_{\bm{\breve{z}},c}[\bm{x}^\star]_\mathbb{R}-\bm{e}_1\| =& \|\bm{A}_{\bm{\breve{z}},c}[\bm{x}^\star]_\mathbb{R}-\bm{A}_{\bm{{z}},c}[\bm{x}^\star]_\mathbb{R}\|\\= &\|\bm{A}_{\bm{\tau},c}[\bm{x}^\star]_\mathbb{R}\|.
     \end{aligned}
 \end{equation}Furthermore, some algebra   estimates that 
	\begin{equation}
	    \begin{aligned}
	    \label{5.3}
	    &\|\bm{A}_{\bm{\tau},c}[\bm{x}^\star]_\mathbb{R}\|^2 \\ =& \frac{1}{\kappa^2m^2} [\Re(\bm{\tau}^*\bm{\Phi x}^\star)]^2 + \frac{2}{3m}\|\Im(\diag(\bm{\tau}^*)\bm{\Phi x}^\star)\|^2 \\
	     \stackrel{(i)}{\leq}& \tau_0^2 \frac{\|\bm{\Phi x}^\star\|_1^2}{\kappa^2m^2} + \frac{2\tau_0^2}{3m}\|\bm{\Phi x}^\star\|^2\\=&\tau_0^2 + \frac{2\tau_0^2}{3}\|\bm{x}^\star\|^2 \cdot \left(\sup_{\bm{u}\in\Sigma^{n,*}_{s,c}}\frac{1}{m}\|\bm{\Phi u}\|^2\right).
	    \end{aligned}
	\end{equation}
 Note that we use $\|\bm{\tau}\|_\infty \leq \tau_0$  in $(i)$.
	By Lemma \ref{lem4}, for any $\delta>0$ the following holds with probability at least $1-2\exp(-\Omega(\delta^2m))$: $$\|\bm{x}^\star\|^2 \leq  \frac{\kappa^2m^2}{\inf_{\bm{w}\in \Sigma^{n,*}_{s,c}}\|\bm{\Phi w}\|_1^2}\leq \frac{1}{(1-\delta)^2}.$$ By Lemma \ref{lem3}, $\frac{1}{m}\|\bm{\Phi u}\|^2 \leq 1+\delta$ for any $\delta >0$ with probability at least $1-2\exp(-c(\delta) m)$. Thus, taking $\delta$ as sufficiently small constant, with probability at least $1-4\exp(-\Omega(m))$ we have $\|\bm{A}_{\bm{\tau},c}[\bm{x}^\star]_\mathbb{R}\|\leq \sqrt{2}\tau_0$, which implies $\|\bm{A}_{\bm{\breve{z}},c}[\bm{x}^\star]_\mathbb{R}-\bm{e}_1\|\leq \sqrt{2}\tau_0$ and justifies the constraint used in (\ref{stableprog}).

 \noindent{\textbf{(Step 2.)    Establish the RIP of  $\bm{A}_{\bm{\breve{z}},c}$}}

 Based on Theorem \ref{theorem1}, we   know that for any fixed $\delta_0\in (0,\frac{1}{3})$, with high probability $\bm{A}_{\bm{z},c}$ possesses $4s$ order RIP smaller than $\frac{1}{3}+\delta_0$. As we have $\bm{A}_{\bm{\breve{z}},c}=\bm{A}_{\bm{z},c}+\bm{A}_{\bm{\tau},c}$, our idea is to control the effect of $\bm{A}_{\bm{\tau},c}$.

 Specifically, for any $\bm{u}\in \Sigma^{2n,*}_{4s,r}$ we    can estimate as \begin{equation}
	    \begin{aligned}
	    \label{5.4}
	    &\|\bm{A}_{\bm{\tau},c}\bm{u}\|^2 \\=& \frac{1}{\kappa^2m^2}[\Re(\bm{\tau}^*\bm{\Phi}[\bm{u}]_\mathbb{C})]^2 + \frac{2}{3m}\|\Im (\diag(\bm{\tau}^*)\bm{\Phi}[\bm{u}]_\mathbb{C})\|^2\\
	    \stackrel{(i)}{\leq} &\tau_0^2\sup_{\bm{w}\in \Sigma^{n,*}_{4s,c}}\frac{\|\bm{\Phi}\bm{w}\|_1^2}{\kappa^2m^2} +\frac{2\tau_0^2}{3m}\sup_{\bm{w}\in \Sigma^{n,*}_{4s,c}} \|\bm{\Phi w}\|^2 \stackrel{(ii)}{<} 2\tau_0^2,
	    \end{aligned}
	\end{equation}
	where we use $\|\bm{\tau}\|_\infty\leq \tau_0$ in $(i)$, and $(ii)$ is because $$\sup_{\bm{w}\in \Sigma^{n,*}_{4s,c}}\frac{\|\bm{\Phi w}\|_1^2}{\kappa^2m^2}~~\mathrm{and}~~\frac{1}{m}\sup_{\bm{w}\in \Sigma^{n,*}_{4s,c}}\|\bm{\Phi w}\|^2$$ can be made sufficiently close to $1$ by Lemma \ref{lem3}, Lemma \ref{lem4} (up to a simple modification of proof to accommodate $4s$-sparse signals). As all the involved ingredients are uniform,   (\ref{5.4}) holds uniformly for all $\bm{\tau}$ (satisfying $\|\bm{\tau}\|_\infty\leq \tau_0$), $\bm{u}\in \Sigma^{n,*}_{4s,c}$.

 Now we are ready to estimate the $4s$ order RIC of $\bm{A}_{\bm{\breve{z}},c}$. Pick any $\bm{u}\in \Sigma^{2n,*}_{4s,r}$, using $\delta_0<\frac{1}{3}$ we can proceed as follows: \begin{equation}
	    \begin{aligned}\label{upperrip}
	     \|\bm{A}_{\bm{\breve{z}},c}\bm{u}\|^2 &\leq (\|\bm{A}_{\bm{z},c}\bm{u}\| + \|\bm{A}_{\bm{\tau},c}\bm{u}\|)^2\\&\leq (\sqrt{\frac{4}{3}+\delta_0}+\sqrt{2}\tau_0)^2 \\
	    &\leq \frac{4}{3}+\delta_0 + 2\tau_0^2 +4\tau_0 \leq \frac{4}{3}+\delta_0+ 5\tau_0.
	    \end{aligned}
	\end{equation}
	On the other hand, similarly we have \begin{equation}
	    \begin{aligned}\label{lowerrip}
	     \|\bm{A}_{\bm{\breve{z}},c}\bm{u}\|^2&\geq (\|\bm{A}_{\bm{z},c}\bm{u}\|-\|\bm{A}_{\bm{\tau},c}\bm{u}\|)^2\\&\geq (\sqrt{\frac{2}{3}-\delta_0}-\sqrt{2}\tau_0)^2\\&= \frac{2}{3}-\delta_0+2\tau_0^2-2\tau_0\sqrt{\frac{4}{3}-2\delta_0}\\&\geq \frac{2}{3}-\delta_0  -2\sqrt{2}\tau_0.
	    \end{aligned}
	\end{equation}
	Combining (\ref{upperrip}) and (\ref{lowerrip}), by taking sufficiently small $\delta_0$ and $\tau_0$, 
 $\bm{A}_{\bm{\breve{z}},c}$ has $4s$ order RIC   smaller than $\frac{\sqrt{2}}{2}-c_0$ where $c_0\in (0,\frac{\sqrt{2}}{2})$ is some absolute constant. This uniformly holds true for all $\bm{x}\in \Sigma^{n}_{s,c}$ and arbitrary noise $\bm{\tau}$ satisfying $\|\bm{\tau}\|_\infty \leq \tau_0$. 
 As stated at the beginning of this proof, a direct application of   \cite[Thm. 2.1]{cai2013sparse} leads to the desired result. \hfill $\square$

	  Theorem \ref{theorem3} can be viewed as an extension of Theorem \ref{theorem1} in that it 
   precisely recovers Theorem \ref{theorem1} when $\tau_0=0$.
	
	\subsection{Full Reconstruction via Gaussian Dithering}
	Although reconstruction of $\|\bm{x}\|$ is hopeless from $\bm{z} = \sign(\bm{\Phi} \bm{x})$, a question of both theoretical and practical interest asks how one can incorporate  norm reconstruction  into PO-CS. We propose here a simple way for this purpose, which is to add   dither before capturing the phases, and hence the reconstruction would be based on the phases of affine   measurements. Specifically, we adopt  $\bm{\tau}_{d}\sim \mathcal{N}(\bm{0},\rho^2\bm{I}_m)+  \mathcal{N}(\bm{0},\rho^2\bm{I}_m)\ii$ as   random Gaussian dither with dithering scale $\rho$ ($\rho>0$), and then change the original PO-CS model (\ref{a3.1}) to dithered PO-CS \begin{equation}
	\label{5.7}
	    \bm{z}_d = \sign(\bm{\Phi x}+ \bm{\tau}_d).
	\end{equation} 
	We shall see shortly that the dithered phases $\bm{z}_d$ manage to encode the norm information of $\bm{x}$.  Our analysis will be built upon Theorem \ref{theorem1}, and specifically we show   uniform exact reconstruction with norm information can be achieved  under a near optimal sample complexity.

	\begin{theorem}
	\label{theorem4}
	{\rm (Uniform Full Recovery with Norm)\textbf{.}}
	Consider   dithered PO-CS (\ref{5.7}) with $\bm{\Phi}\sim \mathcal{N}^{m\times n}(0,1)+\mathcal{N}^{m\times n}(0,1)\ii$, $\bm{\tau}_d \sim \mathcal{N}^{m\times 1}(0,\rho^2)+\mathcal{N}^{m\times 1}(0,\rho^2)\ii$, $\rho$ is some \textcolor{black}{fixed}, known positive dithering scale, and $\bm{\Phi}$ and $\bm{\tau}_d$ are independent. We now describe our reconstruction procedure: let $\bm{\tilde{\Phi}} = [\bm{\Phi}, \frac{\bm{\tau}_d}{\rho}]\in \mathbb{C}^{m\times (d+1)}$ and construct $\bm{A}_{\bm{z}_d,c}\in \mathbb{R}^{(m+1)\times (2n+2)}$ as in (\ref{3.6}) using $\bm{z}_d$ and $\bm{\tilde{\Phi}}$,   we obtain $\bm{\hat{u}}\in \mathbb{R}^{2n+2}$ by solving\begin{equation}\nonumber
	    \bm{\hat{u}} =  \mathop{\arg\min}\limits_{\bm{u}\in \mathbb{R}^{2n+2}} \|\bm{u}\|_1,~~\text{s.t. } \bm{A}_{\bm{z}_d,c}\bm{u}= \bm{e}_1 ;
	\end{equation} then we let $\bm{x^\sharp}: = [\bm{\hat{u}}]_\mathbb{C}\in \mathbb{C}^{n+1}$; denote the $(n+1)$-th entry    of $\bm{x^\sharp}$ by $t^\sharp $ and we finally take $\bm{\hat{x}} = [\frac{\rho}{t^\sharp}\bm{x^\sharp}]^{[1:n]}$ as the reconstructed signal. We have the following uniform exact recovery guarantee: if $m\geq  Cs\log\big(\frac{n^2\log(mn)}{s}\big)$ for some sufficiently large $C$, with probability at least $1-2(mn)^{-9}-c_1\exp(-c_2m)$, $\bm{\hat{x}}= \bm{x}$ holds uniformly   for all $\bm{x}\in \Sigma^n_{s,c}$. 
	\end{theorem}

	\noindent{\it Proof.} Based on Theorem \ref{theorem1}, the main idea of this proof is to view (\ref{5.7}) as a classical PO-CS problem (\ref{a3.1}). Specifically, (\ref{5.7}) is equivalent to \begin{equation}\nonumber
	    \bm{z}_d = \sign\left( \begin{bmatrix}
	        \bm{\Phi}&\frac{\bm{\tau}_d}{\rho}
	    \end{bmatrix} \begin{bmatrix}
	        \bm{x}\\ \rho
	    \end{bmatrix}\right) : = \sign(\bm{\tilde{\Phi}}\bm{x^\natural}),
	\end{equation}
	where $\bm{\tilde{\Phi}}\sim\mathcal{N}^{m\times (n+1)}(0,1)+\mathcal{N}^{m\times (n+1)}(0,1)\ii$, $\bm{x^\natural}\in \Sigma^{n+1}_{s+1,c}$. Then by Theorem \ref{theorem1}, under the conditions and reconstruction procedure stated in   Theorem \ref{theorem4}, for all $\bm{x}\in \Sigma^n_s$, $\bm{x^\sharp}$ exactly reconstruct $\bm{x^\natural}$ up to positive scaling, i.e., $\bm{x^\sharp} = \lambda \bm{x^\natural}$ for some $\lambda>0$. Note that by construction, the last entry of $\bm{x^\natural}$ equals to $\rho$, while by assumption the last entry of $\bm{x^\sharp}$ is $t^\sharp$. 
 Thus we know $\bm{x^\sharp} = \lambda \bm{x^\natural}$ holds with $\lambda = \frac{t^\sharp}{\rho}$. Combining with our choice of $\bm{\hat{x}}$, we obtain  $$\bm{\hat{x}}=\bm{\hat{x}} = \Big[\frac{\rho}{t^\sharp}\bm{x^\sharp}\Big]^{[1:n]}= \Big[\frac{1}{\lambda}\bm{x^\sharp}\Big]^{[1:n]} = [\bm{x^\natural}]^{[1:n]}=\bm{x}.$$ The proof is concluded. \hfill$\square$

	\begin{rem}
   \label{rem4}
   {\rm(Related Work)} A parallel result for 1-bit compressive sensing can be found in  \cite[Thm. 4]{knudson2016one}, which   is obtained by viewing the dithered model as the original model and then applying the uniform recovery guarantee in \cite{plan2013one}. In particular, their result delivers approximate recovery and requires a known upper bound on $\|\bm{x}\|$, while our Theorem \ref{theorem4} achieves exact reconstruction and is free of the prior estimate on $\|\bm{x}\|$. Restricted to the phase-only scenario,    \cite[Sec. SM3]{chen2022signal} (supplementary material) studied full signal reconstruction   from the phases of affine measurements. Specifically, \cite[Thm. SM3.3]{chen2022signal} (supplementary material)  states that all $\bm{x}\in \mathbb{C}^n$ can be exactly reconstructed (with norm information) from the phases of $3n$ affine measurements. Note that Theorem \ref{theorem4} achieves the same goal for all sparse signals using a measurement number proportional to the sparsity $s$ rather than $n$. 
	\end{rem}
	Analogous to the analysis in Theorem \ref{theorem3}, it is possible to establish uniform stable recovery guarantee for dithered PO-CS. We do not pursue this in the present paper. 
	\section{Discussions}
	\label{section7}
	Some discussions are in order. Specifically, our uniform exact reconstruction guarantee can be generalized to signal sets with low covering dimension, whereas our result and covering-based approach  suffer from some limitation. 
	\subsection{Generalization of Signal Structure}
	\label{lowcovering}
	We previously study sparse signals ($\bm{x}\in \Sigma^n_{s,c}$) and low-rank matrices ($\bm{x}\in \mathcal{M}^{n_1,n_2}_{r,c}$) as two canonical examples that lie in the central of compressive sensing theory. However, a more modern way to study compressive sensing is to assume $\bm{x}\in \mathcal{K}$ where $\mathcal{K}$ is a general low-complexity set beyond sparsity and low-rankness.

	In this part, we further discuss other signal structures for which the developed techniques can yield a uniform reconstruction guarantee under near optimal sample complexity.  Without a concrete structure like sparsity, we need to use some geometric quantities to characterize the intrinsic dimension of a low-complexity signal set.    
	
	\subsubsection{Covering Dimension and Gaussian Width}

	Inspired by   \cite[Definition 5.1]{dirksen2016dimensionality}, we define covering dimension as follows. 
	
	\begin{definition}\label{defi3}
    Assume $\mathcal{K}\subset \mathbb{R}^n/\mathbb{C}^n$ has diameter $\Delta_d(\mathcal{K}): = \sup_{\bm{u},\bm{v}\in \mathcal{K}}\|\bm{u}-\bm{v}\|$. We say $\mathcal{K}$ has covering dimension $K>0$ with parameter $c_0\geq 1$ and base covering $N_0>0$ if for all $0<\epsilon\leq \Delta_d(\mathcal{K})$, there exists $\mathcal{G}_\epsilon$ as $\epsilon$-net of $\mathcal{K}$ satisfying $|\mathcal{G}_\epsilon|\leq N_0\big(\frac{c_0\Delta_d(\mathcal{K})}{\epsilon}\big)^K$.
	\end{definition} 
	\begin{rem}
	   From Lemmas \ref{lem2}, \ref{lrnumber}, $\Sigma_{s,r}^{n,*}$ (resp. $(\mathcal{M}^{n_1,n_2}_{r,r})^*$) has covering dimension $O(s)$ (resp. $O(r(n_1+n_2))$). In fact, many other structured signal sets admit covering dimension much lower than the ambient dimension, e.g.,   finite union of subspace, group sparsity or   block structured sparsity. While we follow \cite{dirksen2016dimensionality} and adopt the treatment via covering dimension, similar notions with different appearance are available in \cite{oymak2015near,xu2020quantized,jacques2017small}. Readers may  consult these references for more examples of signal set with low covering dimension.
	\end{rem}

We  switch to another quantity called Gaussian width that often captures the intrinsic dimension of a set stably and accurately (e.g., \cite{plan2012robust,chandrasekaran2012convex}). The Gaussian width for $\mathcal{K}\subset \mathbb{R}^n$ is given by \begin{equation}\nonumber
        \omega(\mathcal{K}) = \mathbbm{E}\sup_{\bm{t}\in \mathcal{K}} \big<\bm{g},\bm{t}\big>,~\text{where}~\bm{g}\sim \mathcal{N} (\bm{0},\bm{I}_n).
    \end{equation}
    To define Gaussian width of a complex signal set,  in this work we identify    $\mathcal{K}\subset \mathbb{C}^n$ with 
    $[\mathcal{K}]_\mathbb{R}\subset \mathbb{R}^{2n}$ and define  
      \begin{equation}\label{complexwidth}
          \omega(\mathcal{K})=\omega\big([\mathcal{K}]_\mathbb{R}\big).
      \end{equation} 
      Here, we briefly provide some relations between covering dimension and Gaussian width. 
      We consider $\mathcal{K}_0$   contained in the unit Euclidean ball and assume it
      has covering dimension $K$ for parameter $(c_0,N_0)$. Firstly, Dudley's inequality (e.g., \cite[Thm. 8.1.10]{vershynin2018high}) implies 
      $\omega(\mathcal{K}_0) = O\Big(\sqrt{ K\log(N_0^{1/K}ec_0)}\Big)$, which further leads to \begin{equation}
          \omega(\mathcal{K}_0)=\tilde{O}(\sqrt{K})
          \label{gauviacover}
      \end{equation} by omitting logarithmic factors. Furthermore,   the covering number of $\mathcal{K}_0$ regarding covering radius $\epsilon$, formally denoted by $\mathscr{N}(\mathcal{K}_0,\epsilon)$, is defined as the minimum cardinality of an $\epsilon$-net of $\mathcal{K}_0$. Using an equivalent notion of  Kolmogorov entropy defined  as $\mathscr{H}(\mathcal{K}_0,\epsilon)=\log\mathscr{N}(\mathcal{K}_0,\epsilon)$ \cite{kolmogorov1959varepsilon}, Definition \ref{defi3} just states that \begin{equation}
          \label{featureepi}\mathscr{H}(\mathcal{K}_0,\epsilon)\leq K\log\Big(\frac{2N_0^{1/K}c_0}{\epsilon}\Big).
      \end{equation} The key feature here is that $\mathscr{H}(\mathcal{K}_0,\epsilon)$ depends on $\epsilon$ merely in a logarithmic manner, which could be much tighter than Sudakov's inequality \begin{equation}\nonumber
          \mathscr{H}(\mathcal{K}',\epsilon)\leq \frac{C\omega(\mathcal{K}')}{\epsilon^2}\label{sudakov}
      \end{equation} that holds for arbitrary $\mathcal{K}'\subset \mathbb{R}^n/\mathbb{C}^n$ (e.g., \cite[Thm. 8.1.13]{vershynin2018high}).

	\subsubsection{Generalization}
	 We  argue that our proof is extendable to signal sets with low covering dimension. Specifically, we assume $\bm{x}\in \mathcal{K}$ for some symmetric cone $\mathcal{K}$ (i.e., $\bm{u}\in \mathcal{K}$ implies $\lambda \bm{u}\in \mathcal{K}$ for all $\lambda \in \mathbb{R}$) and focus on the complex case. As in (\ref{3.6}), we can reformulate PO-CS as  a (real) linear compressive sensing problem
	 \begin{equation}
   \label{generefo}
       \begin{aligned}
       &\text{find }\bm{u}\in[\mathcal{K}]_\mathbb{R},\text{ s.t. }\bm{A}_{\bm{z},c}\bm{u} = \bm{e}_1,\text{ with }\\ \bm{A}_{\bm{z},c} = &\begin{bmatrix}
            \frac{ 1}{\kappa m}\cdot \Re\big(\bm{z^*\Phi}\big) & - \frac{ 1}{\kappa m}\cdot\Im\big(\bm{z^*\Phi}\big) \\
           \sqrt{\frac{2}{3m}}\Im\big(\diag(\bm{z^*})\bm{\Phi}\big) &  \sqrt{\frac{2}{3m}}\Re\big(\diag(\bm{z^*})\bm{\Phi}\big)
       \end{bmatrix}.
       \end{aligned}
   \end{equation}
	We stick to the RIP-based approach and aim to prove $\bm{A}_{\bm{z},c}$ respects RIP over some symmetric cone $\mathcal{U}$   simultaneously for all $\bm{x}\in \mathcal{K}$, for which it is enough to show 
	\begin{equation}\label{genetarget}
	    \sup_{\bm{x}\in \mathcal{K}^*}\sup_{\bm{u}\in (\mathcal{U}')^*}f(\bm{x},\bm{u})< \Delta_0
	\end{equation}
    	where $f(\bm{x},\bm{u})$ is defined in (\ref{short1}), $\Delta_0\in (0,1)$ is some threshold, $\mathcal{K}^*=\mathcal{K}\cap \mathbb{S}^{n-1}_c$, $\mathcal{U}'\subset \mathbb{C}^n$ is some symmetric cone satisfying $[\mathcal{U}']_\mathbb{R}\supset \mathcal{U}$, and we write $(\mathcal{U}')^*= \mathcal{U}'\cap \mathbb{S}^{n-1}_{c}$. Note that we can restrict in the supremum that $\bm{u}\in \mathbb{S}_c^{n-1}$ due to the homogeneity of $\bm{u}$ in the definition of RIP.   
    	Here, some cautiousness is needed to select $\mathcal{U}'$:

    	1) To simply confirm the {\it possibility} of uniform exact reconstruction, it is always sufficient to take $\mathcal{U}'= \mathcal{K}-\mathcal{K} $ that satisfies $[\mathcal{U}']_\mathbb{R}\supset \mathcal{U}:= [\mathcal{K}-\mathcal{K}]_\mathbb{R}$.

    	2) Practically, we want to achieve   uniform exact reconstruction via some tractable algorithm. In this case, we should select $(\mathcal{U}',\Delta_0)$ such that there exists a tractable algorithm which exactly solves (\ref{generefo}) under the RIP of $\bm{A}_{\bm{z},c}$ over $\mathcal{U}$ for some distortion $\Delta_0$. We require $\Delta_0>\frac{1}{3}$ in the complex case as we can only achieve a RIP distortion of $\frac{1}{3}+\delta$.

    	 We pause to demonstrate 2) with the concrete example of $s$-group-sparse signals. Given $g_1,g_2,..,g_N\subset [n]$ as non-overlapping groups (i.e., $g_i\cap g_j=\varnothing$ for $i\neq j$), then $\bm{x}\in \mathbb{R}^n/\mathbb{C}^n$ is $s$-group-sparse if $\sum_{i=1}^N \mathbbm{1}(\bm{x}^{g_i})\leq s$. To reconstruct $s$-group-sparse $\bm{x}$ from $\bm{A}\in \mathbb{R}^{m\times n}$ and $\bm{y}=\bm{Ax}$ one can minimize the group norm \begin{equation}\nonumber
    	     \|\bm{x}\|_g:= \begin{cases}\sum_{i=1}^N \|\bm{x}^{g_i}\|,~~~\text{if }\mathrm{supp}(\bm{x})\subset \cup_{i\in[N]}g_i,\\
    	     \infty,~~~~~~~~~~~~~~\text{otherwise}.
    	     \end{cases}
    	 \end{equation}  
    	 under the linear constraint from observations. This program exactly recovers $\bm{x}$ as long as $\bm{A}$ respects RIP over the set of $(2s)$-sparse signals   with distortion lower than $\frac{\sqrt{2}}{2}$ \cite[Sec. 4]{traonmilin2018stable}. To analyse PO-CS of $s$-group-sparse $\bm{x}\in \mathbb{C}^n$ regarding some non-overlapping groups $\{g_1,...,g_N\}$, we first observe that $[\bm{x}]_\mathbb{R}$ is $(2s)$-group-sparse regarding $\{g_1,...,g_N,g_1+n,...,g_N+n\}$. Thus, we can take $\mathcal{U}$ as  the set of all $(4s)$-group-sparse (real) signals (regarding $\{g_1,...,g_N,g_1+n,...,g_N+n\}$), while $\mathcal{U}'$ as   the set of all $(4s)$-group-sparse (complex) signals (regarding $\{g_1,...,g_N\}$).

	Recall that we assume $\bm{x}\in \mathcal{K}$ for some symmetric cone.  We claim that, our methodology can derive uniform exact reconstruction as long as $\mathcal{K}_-^*:=(\mathcal{K}-\mathcal{K})\cap \mathbb{S}^{n-1}_c$   and $(\mathcal{U}'_-)^*:=(\mathcal{U}'-\mathcal{U}')\cap\mathbb{S}^{n-1}_c$   have {\it{low covering dimension}} (note that $\mathcal{K}^*\subset \mathcal{K}_-^*$, $(\mathcal{U}')^*\subset (\mathcal{U}'_-)^*$). Consequently, our uniform exact reconstruction guarantee generalizes to many other structures such as finite union of subspace \cite{eldar2009robust}, group sparsity or block structured sparsity \cite{eldar2010block,ayaz2016uniform,traonmilin2018stable},  cosparse signals \cite{nam2013cosparse,giryes2014greedy}, to name just a few.

	To be more precise, regarding some other secondary parameters we assume $\mathcal{K}_-^*$ and $(\mathcal{U}_-')^*$ have  covering number of scaling $O(K)$ for some $K$ that is order-wisely lower than the ambient dimension $n$. To justify what we claimed in the last paragraph,  we will demonstrate that  one can prove (\ref{genetarget})  with a sample size of $m=\tilde{O}(K)$.
	This can be done by suitably adjusting the   technical ingredients used in the proof of Theorem \ref{theorem1}:

	1) For Lemma  \ref{lem3} concerning RIP of the original sensing matrix, one can show (\ref{a3.16}) remains valid for $\bm{u}\in \mathcal{K}_-^*$ and $\bm{u}\in (\mathcal{U}_-')^*$ with a sample size of  $m=O(\omega^2(\mathcal{K}^*_-)+\omega^2(\mathcal{U}')) =\tilde{O}(K)$ (\ref{gauviacover}) by using the main result in  \cite{mendelson2007reconstruction}.\footnote{Also see      \cite[Thm. 2.2]{jacques2021importance},     \cite[Thm. 2.1]{mendelson2008uniform}, as well as Fact \ref{fact1} of the present work.}  This is recurring in the whole proof, and sometimes some trivial modification is required, e.g.,  one should use $\bm{u_x}^\bot=\bm{u}-\Re\big<\bm{u},\bm{x}\big>\bm{x}$ and $\sup_{\bm{u}\in \mathcal{K}^*_- \cup (\mathcal{U}'_-)^*}\|\bm{\Phi u}\|=O(\sqrt{m})$ to proceed (\ref{3.28add}) from $(ii)$ to $(iii)$, rather than using (\ref{lem5impli}).

	2) Our proof for Lemma \ref{lem4}  directly works for $\bm{w}\in\mathcal{K}$ where $\mathcal{K}$ has low covering dimension. More specifically, (\ref{3.14}) is still valid for $\bm{w}\in \mathcal{K}^*$ using a measurement number of $m=\tilde{O}(\delta^{-2}K)$. This will be used to control   $f_1^\|(\bm{x},\bm{u})$ for analysing the parallel part (Corollary \ref{coro1}). 
	
	3) (\ref{a4.9})  in the proof of Lemma \ref{lem10} utilizes the sparsity to show  $\|\bm{\Phi}(\bm{\hat{x}}-\bm{\tilde{x}})\|\leq \eta$. For a general $\mathcal{K}^*$ with $O(K)$ covering dimension, we can use a finer  $\hat{\delta}$-net with $\hat{\delta}\asymp \frac{\eta}{\sqrt{n}}$ for covering $\mathcal{K}^*$. Combined with a universal bound $\max_{k\in [m]}\|\bm{\Phi}_k\|=O(\sqrt{n})$ that holds with high probability, one still has $\|\bm{\Phi}(\bm{\hat{x}}-\bm{\tilde{x}})\|_\infty\leq (\max_{k\in [m]}\|\bm{\Phi}_k\|)\cdot\|\bm{\hat{x}}-\bm{\tilde{x}}\|=\eta$. On the other hand, using such a finer net at worst induces some $\log n$ factors to the sample complexity  (\ref{4.3}) due to the inessential dependence of $\mathscr{H}(\mathcal{K}^*,\epsilon)$ on $\epsilon$  (\ref{featureepi}). Overall, $\sup_{\bm{x}\in \mathcal{K}^*}|\mathcal{J}_{\bm{x}}|=O(\beta m)$ holds true with measurement number $m=\tilde{O}(\beta^{-2}K)$.

	4) Instead of the operator norm bounds, 
	Lemma \ref{lem11} should be changed to a bound \begin{equation}
	    \label{newlem10}\begin{aligned}
	    \sup_{|\mathcal{S}|=\beta m}\sup_{\bm{u}\in (\mathcal{U}')^*} \|\bm{\Phi}^\mathcal{S}\bm{u}\| &= \sup_{\bm{v}\in \Sigma^{m,*}_{\beta m,c}}\sup_{\bm{u}\in (\mathcal{U}')^*} \Re\big(\bm{v}^*\bm{\Phi} \bm{u}\big)\\&=\tilde{O}(\sqrt{\beta m+K}),\end{aligned}
	\end{equation}
 which will be used once in (\ref{a3.63}). Note that (\ref{sparsede}) involves the sparse decomposition $\frac{\bm{\hat{b}}-\bm{\tilde{b}}}{\|\bm{\hat{b}}-\bm{\tilde{b}}\|}=\bm{b}_1+\bm{b}_2$ and hence does not directly generalize. To circumvent the issue, one can take  $\mathcal{G}_2$ as a  $\sqrt{\beta}$-net of $(\mathcal{U}')^*$ and then obtain  $\Re(\bm{\tilde{a}}^*\bm{\Phi}(\bm{\hat{b}}-\bm{\tilde{b}}))\leq \|\bm{\tilde{a}}\|\|\bm{\Phi}(\bm{\hat{b}}-\bm{\tilde{b}})\|\lesssim \sqrt{\beta m}$, which is then dominated by the main bound (\ref{overnet1}). Compared to original $\frac{1}{8}$-net, using a $\sqrt{\beta}$-net  $\mathcal{G}_2$ only induces logarithmic changes to the bound due to (\ref{featureepi}).

 5) With Lemmas \ref{lem3}-\ref{lem11} in place,  nearly all arguments in the main proof (specifically Lemmas \ref{lem7}, \ref{lem9}) do not depend on the sparsity and thus directly transfer to general signal structure. In fact, the only  modification is in (\ref{3.43}) that involves the estimate (\ref{3.41}) built upon sparsity. The remedy is the same as  3): we can apply the universal estimate $\sup_{\bm{w}\in \mathbb{S}^{n-1}_c}\|\bm{\Phi w}\|\leq \max_{k\in [m]}\|\bm{\Phi}_k\|=O(\sqrt{n})$ in $(i)$ of (\ref{3.43}) to obtain $\|\bm{\hat{z}}^{\hat{E}^c}-\bm{\tilde{z}}^{\hat{E}^c}\|_\infty=O(\eta^{-1}\tilde{\delta}\sqrt{n})$; then we can retain $\|\bm{\hat{z}}^{\hat{E}^c}-\bm{\tilde{z}}^{\hat{E}^c}\|_\infty=O(\delta)$ by using a finer net with $\tilde{\delta}=\frac{\eta\delta}{\sqrt{n}}$; this can only induce some $\log n$ factors to the sample complexity (\ref{3.37}) due to (\ref{featureepi}). Overall, Lemma \ref{lem9} remains valid with a measurement number of $m=\tilde{O}(\eta^{-2}K+\delta^{-2}K)$.

 Therefore, we arrive at the desired conclusion --- our covering-based approach can prove uniform   reconstruction result with $m=\tilde{O}(K)$ phase-only measurements, provided that $\mathcal{K}^*_-$ and $(\mathcal{U}'_-)^*$ have $O(K)$ covering dimension (recall that $\mathcal{K},\mathcal{U}'$ are the symmetric cones in (\ref{genetarget})).



	\subsection{Limitation of Our Result and Approach}\label{limit}
	
	While enjoying the aforementioned generalization, our uniform recovery result and proof approach do suffer from some limitation. The main downside of our Theorems \ref{theorem1}-\ref{theorem2} is that they do not accommodate model error. More specifically, Theorem \ref{theorem1} is uniform over all $s$-sparse $\bm{x}$ but provide no guarantee for  $\bm{x}\notin \Sigma^{n,*}_{s,c}$ (even if such $\bm{x}$ may still be close to some $\bm{x}'\in \Sigma^{n,*}_{s,c}$). By contrast, an {\it instance optimal} algorithm (e.g., \cite[Sec. 3.2]{traonmilin2018stable}, \cite[Sec. 2]{jacques2021importance}) in linear compressive sensing is robust to model error. In particular, if $\delta_{2s}^{\bm{A}}<\frac{\sqrt{2}}{2}$, then the basis pursuit (\ref{2.5}) with constraint changed to $\|\bm{Au}-\bm{y}\|\leq \varepsilon$ is instance optimal, in that \begin{equation}\label{6.10}
	    \|\bm{\hat{x}}-\bm{x}\| \leq C_1 \varepsilon+ C_2e_0(\bm{x},\Sigma^n_{s,r})
	\end{equation}
	holds as long as $\bm{x}$ is feasible to the constraint (i.e., $\|\bm{Ax}-\bm{y}\|\leq \varepsilon$), where\begin{equation}\label{instane0}
	    e_0(\bm{x},\Sigma^n_{s,r})=\min_{\bm{x}_1\in \Sigma^n_{s,r}}\frac{\|\bm{x}-\bm{x}_1\|}{\sqrt{s}}.
	\end{equation} It means that, under one RIP matrix $\bm{A}$, (\ref{2.5})   still delivers accurate estimation provided the model error $e_0(\bm{x},\Sigma^n_{s,r})$ is small. Therefore, albeit achieving uniform exact reconstruction using near optimal sample complexity (Remark \ref{rem1}), there is still a gap in {\it robustness to model error} between our results and the instance optimal ones in linear compressive sensing. It remains an open question whether it is possible to prove a uniform recovery guarantee that nicely accommodates model error.

	Then, we describe a concrete possibility to narrow the above gap. By proving RIP of the new sensing matrices as in \cite{jacques2021importance} and the present paper, it may be overly difficult to establish a guarantee  comparable to (\ref{6.10}) that holds for all feasible $\bm{x}$, since the new sensing matrices change with $\bm{x}$. Indeed, it may be more realistic to pursue a uniform guarantee for   signals with small model error. For instance, specialized to sparsity, it is of particular interest to establish a uniform guarantee over the set of {\it approximately sparse signals}  \begin{equation}
	     \mathcal{K}_{q,c} = \left\{\bm{x}=[x_i]\in \mathbb{C}^n: \sum_{i=1}^n|x_i|^q \leq s^{1-\frac{q}{2}},\|\bm{x}\|\leq1\right\}
	     \label{approsp}
	 \end{equation}
	 for some $q\in (0,1]$. Note that it is common to use (quasi) $\ell_q$-norm ($q\in (0,1]$) to capture approximate sparsity, e.g., \cite{plan2013one,mendelson2008uniform,negahban2011estimation,chen2022high,raskutti2011minimax}, and $\mathcal{K}_{q,c}$ is a relaxation of $\Sigma^{n,*}_{s,c}$ as $\Sigma^{n,*}_{s,c}\subset \mathcal{K}_{q,c}$. To this end, we only need to prove the new sensing matrices respect RIP simultaneously for all $\bm{x}\in \mathcal{K}_q$ and then invoke (\ref{6.10}), and it is enough to prove
	 \begin{equation}
            \label{approxigoal}\begin{aligned}
            &\sup_{\bm{x}\in\mathcal{K}_q} \sup_{\bm{u}\in\Sigma^{n,*}_{4s,c}} \Big|\frac{1}{\kappa^2m^2}[\Re(\bm{z^*\Phi u})]^2\\&~~~~~~~+\frac{2}{3m}\big\|\Im \big(\diag(\bm{z}^*)\bm{\Phi u}\big)\big\|^2 -1\Big|< \frac{\sqrt{2}}{2}.\end{aligned}
        \end{equation}
	  Compared to (\ref{a3.11}), the only difference is on the range of $\bm{x}$ where the supremum is taken.

	  Unfortunately, our   approach based on covering arguments cannot establish (\ref{approxigoal}) with $\tilde{O}(s)$ phase-only measurements. Actually, while $ \mathcal{K}_{q,c} $ has a Gaussian width of the same order as $\omega(\Sigma^{n,*}_{s,c})$ (i.e.,  $ O\big(s\log\frac{n}{s}\big)$),\footnote{Based on the fact that taking convex hull cannot change Gaussian width \cite[Prop. 7.5.2]{vershynin2018high}, this can be shown by     \cite[Lem. 3.7]{mendelson2008uniform}.}     $\mathcal{K}_{q,c}$ does not admit $\tilde{O}(s)$ covering dimension. In particular, the cardinality of an $\epsilon$-net for $\mathcal{K}_{q,c}$ is of order $O\big(\big(\frac{C_qn}{\epsilon^{q_1} s}\big)^{s/\epsilon^{q_2}} \big)$ where $C_q$, $q_1\in [-1,1]$ are absolute constants depending on $q$, $q_2 = \frac{2q}{2-q}$, see   \cite[Lem. 3.8]{mendelson2008uniform} and   \cite[Lem. 3.4]{plan2013one}. In stark contrast to Definition \ref{defi3}, the covering number of $\mathcal{K}_{q,c}$ depends on $\epsilon$ in an exponential way. 
	  Using the notion of Kolmogorov entropy, by ignoring logarithmic factors $\mathscr{H}(\mathcal{K}_{q,c},\epsilon)$  is   proportional to $s\epsilon^{-q_2}$, as opposed to the entropy for a set with low covering dimension (\ref{featureepi}). This is problematic in our covering arguments  because a net with approximation error $\delta\ll 1$ is needed in our proofs of Lemmas \ref{lem10}, \ref{lem9} (e.g., $\hat{\delta}\lesssim\frac{\eta}{\sqrt{s}}$ in Lemma \ref{lem10}). Specifically, when $\mathscr{H}(\mathcal{K}_{q,c},\epsilon)$ linearly depends on $s\epsilon^{-q_2}$, using a $\epsilon$-net of $\mathcal{K}_{q,c}$ with $\epsilon\ll 1$ significantly worsens the required sample complexity (rather than logarithmically).
	  

	 In a nutshell, our covering approach is insufficient for proving uniform recovery over a signal set that does not admit low covering dimension, even though the set may have Gaussian width much lower than $n$. One interesting example is $\mathcal{K}_{q,c}$ in (\ref{approsp}) that has Gaussian width $O\big(s\log\frac{n}{s}\big)$ but does not admit a low covering dimension. We do not know whether uniform recovery over $\mathcal{K}_{q,c}$ (with instance optimality) is achievable using $\tilde{O}(s)$ phase-only measurements. 
	 We conjecture that, resorting to   more advanced concentration inequalities or generic chaining \cite{talagrand2005generic} may be fruitful for such investigation, as the work has been reduced to bounding an empirical process (\ref{approxigoal}). The difficulty may still lie in the essential dependence of the process on $\sign(\bm{\Phi x})$.
	 In any case, our current proof for Theorem \ref{theorem1}  is of independent interest, in that it only involves the elementary covering arguments. Note that in classical compressive sensing, the role of such elementary proof is also played by \cite{mendelson2007reconstruction,baraniuk2008simple} that applied covering argument to show sub-Gaussian random matrix respects RIP over sparse signals.

	\subsection{Non-Uniform Guarantee for Complex Signal}
	As discussed above, our uniform recovery guarantee generalizes to 
	many other signal sets that admit low covering dimension, but it cannot well accommodate model error. Recall that  our Theorem \ref{theorem1} states that, all $s$-sparse {\it complex} signals can be {\it uniformly, exactly} reconstructed from $\tilde{O}(s)$ phase-only measurements, thereby simultaneously answering two open questions in \cite{jacques2021importance} in affirmative.  Interestingly, 
	 we note that our attempts to handle complex signals, and the covering approach to  uniform guarantee, are totally vertical. As it turns out, the limitation described in Section \ref{limit} mainly stems from the covering arguments for achieving a uniform guarantee. Therefore,  without pursuing the uniformity, we can extract part of the technical ingredients to establish a non-uniform guarantee for a {\it fixed}, complex signal. This   complements our    uniform reconstruction guarantee in that it  accounts for model error.

	 
	 \begin{theorem}
	 \label{theorem5}{\rm (Non-uniform Guarantee for Fixed Complex Signal)\textbf{.}}
	 Given a symmetric cone  $\mathcal{K}\subset \mathbb{C}^n$ (i.e., $\bm{u}\in \mathcal{K}$ implies $\lambda \bm{u}\in \mathcal{K}$ for all $\lambda \in \mathbb{R}$) and some $\delta>0$. We consider a fixed signal $\bm{x}\in \mathbb{S}^{n-1}_c$, the observed measurement phases $\bm{z} = \sign(\bm{\Phi x})$ where     $\bm{\Phi}\sim \mathcal{N}^{m\times n}(0,1)+\mathcal{N}^{m\times n}(0,1)\ii$, and we construct the new sensing matrix $\bm{A}_{\bm{z},c}$   as in (\ref{3.6}) with $\hat{t}=\sqrt{\frac{2}{3}}$. If \begin{equation}
	 \label{6.4}
	     m \geq \frac{C}{\delta^2}\omega^2\big(\mathcal{K}\cap\mathbb{S}^{n-1}_c\big),
	 \end{equation} 
	 then with probability at least $1-c_1\exp(-c_2\delta^2 m )$, $\bm{A}_{\bm{z},c}$ respects \rm{RIP}$([\mathcal{K}]_\mathbbm{R},\frac{1}{3}+\delta)$, i.e., \begin{equation}
	 \label{a6.5}
	    \big(\frac{2}{3}-\delta\big)\|\bm{u}\|^2\leq  \|\bm{A}_{\bm{z},c}\bm{u}\|\leq \big(\frac{4}{3}+\delta\big)\|\bm{u}\|^2 
	 \end{equation}
  holds for all $\bm{u}\in [\mathcal{K}]_\mathbb{R}$.
	 \end{theorem}
	 	 Fix the distortion $\delta$, the sample complexity in Theorem \ref{theorem5} is captured by Gaussian width rather than covering dimension as in   uniform reconstruction (Section \ref{lowcovering}).  Indeed,  it extends \cite[Thm. 3.3]{jacques2021importance} to the complex case, and note that (\ref{6.4}) slightly refines the sample complexity in \cite{jacques2021importance} that reads as $O\big(\omega^2([\mathcal{K}-\mathbb{R}\bm{x}]\cap\mathbb{S}^{n-1}_r)\big)$, where $\mathbb{R}\bm{x}=\{\lambda\bm{x}:\lambda\in\mathbb{R}\}$. Compared to Theorems \ref{theorem1}-\ref{theorem2}, Theorem \ref{theorem5} allows model error as $\bm{x}\in\mathcal{K}$ is not required.    For instance, one outcome is that for any fixed $\bm{x}\in \mathbb{S}_c^{n-1}$, from $O\big(s\log\frac{n}{s}\big)$ noiseless phase-only measurements one can reconstruct $\bm{x^\star}=\frac{\kappa m}{\|\bm{\Phi x}\|_1}\bm{x}$ with an $\ell_2$-norm error of   $O\big(e_0([\bm{x}]_\mathbb{R},\Sigma^{2n}_{2s,r})\big)$ (see (\ref{instane0}) for the definition). Readers may consult \cite[Sec. 3-4]{jacques2021importance} for more discussions on the implications of Theorem \ref{theorem5}.


	 Compared with the proof of Theorem \ref{theorem1} where one aims to upper bound $\sup_{\bm{x},\bm{u}}f(\bm{x},\bm{u})$ (\ref{a3.12}), for Theorem \ref{theorem5} the goal simplifies to bounding $\sup_{\bm{u}} f(\bm{x},\bm{u})$ since $\bm{x}$ is now a fixed signal. One shall see that, this allows a significantly simpler and cleaner analysis where all technicalities for near vanishing measurements are not needed --- the dependence of $f_2^\|(\bm{x},\bm{u})$ (\ref{aa3.40}) and $\hat{f}^\bot(\bm{x},\bm{u})$ (\ref{aa3.54}) on $\sign(\cdot)$ are inessential, and the rotational transform in Lemma \ref{lemma6}, \ref{lem8} reveals that they can   actually be viewed as quite standard Gaussian processes. As a result, many well-established tools can be applied to replace the covering arguments. For example, the concentration of Lipschitz function of Gaussian variables straightforwardly handles $\sup_{\bm{u}}f_2^\|(\bm{x},\bm{u})$ for the parallel part, as done in   \cite[Lem. 5.4]{jacques2021importance}.

	  Since we position our primary contribution  in the uniform reconstruction guarantee,    the proof of Theorem \ref{theorem5} is relegated to Appendix for preserving the presentation flow. Here, we only  highlight some additional technicalities compared to the real case in \cite{jacques2021importance}. For example,  
	 the proof in \cite{jacques2021importance} utilizes the RIP of a  Gaussian matrix \cite[Thm. 2.2]{jacques2021importance} to control the orthogonal part, while this becomes insufficient in our proof since the first component of $\bm{\Psi}_k$ in  (\ref{A.13}) is not Gaussian and not zero-mean. To circumvent the issue, we invoke      \cite[Thm. 3.2]{dirksen2016dimensionality} instead and this requires more work such as estimating the $\gamma_2$-functional \cite{talagrand2005generic}. Indeed, the refinement of sample complexity is 
	due to this more refined technical tool and more careful analysis.                
	 

\section{Numerical Simulation}
\label{section8}
In this section, we conduct numerical experiments to validate our theoretical
results. Note that exact recovery guarantee for the complex case of  PO-CS is first proved in this work. To verify this, the underlying signals in all experiments are complex-valued. Specifically, we consider $(s=3)$-sparse complex signal with ambient dimension $n = 80$. Parallel to \cite{jacques2021importance}, the support of $\bm{x}$ is uniformly, randomly drawn from $\binom{80}{3}$ possibilities, then non-zero entries are i.i.d. distributed as $\mathcal{N}(0,1)+\mathcal{N}(0,1)\ii$. Eventually, we normalize $\bm{x}$ such that $\|\bm{x}\| = 1$. As assumed throughout the present paper, entries of $\bm{\Phi}$ are independent copies of  $\mathcal{N}(0,1)+\mathcal{N}(0,1)\ii$. We test the success rate under different $m$ based on 100 independent trials. In a single trial, signals in $\mathbb{C}^n$ are identified with the real ones in $\mathbb{R}^{2n}$ via $[\cdot]_\mathbb{R}$.   For recovery, we find $\bm{\hat{u}}\in \mathbb{R}^{2n}$ by solving the basis pursuit  (\ref{3.6}) \begin{equation}
    \bm{\hat{u}} =  \mathop{\arg\min}\limits_{\bm{u}\in \mathbb{R}^{2n}} \|\bm{u}\|_{1},~~\text{s.t. }\bm{A}_{\bm{z},c}\bm{u}= \bm{e}_1,\nonumber
\end{equation} 
and then reconstruct $\bm{x}$ as $\bm{\hat{x}} = [\bm{\hat{u}}]_\mathbb{C}$. We use the ADMM solver that is available online.\footnote{\url{https://web.stanford.edu/~boyd/papers/admm/basis_pursuit/basis_pursuit.html##4}}

\begin{figure*}[ht]
    \centering
    \includegraphics[scale = 0.6]{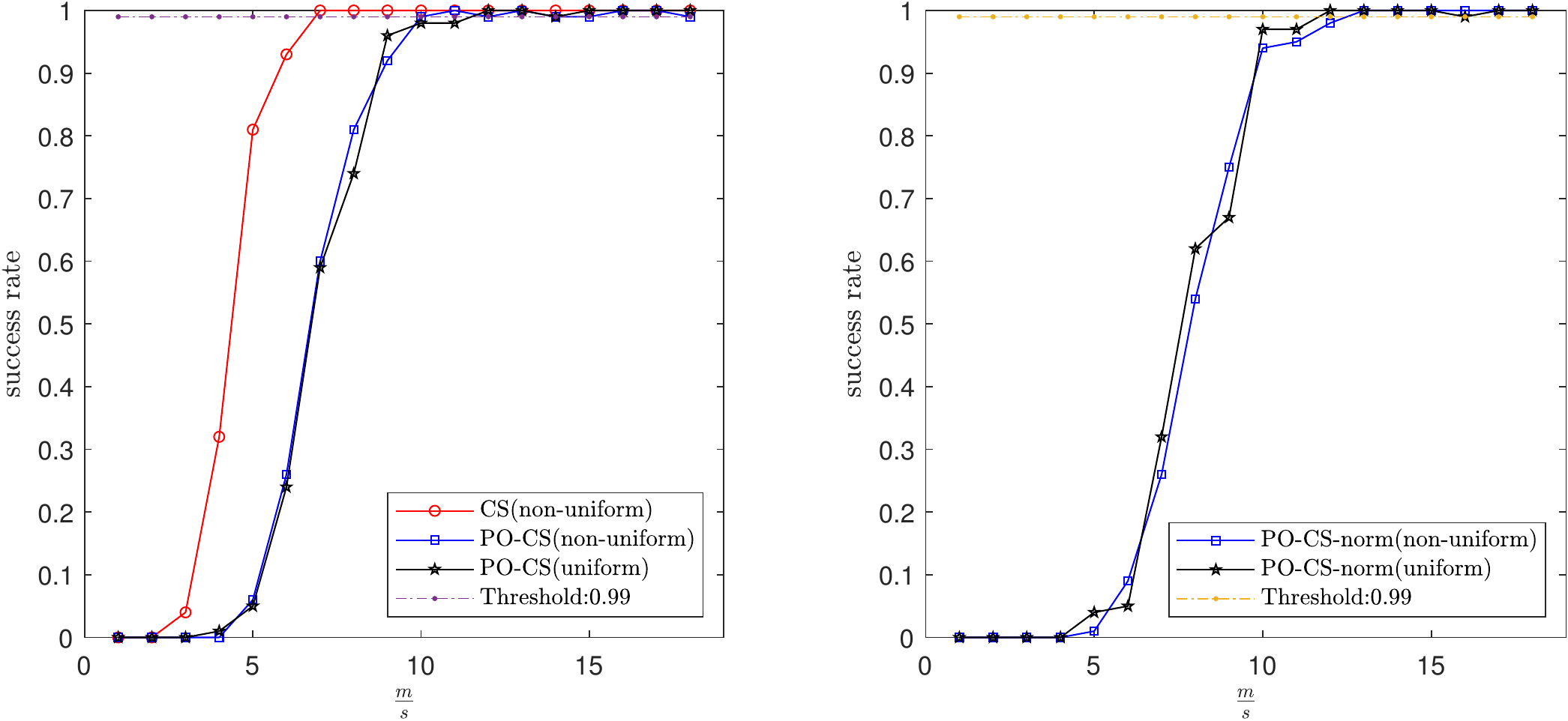}
    \caption{(left): PO-CS of direction recovery; (right): dithered  PO-CS with norm recovery.}
    \label{fig1}
\end{figure*}
The main aim of our first experiment is to confirm the uniform guarantee in Theorem \ref{theorem1}, which states that one randomly drawn matrix $\bm{\Phi}$ can simultaneously ensure the  recovery of all sparse signals. Note that it is not possible to recover  $\|\bm{x}\|$ from $\bm{z}=\sign(\bm{\Phi x})$, so we adopt an error measure given by $\| \frac{\bm{\hat{x}}}{\|\bm{\hat{x}}\|}- \bm{x}\|$ (recall that $\|\bm{x}\|=1$). A trial is claimed to be successful if $\|\frac{\bm{\hat{x}}}{\|\bm{\hat{x}}\|}- \bm{x}\|<10^{-3}$. Similar to \cite{jacques2021importance}, in our first experimental setting, both $\bm{x}$ and 
$\bm{\Phi}$ are randomly regenerated in each trial, which is evidently the experimental design for validating a non-uniform recovery guarantee.  By contrast, in our second setting, the 100 independent trials share a fixed $\bm{\Phi}$ that is drawn beforehand, while only the underlying signal is regenerated for each trial. The empirical success rates for  these two settings   are plotted in Figure \ref{fig1}(left), labeled "PO-CS(non-uniform)" and 
"PO-CS(uniform)", respectively. In addition, we are also interested in comparing PO-CS with the classical compressive sensing setting where the full measurements $\bm{z}_f = \bm{\Phi} \bm{x}$ are observed.  To be fair, we also formulate it as a real compressive sensing problem  and solve $\bm{\tilde{u}}\in \mathbb{R}^{2n}$ from 
\begin{equation}\nonumber\begin{aligned}
    &\bm{\tilde{u}} =  \mathop{\arg\min}\limits_{\bm{u}\in \mathbb{R}^{2n}} \|\bm{u}\|_{1},\\\text{s.t. }&\begin{bmatrix}
        \bm{\Phi}^\Re-\bm{\Phi}^\Im
    \end{bmatrix}\bm{u} = \bm{z}_f^\Re,\\&\begin{bmatrix}
        \bm{\Phi}^\Im&\bm{\Phi}^\Re
    \end{bmatrix}\bm{u} = \bm{z}_f^\Im,\end{aligned}
\end{equation}
which then recovers $\bm{x}$ as $\bm{\tilde{x}} = 
[\bm{\tilde{u}}]_\mathbb{C}$. For this setting we use new $\bm{\Phi}$ and $\bm{x}$ for each trial, and 
report the experimental success rate in the curve with label "CS(non-uniform)".

It is natural that success rate increases under larger sample size in 
the above three settings. Note that, the two curves for PO-CS, differentiated by whether the measurement matrix $\bm{\Phi}$ is new in each trial, are extremely close and almost coincident. This indicates that there is no need to regenerate the sensing matrix for the recovery of a new sparse signal, and hence is consistent with our   Theorem  \ref{theorem1}.
In addition, we find that a high success rate ($\geq 0.99$) for PO-CS is achieved at about twice  the measurement number needed for linear compressive sensing. Particularly, the success rate remains higher than $0.99$ when 
$\frac{m}{s}\geq 7$ for classical linear compressive sensing, or when $\frac{m}{s}\geq 12$ for PO-CS.  In fact, this phenomenon has already been experimentally concluded in \cite{jacques2021importance}, and now it is also observed in the complex case.

In our second experiment, we implement PO-CS under the additional Gaussian dither to achieve norm reconstruction. Recall that   the model now reads as $\bm{z}_d = \sign(\bm{\Phi x}+\bm{\tau}_d)$, and we use the dither $\bm{\tau}_d \sim \mathcal{N}(\bm{0},\frac{1}{9}\bm{I}_n)+\mathcal{N}(\bm{0},\frac{1}{9}\bm{I}_n)\ii$. The reconstruction procedure is given in Theorem \ref{theorem4}. Let $\bm{\hat{x}}$ be the reconstructed signal, and we claim a  trial to be successful if $\|\bm{\hat{x}}-\bm{x}\| <10^{-3}$. Recall that Theorem \ref{theorem4} delivers uniform reconstruction, that is, a single generation of the measurement ensemble $(\bm{\Phi},\bm{\tau})$ suffices to ensure the recovery of all sparse signals.  To verify this point, we similarly compare the settings with or without new sensing matrix and new dither for each trial. The experimental results  are shown in Figure \ref{fig1}(right). Evidently, the success rate in both curves are comparable, thus validating the theoretical uniform recovery.

\section{Conclusion and Future Direction}
\label{section9}
Phase-only compressive sensing (PO-CS) generalizes 1-bit CS to complex sensing matrix,  and also provides a practically appealing sensing scenario that merits advantages like stability and easier   quantization. A recent breakthrough for PO-CS     due to  Jacques and Feuillen establishes exact recovery guarantee for the direction of a fixed real signal \cite{jacques2021importance}, thus theoretically supporting    previous experimental observations   \cite{boufounos2013sparse}. Nevertheless, it remains unknown whether   uniform exact recovery for all signals of interest can be achieved, and whether a complex signal can be exactly recovered  from phase-only compressive measurements. These are among several open questions  raised in \cite{jacques2021importance}.

We almost completely address the above two open questions in this work. It was proved that  all   sparse signals in $\Sigma^n_{s,c}$ (resp. low-rank matrices in $\mathcal{M}^{n_1,n_2}_{r,c}$)  can be uniformly, exactly reconstructed up to a positive scaling,  from the near optimal $\tilde{O}(s)$ (resp. $\tilde{O}(r(d_1+d_2))$) phase-only measurements produced by a single  $\bm{\Phi}$ with i.i.d. $\mathcal{N}(0,1)+\mathcal{N}(0,1)\ii$ entries. {\em This simultaneously provides a complete affirmative answer for the two open questions, under the most classical sparse and low-rank structures}. From a technical perspective, we achieve uniform recovery by covering arguments. To handle the pathological behaviour of $\sign(a)$ when $a$ is close to 0,   a delicate analysis is carried out to control near vanishing measurements. Compared to the real case in \cite{jacques2021importance}, we establish   a different sign-product embedding property and utilize a   rescaling of the new sensing matrix for analysing PO-CS of complex signal. We discuss that, Theorem \ref{theorem1} directly generalizes to many other structured signal sets with low covering dimension, while the main downside is the lack of robustness to model error. Therefore, a non-uniform result for PO-CS of complex signal is presented to complement    our uniform guarantee. {\em This provides a complete affirmative answer for the open question of complex signal recovery in PO-CS}.

 To close this paper, we point out several open questions as future research directions. \textcolor{black}{Firstly, a uniform recovery guarantee that nicely accommodates model error is urgently desired to narrow the gap between PO-CS and linear compressive sensing. Following the RIP-based analysis in \cite{jacques2021importance} and this work, a possible first step is to pursue uniform recovery over the set of approximately sparse signals $\mathcal{K}_{q,c}$ (\ref{approsp}).} Besides, as both \cite{jacques2021importance} and this work assume a moderate bounded noise after capturing the phases, one may   study the noisy PO-CS under other noise assumptions;  for instance, i.i.d. statistical noise or noise before capturing the phases. Our last direction  is the relaxation of complex Gaussian sensing matrix, which is also raised as the third open question in \cite[Sec. VII]{jacques2021importance}. While the extension to structured random matrix was emphasized in \cite{jacques2021importance}, under complex $\bm{\Phi}$ with i.i.d. sub-Gaussian entries that has randomness closer to the current $\bm{\Phi}\sim \mathcal{N}^{m\times n}(0,1)+ \mathcal{N}^{m\times n}(0,1)\ii$, it is already open whether exact reconstruction is still   achievable. Note that in 1-bit CS, it is in general impossible to achieve meaningful recovery under sub-Gaussian $\bm{\Phi}$ \cite{ai2014one}, \textcolor{black}{while    the  sensing vectors can be sub-Gaussian or even heavy-tailed   \cite{thrampoulidis2020generalized,chen2022high,dirksen2021non} if a uniform dither is added before the 1-bit quantization. In a related setting where the measurements are quantized under a uniform quantizer, adding random dithering also allows accurate signal reconstruction under general RIP sensing matrix \cite{xu2020quantized}, sub-Gaussian \cite{thrampoulidis2020generalized,genzel2022unified,chen2022quantizing} or even heavy-tailed sensing vectors  \cite{chen2022quantizing} with possibly unknown covariance matrix. Inspired by these developments, it may   be interesting to  investigate other possible privileges of dithering (besides the full reconstruction in our Theorem \ref{theorem4}); in particular, whether substantial relaxation on the construction of $\bm{\Phi}$ is easier by using suitable dither.} 

  \bibliographystyle{plain}
\bibliography{libr}
 \begin{appendix}

\subsection{The Proof of Theorem \ref{theorem5}}
With no loss of generality, we assume the fixed underlying signal $\bm{x}$ satisfies $\|\bm{x}\|=1$. We will use the notation $[\cdot]_\mathbb{R},[\cdot]_\mathbb{C}$ introduced in (\ref{c2r}), (\ref{r2c}). We let $\mathcal{K}^* = \mathcal{K}\cap \mathbb{S}^{n-1}_{c}$, and some algebra translates (\ref{a6.5}) into 
\begin{equation}
\label{A.1}\begin{aligned}
        &\sup_{\bm{u}\in \mathcal{K}^*} \Big|\frac{1}{\kappa^2m^2}[\Re(\bm{z}^*\bm{\Phi u})]^2\\&~~~~~~~~~~~
    + \frac{2}{3m}\|\Im(\diag(\bm{z}^*)\bm{\Phi u})\|^2-1\Big|<\frac{1}{3}+\delta,\end{aligned}
\end{equation}
which can be implied by an upper bound for the parallel part 
\begin{equation}
\label{A.2}\begin{aligned}
    &\sup_{\bm{u}\in \mathcal{K}^*} f^\|(\bm{u}) \\:
    =& \sup_{\bm{u}\in \mathcal{K}^*} \Big|\frac{1}{\kappa^2m^2}[\Re(\bm{z}^*\bm{\Phi u})]^2 - [\Re\big<\bm{u},\bm{x}\big>]^2\Big| \\ =& O(\delta),\end{aligned}
\end{equation}
and the bound for the orthogonal part \begin{equation}
\label{A.3}\begin{aligned}
    &\sup_{\bm{u}\in \mathcal{K}^*}f^\bot(\bm{u}) \\:=&\sup_{\bm{u}\in \mathcal{K}^*} \left|\frac{2}{3m}\|\Im(\diag(\bm{z}^*)\bm{\Phi u})\|^2 - \|\bm{u}_{\bm{x}}^\bot\|^2\right| \\\leq &\frac{1}{3}+O(\delta). \end{aligned}
\end{equation}
Note that we use $O(\delta)$ to allow an multiplicative absolute constant; up to rescaling (\ref{A.2}) and (\ref{A.3}) jointly yield (\ref{A.1}). In the remainder of this proof, we establish (\ref{A.2}), (\ref{A.3}) in the following two subsections. For clarity, some auxiliary facts are collected in Subsection \ref{auxi}.
\subsection{The Parallel Part}
By using similar algebra as in (\ref{3.21}), (\ref{aa3.40}) we obtain (\ref{A.4}).
 \begin{figure*}[!t]
\normalsize
\begin{equation}
    \begin{aligned}  \label{A.4}
    \sup_{\bm{u}\in\mathcal{K}^*}f^\|(\bm{u})&\leq  \Big\{  \sup_{\bm{u}\in \mathcal{K}^*} \Big|\frac{1}{\kappa m}\Re\big<\sign(\bm{\Phi x}),\bm{\Phi u}\big>- \Re\big<\bm{u},\bm{x}\big>\Big|\Big\}  
    \cdot \Big\{ \sup_{\bm{u}\in \mathcal{K}^*}\frac{\|\bm{\Phi u}\|}{\sqrt{m}}+1\Big\}\\&\stackrel{(i)}{\lesssim} \sup_{\bm{u}\in \mathcal{K}^*} \Big|\frac{1}{\kappa m}\Re\big<\sign(\bm{\Phi x}),\bm{\Phi u}\big>- \Re\big<\bm{u},\bm{x}\big>\Big| | \\&\leq\sup_{\bm{u}\in \mathcal{K}^*} \underbrace{|\Re\big<\bm{u},\bm{x}\big>| \Big|\frac{\|\bm{\Phi x}\|_1}{\kappa m}-1\Big|}_{f_1^\|(\bm{u})}   + \sup_{\bm{u}\in \mathcal{K}^*} \underbrace{\frac{1}{\kappa m}\Big|\Re\big<\sign(\bm{\Phi x}), \bm{\Phi u}^\bot_{\bm{x}}\big>\Big| }_{f_2^\|(\bm{u})}.
    \end{aligned}
\end{equation}
\hrulefill
\vspace*{4pt}
\end{figure*}

\noindent{\textbf{(Step 1.) Bound $\sup_{\bm{u}\in \mathcal{K}^*}f_1^\|(\bm{u})$}}

Note that in $(i)$ we use $\sup_{\bm{u}\in \mathcal{K}^*}\frac{\|\bm{\Phi u}\|}{\sqrt{m}}=O(1)$ due to Fact \ref{fact1}. 
By  \cite[Lem. 5.2]{jacques2021importance},  $$\sup_{\bm{u}\in \mathcal{K}^*}f_1^\|(\bm{u}) \leq \Big|\frac{\|\bm{\Phi x}\|_1}{\kappa m}  - 1\Big| = O(\delta)$$ holds with probability at least $1-c_1\exp(-\Omega(\delta^2m))$.\footnote{   \cite[Lem. 5.2]{jacques2021importance} is stated for real $\bm{x}$ but the proof obviously applies to $\bm{x}\in \mathbb{C}^n$.}

\noindent{\textbf{(Step 2.) Bound $\sup_{\bm{u}\in \mathcal{K}^*}f_2^\|(\bm{u})$}}

We make use of the rotational invariance of $\bm{\Phi}$ as in the proof of Lemma \ref{lemma6}. We fix a unitary matrix $\bm{P}$ such that $\bm{Px}=\bm{e}_1$, and we let $\bm{\widetilde{\Phi}}=\bm{\Phi P^*}$. Recall that $\bm{u_x}^\bot = \bm{u}-\Re\big<\bm{u},\bm{x}\big>\bm{x}$, we let $\bm{v}=\bm{P}\bm{u_x}^\bot$. Then following the analysis in the proof of Lemma \ref{lemma6}, especially (\ref{3.23}), we know that by conditionally on the first column of $\bm{\widetilde{\Phi}}$, (almost surely)  $\sup_{\bm{u}\in \mathcal{K}^*}f_2^\|(\bm{u})$  has the same distribution   as \begin{equation}\begin{aligned}
    &\sup_{\bm{u}\in \mathcal{K}^*}\frac{1}{\kappa\sqrt{m}}|\Re\big<\bm{g}, \bm{v}^{[2:n]}\big>|,~\text{where }\\&\bm{g} \sim \mathcal{N}^{(n-1)\times 1}({0},1)+\mathcal{N}^{(n-1)\times 1}({0},1)\ii,\label{same1}\end{aligned}
\end{equation}  or equivalently, the same distribution as\begin{equation}
    \label{same2}\begin{aligned}
    &\sup_{\bm{u}\in \mathcal{K}^*}\frac{1}{\kappa\sqrt{m}}\big|\big<\bm{\hat{g}},[\bm{v}^{[2:n]}]_\mathbb{R}\big>\big|,\\&\text{where } \bm{\hat{g}}\sim \mathcal{N}^{(2n-2)\times 1}(0,1).\end{aligned}
\end{equation} Moreover,  from Fact \ref{fact2} there exists some $\bm{L}\in \mathbb{R}^{2n\times 2n}$ satisfying $\|\bm{L}\|\leq 1$ such that $[\bm{v}] _\mathbb{R}=\bm{L}[\bm{u}]_\mathbb{R}$. More evidently, there exists $\bm{L}_1\in \mathbb{R}^{(2n-2)\times 2n}$ satisfying $\|\bm{L}_1\|\leq 1$ such that $[\bm{v}^{[2:n]}]_\mathbb{R}=\bm{L}_1[\bm{v}]_\mathbb{R}$. Overall, we have $[\bm{v}^{[2:n]}]_\mathbb{R}=\bm{L}_1\bm{L}[\bm{u}]_\mathbb{R}$, hence 
$\bm{u}\in \mathcal{K}^*$ implies $[\bm{v}^{[2:n]}]_\mathbb{R}\in \bm{L}_1\bm{L}[\mathcal{K}^*]_\mathbb{R}$.  Assuming $\bm{\hat{g}}\sim \mathcal{N}(\bm{0},\bm{I}_{2n-2})$ we start from (\ref{same2}) and proceed as follows:  
\begin{equation}
   \begin{aligned}\nonumber
       \mathbbm{E} \sup_{\bm{u}\in \mathcal{K}^*} f_1^\|(\bm{u})& = \frac{1}{\kappa\sqrt{m}}\mathbbm{E}   \sup_{\bm{u}\in \mathcal{K}^*} \Big|\big<\bm{\hat{g}},[\bm{v}^{[2:n]}]_\mathbb{R}\big>\Big|\\&\stackrel{(i)}{=} \frac{1}{\kappa\sqrt{m}} \mathbbm{E} \sup_{\bm{u}\in \mathcal{K}^*}\big<\bm{\hat{g}},[\bm{v}^{[2:n]}]_\mathbb{R}\big>\\
      &\stackrel{(ii)}{\leq} \frac{1}{\kappa\sqrt{m}}\omega\big(\bm{L}_1\bm{L}[\mathcal{K}^*]_\mathbb{R}\big)\\&\stackrel{(iii)}{\leq} \frac{1}{\kappa\sqrt{m}}\omega([\mathcal{K}^*]_\mathbb{R}){=} \frac{1}{\kappa\sqrt{m}}\omega(\mathcal{K}^*)\stackrel{{(iv)}}{<}\delta,
   \end{aligned}
\end{equation}
where $(i)$ is because $\mathcal{K}$ is a symmetric cone, $(ii)$ is from $[\bm{v}^{[2:n]}]_\mathbb{R}\in \bm{L}_0[\mathcal{K}^*]_\mathbb{R}$, in $(iii)$ we use \cite[Exercise 7.5.4]{vershynin2018high} and $\|\bm{L}_1\bm{L}\|\leq \|\bm{L}_1\|\|\bm{L}\|\leq 1$, $(iv)$ guaranteed by the sample complexity (\ref{6.4}).

Moreover, we define $F(\bm{g}):= \sup_{\bm{u}\in \mathcal{K}^*}\frac{1}{\kappa\sqrt{m}}|\Re\big<\bm{g},\bm{v}^{[2:n]}\big>|$ and view it as a function of the Gaussian variable $\bm{g}$. It is $\big(\frac{1}{\kappa\sqrt{m}}\big)$-Lipschitz since
\begin{equation}\nonumber\begin{aligned}
    |{F}(\bm{g}_1) -{F}(\bm{g}_2)|& \leq \sup_{\bm{u}\in \mathcal{K}^*} \frac{1}{\kappa\sqrt{m}} |\Re\big<\bm{g}_1-\bm{g}_2, \bm{v}^{[2:n]}\big>| \\&\leq \frac{1}{\kappa\sqrt{m}}\|\bm{g}_1-\bm{g}_2\|.\end{aligned}
\end{equation}
Hence, combine these pieces and apply    \cite[Lem. 5.1]{jacques2021importance}  we obtain \begin{equation}
    \begin{aligned}\nonumber
       \mathbbm{P}(F(\bm{g}) \geq 2\delta) &\leq \mathbbm{P}(F(\bm{g}) - \mathbbm{E}F(\bm{g})>\delta)\\& \leq 2\exp\Big(-\frac{1}{2}\kappa^2\delta^2m\Big).
    \end{aligned}
\end{equation}   
Further use the relation between the distribution of $F(\bm{g})$ and $\sup_{\bm{u}\in \mathcal{K}^*}f_2^\|(\bm{u})$ (see (\ref{same1})),  we obtain  $\sup_{\bm{u}\in \mathcal{K}^*}f_2^\|(\bm{u})\leq 2\delta$ with probability at least $1-2\exp(-\Omega(\delta^2m))$. Overall, we arrive at $\sup_{\bm{u}\in\mathcal{K}^*}f^\|(\bm{u}) = O(\delta)$ with high probability.

We have indeed shown the following local sign-product embedding property in (\ref{A.4}), which generalizes  \cite[Lem. 5.4]{jacques2021importance} to the  complex case and may be of independent interest. Note that even restricted to $\mathbb{R}^n$, Corollary \ref{local} is not fully coincident with   \cite[Lem. 5.4]{jacques2021importance} since $\big<\sign(\bm{\Phi u}),\bm{\Phi v}\big>$ is complex in general.

\begin{coro}
\label{local}
{\rm (Local Sign-Product Embedding Property)\textbf{.}} Assume $\bm{\Phi}\sim \mathcal{N}^{m\times n}(0,1)+   \mathcal{N}^{m\times n}(0,1)\ii$. Given a symmetric cone $\mathcal{K}\subset \mathbb{C}^n$  and a fixed $\bm{u}\in \mathbb{S}^{n-1}_c$, for any $\delta>0$, if $m\geq \frac{C}{\delta^2}\omega^2(\mathcal{K}\cap \mathbb{S}^{n-1}_c)$, then with probability at least $1-c_1\exp(-c_2\delta^2m)$ we have \begin{equation}\nonumber
    \Big|\frac{1}{\kappa m}\Re\big<\sign(\bm{\Phi u}),\bm{\Phi v}\big> - \Re\big<\bm{u},\bm{v}\big>\Big| \leq \delta\|\bm{v}\|,~\forall~\bm{v}\in \mathcal{K}.
\end{equation}  
\end{coro}

\subsection{The Orthogonal Part}
It remains to show (\ref{A.3}). Recall that  $\|\Im(\diag(\bm{z}^*)\bm{\Phi u})\|^2$   concentrates around $\|\bm{u}_{\bm{x}}^\bot\|^2 + |\Im\big<\bm{x},\bm{u}\big>|^2$ with a bias term $|\Im\big<\bm{x},\bm{u}\big>|^2$ (Lemma \ref{lem8}). Thus, we introduce 
\begin{equation}
     \label{A.9}\begin{aligned}
         &\hat{f}^\bot( \bm{u}) : = \Big|\frac{1}{m}\|\Im\big(\diag(\bm{z}^*)\bm{\Phi u}\big) \|^2 \\&~~~~~~~~~~~~~~~~~~~~~- (\|\bm{u}_{\bm{x}}^\bot\|^2 + |\Im\big<\bm{x},\bm{u}\big>|^2)\Big|.\end{aligned}
     \end{equation}
     Similar to (\ref{aa3.66}), we have\begin{equation}
         \begin{aligned}\nonumber
            \sup_{\bm{u}\in\mathcal{K}^*}f^\bot(\bm{u}) \leq \frac{1}{3}+\frac{2}{3}\sup_{\bm{u}\in \mathcal{K}^*} \hat{f}^\bot(\bm{u}),
         \end{aligned}
     \end{equation}
     thus it is sufficient to show 
     \begin{equation}
         \label{suffi}
         \sup_{\bm{u}\in \mathcal{K}^*}\hat{f}^\bot(\bm{u}) = O(\delta).
     \end{equation}
     By (\ref{3.11}), (\ref{3.31}), we can write \begin{equation}\nonumber\begin{aligned}
         &\hat{f}^\bot(\bm{u}) = \Big|\frac{1}{m}\sum_{k=1}^m \big[\Im(\overline{\sign(\bm{\Phi}_k^*\bm{x})}\bm{\Phi}_k^*\bm{u}^\bot_{\bm{x}})\big]^2\\&~~~~~~~~~~~~~~~~~~~ ~~- \mathbbm{E}\big[\Im(\overline{\sign(\bm{\Phi}_k^*\bm{x})}\bm{\Phi}_k^*\bm{u}^\bot_{\bm{x}})\big]^2\Big|\end{aligned}
     \end{equation}
    as the mean of independent, zero-mean random variable. To prove (\ref{suffi}), our strategy is to first identify the distribution of $\Im(\overline{\sign(\bm{\Phi}_k^* \bm{x})}\bm{\Phi}^*_k\bm{u_x^\bot})$ and then apply the concentration result developed in \cite{dirksen2016dimensionality}.

    \noindent{\textbf{(Step 1.) Identify the distribution of $\Im(\overline{\sign(\bm{\Phi}_k^* \bm{x})}\bm{\Phi}^*_k\bm{u_x^\bot})$}}

    We find a fixed unitary matrix $\bm{P}$ such that $\bm{Px}  = \bm{e}_1$, then we let $\bm{v} = \bm{Pu_x}^\bot = [v_i]$. Because $\bm{v}=\bm{Pu_x}^\bot=\bm{Pu}-\Re\big<\bm{Pu},\bm{Px}\big>\bm{Px}$, $\bm{v}$ satisfies $\Re (v_1)=0$. For the sensing vectors, we also let $\bm{\tilde{\Phi}}_k = \bm{P\Phi}_k = [\phi_{ki}]_{i\in [n]}$. Then by rotational invariance,  $\{\bm{\tilde{\Phi}}_k:k\in [m]\}$ are independent copies of $\mathcal{N}^{n\times 1}(0,1) +\mathcal{N}^{n\times 1}(0,1)\ii$. Moreover, some algebra delivers 
    \begin{equation}
        \begin{aligned}\label{A.15}
           &\Im\big(\overline{\sign(\bm{\Phi}_k^* \bm{x})}\bm{\Phi}^*_k\bm{u_x^\bot}\big) = \Im\Big(\overline{\sign(\bm{\tilde{\Phi}}_k^*\bm{e}_1)}\bm{\tilde{\Phi}}_k^*\bm{v}\Big)\\&=\Im\Big(\sign(\phi_{k1}) \cdot \sum_{i=1}^n\overline{\phi_{ki}}v_i \Big) \\
           & = \frac{|\phi_{k1}|}{\sqrt{2}}\cdot\big(\sqrt{2}v_1^\Im\big) + \sum_{i=2}^n\Re(\sign(\phi_{k1})\overline{\phi_{ki}})\cdot v_i^\Im \\&~~~~~~~~~~+\sum_{i=2}^n  \Im(\sign(\phi_{k1})\overline{\phi_{ki}})\cdot v_i^\Re \\&= \big<\bm{\Psi}_k, \bm{\hat{v}}\big>,
        \end{aligned}
    \end{equation}
    where 
    \begin{equation}\label{hatv}
        \bm{\hat{v}} = [\sqrt{2}v_1^\Im,v_2^\Im,\cdots,v_n^\Im,v_2^\Re,\cdots,v_n^\Re]^\top\in \mathbb{R}^{2n-1},
    \end{equation}
    $\bm{\Psi}_k$ is  given by  \begin{equation}
        \begin{aligned}
        \label{A.13}
          & \bm{\Psi}_k =  \Big[\frac{|\phi_{k1}|}{\sqrt{2}} , \Re(\sign(\phi_{k1})\overline{\phi_{k2}}),\cdots,\Re(\sign(\phi_{k1})\overline{\phi_{kn}}),\\&\Im(\sign(\phi_{k1})\overline{\phi_{k2}}),\cdots,\Im(\sign(\phi_{k1})\overline{\phi_{kn}})\Big]^\top \in \mathbb{R}^{2n-1}.
        \end{aligned}
    \end{equation} 
    Note that $|\phi_{k1}|$ and $\sign(\phi_{k1})$ are statistically independent (e.g., \cite[Exercise 3.3.7]{vershynin2018high}), thus by conditionally on $\sign(\phi_{k1})$, $\{\sign(\phi_{k1})\overline{\phi_{ki}}:2\leq i\leq n\}$ are (almost surely) i.i.d. copies of $\mathcal{N}(0,1)+\mathcal{N}(0,1)\ii$. Consequently, with an additional entry of $\frac{|\phi_{k1}|}{\sqrt{2}}$,  $\bm{\Psi}_k$ is isotropic random vector satisfying $\|\bm{\Psi}_k\|_{\psi_2}=O(1)$. 
    
    \noindent{\textbf{(Step 2.) Establish the concentration via a tool from \cite{dirksen2016dimensionality}}}
    
    By (\ref{suffi}) and (\ref{A.15}), we only need to show \begin{equation}
    \label{A.14}
        \sup_{\bm{u}\in \mathcal{K}^*}\Big|\frac{1}{m}\sum_{k=1}^m\big<\bm{\Psi}_k,\bm{\hat{v}}\big>^2 - \|\bm{\hat{v}}\|^2\Big| \leq C_1\delta  ,  
    \end{equation} 
    and recall that $\bm{\Psi}_k$ is given in (\ref{A.13}), $\bm{\hat{v}}$ defined in (\ref{hatv}) is constructed from   $\bm{v} = \bm{P}(\bm{u}-\Re\big<\bm{u},\bm{x}\big>\bm{x})$. By construction we have  $\bm{v}=\bm{P}\big(\bm{u}-\Re\big<\bm{u},\bm{x}\big>\bm{x}\big)$, then we use Fact \ref{fact2} to obtain that there exists some $\bm{L}\in \mathbb{R}^{2n\times 2n}$ satisfying $\|\bm{L}\|\leq 1$ such that $[\bm{v}]_\mathbb{R}=\bm{L}[\bm{u}]_\mathbb{R}$. More evidently, from (\ref{hatv}) we know there exists some $\bm{L}_2\in \mathbb{R}^{(2n-1)\times 2n}$ satisfying $\|\bm{L}_2\|=\sqrt{2}$ such that $\bm{\hat{v}}=\bm{L}_2[\bm{v}]_\mathbb{R}$. Overall, we have $\bm{\hat{v}}=\bm{L}_2\bm{L}[\bm{u}]_\mathbb{R}$, thus
        $\bm{u}\in \mathcal{K}^*$ implies $\bm{\hat{v}} \in \bm{L}_2\bm{L}[\mathcal{K}^*]_\mathbb{R}$. Therefore, by letting $\mathcal{K}_1 :=\bm{L}_2\bm{L}  [\mathcal{K}^*]_\mathbb{R}$,  the desired   (\ref{A.14}) can be implied by  \begin{equation}
    \label{A.16}
        \sup_{\bm{w}\in \mathcal{K}_1} \Big|\frac{1}{m}\sum_{k=1}^m\big(\big<\bm{\Psi}_k,\bm{w}\big>^2- \mathbbm{E}\big<\bm{\Psi}_k,\bm{w}\big>^2\big)\Big|\leq C_1\delta.
    \end{equation}
    We prove this by using  \cite[Thm. 3.2]{dirksen2016dimensionality}. Following the notations in \cite{dirksen2016dimensionality} we first verify the  Assumptions of \cite[Thm. 3.2]{dirksen2016dimensionality}.

    Firstly, as we have $\max_{k}\|\bm{\Psi}_k\|_{\psi_2} \leq C_2$ for some absolute constant $C_2$, for any $\bm{w}_1,\bm{w}_2\in \mathcal{K}_1$ it holds that
    \begin{equation}
        \begin{aligned}
        \label{A.17}
            d_{\psi_2}(\bm{w}_1,\bm{w}_2)&= \max_{k\in [m]} \big\|\big<\bm{\Psi}_k,\bm{w}_1\big> -\big<\bm{\Psi}_k,\bm{w}_2\big> \big\|_{\psi_2} \\&\leq C_2\|\bm{w}_1-\bm{w}_2\|\\:&=C_2d_2(\bm{w}_1,\bm{w}_2),
        \end{aligned}
    \end{equation}
    where we write $d_2(\bm{w}_1,\bm{w}_2)=\|\bm{w}_1-\bm{w}_2\|$ as the $\ell_2$ distance. 
     Assume $\bm{g}_1\sim \mathcal{N}^{(2n-1)\times 1}({0},1)$, we have\begin{equation}\nonumber\begin{aligned}
        \gamma_2(\mathcal{K}_1,d_{\psi_2})&\stackrel{(i)}{\leq } C_2\gamma_2(\mathcal{K}_1,d_{\ell_2}) \\&\stackrel{(ii)}{\lesssim} \mathbbm{E}\Big|\sup_{\bm{w}\in \mathcal{K}_1}\big<\bm{g}_1, \bm{w}\big>\Big| \stackrel{(iii)}{=}\omega(\mathcal{K}_1),\end{aligned}
    \end{equation} 
    where   $(i)$ is due to    definition of the $\gamma_2$-functional (see  \cite[Definition 3.1]{dirksen2016dimensionality}) and (\ref{A.17}), $(ii)$ is due to   \cite[Equation (5)]{dirksen2016dimensionality},   $(iii)$ is because $\mathcal{K}_1$ is symmetric.

    Secondly, we have \begin{equation}
        \begin{aligned}\nonumber
       & \bar{\Delta}_{\psi_2}(\mathcal{K}_1) = \sup_{\bm{w}\in \mathcal{K}_1}\max_{k\in [m]} \|\big<\bm{\Psi}_k,\bm{w}\big>\|_{\psi_2}\\&\leq C_2 \sup_{\bm{w}\in \mathcal{K}_1}\|\bm{w}\| \stackrel{(i)}{\leq} C_2\|\bm{L}_2\|\|\bm{L}\| \stackrel{(ii)}{\leq} C_2\sqrt{2},
        \end{aligned}
    \end{equation} 
    where we use $\mathcal{K}_1=\bm{L}_2\bm{L}[\mathcal{K}^*]_\mathbb{R}$ in $(i)$, and then $\|\bm{L}\|\leq1$, $\|\bm{L}_2\|\leq \sqrt{2}$ in $(ii)$.

    Therefore, we can invoke     \cite[Thm. 3.2]{dirksen2016dimensionality} to obtain \begin{equation}
        \begin{aligned}\nonumber
        &\mathbbm{P}\Big(\sup_{\bm{w}\in \mathcal{K}_1}\Big|\frac{1}{m}\sum_{k=1}^m\big(\big<\bm{\Psi}_k,\bm{w}\big>^2- \mathbbm{E}\big<\bm{\Psi}_k,\bm{w}\big>^2\big)\Big|\geq \\&~~~C_3\Big[\frac{\omega^2(\mathcal{K}_1)}{m}  + \frac{\omega (\mathcal{K}_1)}{\sqrt{m}} + \frac{t}{m} + \frac{\sqrt{t}}{\sqrt{m}}\Big]\Big) \leq \exp(-t)  
        \end{aligned}
    \end{equation}
    for any $t>0$.
    Thus, we take $t = \delta^2m$ and assume a  sample size $m = \Omega(\delta^{-2}\omega^2(\mathcal{K}_1))$, (\ref{A.16}) follows with probability at least $1-\exp(-C_4\delta^2m)$.

    Now it remains to confirm $\omega^2(\mathcal{K}_1)\lesssim \omega^2(\mathcal{K}\cap \mathbb{S}^{n-1}_c)$. Since $\mathcal{K}_1 = \bm{L}_2\bm{L}[\mathcal{K}^*]_\mathbb{R}$ for some $\|\bm{L}\|\leq 1,\|\bm{L}_2\|=\sqrt{2}$, by using    \cite[Exercise 7.5.4]{vershynin2018high} we obtain $\omega (\mathcal{K}_1) \leq \sqrt{2} \omega([\mathcal{K}^*]_\mathbb{R}) = \omega(\mathcal{K}\cap \mathbb{S}^{n-1}_c)$. The proof is concluded. \hfill $\square$

\subsection{Auxiliary Facts}\label{auxi}
The following two facts would be used in the proof of Theorem \ref{theorem5}.

\begin{fact}\label{fact1}
Under the setting of Theorem \ref{theorem5},   recall that $\mathcal{K}^*=\mathcal{K}\cap \mathbb{S}^{n-1}_c$, for any $\delta>0$,  then with probability at least $1-C\exp(-\Omega(\delta^2 m ))$,  $1-\delta\leq  \frac{\|\bm{\Phi u}\|^2}{{2m}}\leq 1+\delta$ holds for all $\bm{u}\in \mathcal{K}^*$. 
\end{fact}
\begin{proof}
The proof is based on the real case in \cite[Thm. 2.2]{jacques2021importance} (that is adapted from \cite{mendelson2008uniform}). We calculate that $\|\bm{\Phi u}\|^2=\|\bm{\Phi}_1[\bm{u}]_\mathbb{R}\|^2+ \|\bm{\Phi}_2[\bm{u}]_\mathbb{R}\|^2$ where $\bm{\Phi}_1=[\bm{\Phi}^{\Re},-\bm{\Phi}^{\Im}]$, $\bm{\Phi}_2= [\bm{\Phi}^{\Im},\bm{\Phi}^{\Re}]$. Recall that $[\mathcal{K}]_\mathbb{R}=\{[\bm{u}]_\mathbb{R}:\bm{u}\in \mathcal{K}\}$, and so $\bm{u}\in \mathcal{K}^*$ is equivalent to $[\bm{u}]_\mathbb{R}\in [\mathcal{K}]_\mathbb{R}\cap \mathbb{S}^{2n-1}_r$. 
Because $\bm{\Phi}_1,\bm{\Phi}_2\sim \mathcal{N}^{m\times 2n}(0,1)$, and in Theorem \ref{theorem5} we assume $m\gtrsim \delta^{-2}\omega(\mathcal{K}^*)=\delta^{-2}\omega\big([\mathcal{K}]_\mathbb{R}\cap \mathbb{S}_r^{2n-1}\big)$, by applying \cite[Thm. 2.2]{jacques2021importance} we obtain $$ \frac{\|\bm{\Phi}_1[\bm{u}]_\mathbb{R}\|^2}{m},\frac{\|\bm{\Phi}_2[\bm{u}]_\mathbb{R}\|^2}{m}\in [1-\delta, 1+\delta]$$ holds for all $\bm{u}\in \mathcal{K}^*$ with the promised probability. The result immediately follows from $\frac{\|\bm{\Phi u}\|^2}{2m}=\frac{1}{2}\big(\frac{\|\bm{\Phi}_1[\bm{u}]_\mathbb{R}\|^2}{m}+\frac{\|\bm{\Phi}_2[\bm{u}]_\mathbb{R}\|^2}{m}\big)$.
\end{proof}
\begin{fact}
\label{fact2}
Given a fixed $\bm{x}\in \mathbb{S}^{n-1}_c$ and a unitary matrix $\bm{P}\in \mathbb{C}^{n\times n}$, for any $\bm{u}\in \mathbb{C}^n$ we define $\bm{v}=\bm{P}\big(\bm{u}-\Re\big<\bm{u},\bm{x}\big>\bm{x}\big)\in \mathbb{C}^n$. Then there exists a matrix $\bm{L}\in \mathbb{R}^{2n\times 2n}$ satisfying $\|\bm{L}\|\leq 1$ such that $[\bm{v}]_\mathbb{R}=\bm{L}[\bm{u}]_\mathbb{R}$ holds for all $\bm{u}\in \mathbb{C}^n$. 
\end{fact}
\begin{proof}
Given $\bm{x},\bm{P}$, we write $\bm{v}=\mathscr{L}(\bm{u})=\bm{P}\big(\bm{u}-\Re\big<\bm{u},\bm{x}\big>\bm{x}\big)$. Because for any $\bm{u}_1,\bm{u}_2\in\mathbb{C}^n$ and $t\in \mathbb{R}$ it holds that $\mathscr{L}(\bm{u}_1+t\bm{u}_2)=\mathscr{L}(\bm{u}_1)+t\mathscr{L}(\bm{u}_2)$,   there exists some $\bm{L}\in \mathbb{R}^{2n\times 2n}$ such that for any $\bm{u}\in \mathbb{C}^n$ we have the following: 
  \begin{equation}\label{uvrelation}\nonumber
      \bm{v}=\bm{P}\big(\bm{u}-\Re\big<\bm{u},\bm{x}\big>\bm{x}\big) \iff [\bm{v}]_\mathbb{R} = \bm{L}[\bm{u}]_\mathbb{R}.
  \end{equation}
  It remains to show $\|\bm{L}\|\leq 1$.
 Note that $[\bm{u}]_\mathbb{R}\in \mathbb{S}^{2n-1}_r$ if and only if $\bm{u}\in \mathbb{S}^{n-1}_c$, and for any $\bm{w}\in \mathbb{S}_r^{2n-1}$ there exists some $\bm{u}\in \mathbb{S}^{n-1}_c$ such that $\bm{w}=[\bm{u}]_\mathbb{R}$. Thus, assuming $\bm{v}$ is determined by $\bm{u}$ as in (\ref{uvrelation}), 
 we have \begin{equation}
     \begin{aligned}\nonumber
     \|\bm{L}\|&=\sup_{\bm{u}\in \mathbb{S}^{n-1}_c}\|[\bm{v}]_\mathbb{R}\|=\sup_{\bm{u}\in \mathbb{S}^{n-1}_c}\|\bm{v}\|\\& = \sup_{\bm{u}\in \mathbb{S}^{n-1}_c} \big\|\bm{P}\big(\bm{u}-\Re\big<\bm{u},\bm{x}\big>\bm{x}\big)\big\|\\
     &\stackrel{(i)}{=}\sup_{\bm{u}\in \mathbb{S}^{n-1}_c}\big\|\bm{u}-\Re\big<\bm{u},\bm{x}\big>\bm{x}\big\| =\sup_{\bm{u}\in \mathbb{S}_c^{n-1}}\big\|\bm{u_x}^\bot\big\|\stackrel{(ii)}{\leq}1.
     \end{aligned}
 \end{equation}
 Note that $(i)$ is because $\bm{P}$ is unitary, $(ii)$ is due to $\|\bm{u_x}^\bot\|^2+\|\bm{u_x}^\|\|^2=1$. The proof is complete. 
\end{proof}


  \end{appendix}
%





\end{document}